\documentclass[twoside,11pt]{article}

\usepackage{jmlr2e}

\usepackage[utf8]{inputenc} 
\usepackage[T1]{fontenc}    
\usepackage{graphicx}
\usepackage{subfigure}

\usepackage{amsmath}
\usepackage{amssymb}
\usepackage{dsfont}
\usepackage{mathrsfs}
\usepackage{hyperref}
\usepackage{braket}
\usepackage{bbold}
\usepackage{cancel}
\usepackage{units}
\usepackage{array}
\usepackage{enumitem}
\usepackage{url}            
\usepackage{booktabs}       
\usepackage{amsfonts}       
\usepackage{nicefrac}       
\usepackage{microtype}   
\usepackage{hyperref}
\usepackage[ruled,vlined]{algorithm2e}
\usepackage{multirow}
\usepackage{multicol}

\usepackage{xcolor}

\DeclareMathOperator*{\argmin}{arg\,min}

\newcommand\independent{\protect\mathpalette{\protect\independenT}{\perp}}
\def\independenT#1#2{\mathrel{\rlap{$#1#2$}\mkern2mu{#1#2}}}

\newcommand{\Db}{\boldsymbol{D}}

\newcommand{\Gb}{\boldsymbol{G}}
\newcommand{\Hb}{\boldsymbol{H}}

\newcommand{\Mb}{\boldsymbol{M}}

\newcommand{\Sb}{\boldsymbol{S}}

\newcommand{\Ub}{\boldsymbol{U}}

\newcommand{\Wb}{\boldsymbol{W}}
\newcommand{\Xb}{\boldsymbol{X}}

\newcommand{\Zb}{\boldsymbol{Z}}

\newcommand{\gb}{\boldsymbol{g}}

\newcommand{\ub}{\boldsymbol{u}}

\newcommand{\xb}{\boldsymbol{x}}

\newcommand{\epsb}{\boldsymbol{\epsilon}}

\firstpageno{1}

\usepackage{lastpage}
\jmlrheading{23}{2022}{1-\pageref{LastPage}}{6/21; Revised
7/22}{11/22}{21-0656}{Jaime Roquero Gimenez and Dominik Rothenh\"ausler}
\ShortHeadings{Causal Aggregation}{Roquero Gimenez and Rothenh\"ausler}

\begin{document}

\title{Causal Aggregation: Estimation and Inference of Causal Effects by Constraint-Based Data Fusion}

\author{\name Jaime Roquero Gimenez \email roquero@stanford.edu \\
       \addr Department of Statistics\\
       Stanford University\\
       Stanford, CA 94305, USA
       \AND
       \name Dominik Rothenh\"ausler \email rdominik@stanford.edu \\
       \addr Department of Statistics\\
       Stanford University\\
       Stanford, CA 94305, USA}

\editor{David Jensen}

\maketitle

\begin{abstract}%
In causal inference, it is common to estimate the causal effect of a single treatment variable on an outcome. However, practitioners may also be interested in the effect of simultaneous interventions on multiple covariates of a fixed target variable. We propose a novel method that allows to estimate the effect of joint interventions using data from different experiments in which only very few variables are manipulated. If there is only little randomized data or no randomized data at all, one can use observational data sets if certain parental sets are known or instrumental variables are available. 
If the joint causal effect is linear, the proposed method can be used for estimation and inference of joint causal effects, and we characterize conditions for identifiability. In the overidentified case, we indicate how to leverage all the available causal information across multiple data sets to efficiently estimate the causal effects. If the dimension of the covariate vector is large, we may only have a few samples in each data set. Under a sparsity assumption, we derive an estimator of the causal effects in this high-dimensional scenario. In addition, we show how to deal with the case where a lack of experimental constraints prevents direct estimation of the causal effects. When the joint causal effects are non-linear, we characterize conditions under which identifiability holds, and propose a non-linear causal aggregation methodology for experimental data sets similar to the gradient boosting algorithm where in each iteration we combine weak learners trained on different datasets using only unconfounded samples. We demonstrate the effectiveness of the proposed method on simulated and semi-synthetic data.
\end{abstract}


\section{Introduction}\label{section:introduction}
Causal inference is a centerpiece of scientific research, with applications ranging from the social sciences to biology. Often, the goal in causal inference is to estimate the effect of one single variable on one outcome: randomizing that variable and evaluating its effect on the outcome is one way to do so. Randomization provides a gold standard procedure for identifying causal effects related to that variable, as the intervention removes any spurious association with the response due to unmeasured factors. Sometimes it is also of interest to estimate the effect of joint interventions on multiple variables on an outcome. Ideally, if all the variables are jointly randomized, one can precisely reconstruct the global causal mechanism, including potential interactions between covariates. In practice, however, we only have access to multiple individual experiments---also called environments---where just a few variables are simultaneously manipulated, which only provide partial information about such mechanism. We formulate a procedure for aggregating the knowledge obtained from several experiments that reconstructs a complex global causal model, capturing the causal effect of multiple covariates on the response as if they were all simultaneously manipulated. Our framework for aggregating causal information also allows the use of Instrumental Variables (IV) and covariate adjustment as building blocks. The approach is motivated by the following: during the last decade, internet companies have massively adopted a new experimental framework to improve their web-based products: the WebLab. Any website has a myriad of design choices built in it that affect the customer behavior, such as the ranking of articles in a newsfeed, the location of an ad within a webpage, etc. These companies have the possibility of running randomized A/B experiments where customers are redirected to slightly different versions of the website. Actionable insights may be obtained by evaluating downstream metrics that reflect the effect of a particular change in the website. Estimating the combined effect of several changes would require simultaneous randomization of many parameters, which may not be feasible in practice as the user experience would vary too much. Therefore it may be of interest to understand how insights from different individual experiments, each manipulating a small number of parameters, may be aggregated into a single causal model. The following examples illustrate the purpose of our method.

\begin{example}\label{example:LinearSEM}
Assume that variables are related via the structural causal model \citep{wright1921correlation,bollen1989structural} presented in Figure~\ref{fig:Ex1-1}.
\begin{figure*}[ht]
\centering
\includegraphics[trim={1cm 7.5cm 1cm 7.5cm},clip,width=.8\linewidth]{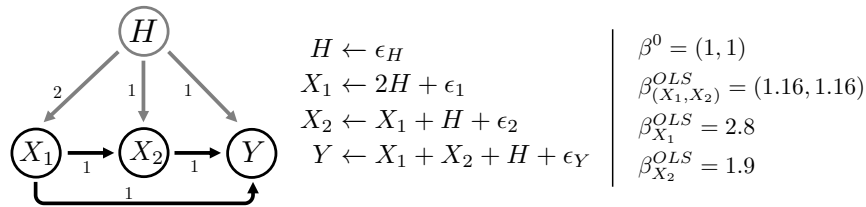}
\caption{Simple Linear Structural Equation Model: observational setting. Graphical representation of the linear SEM, along with the corresponding set of equations and regression coefficients. The disturbance variables $\epsilon_i$ are jointly independent standard Gaussian. }
\label{fig:Ex1-1}
\end{figure*}
\begin{figure*}[ht]
\centering
\includegraphics[trim={1cm 4.2cm 1cm 4.2cm},clip,width=.8\linewidth]{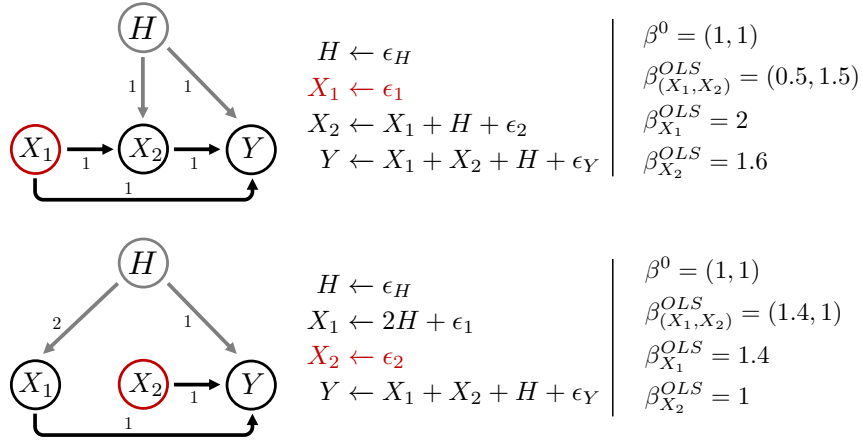}
\caption{Experimental environments.  Top: environment where $X_1$ is randomized. Graphically, the dependency on $H$ is removed. Only the structural equation defining $X_1$ is modified, and the regression coefficients are still biased. Bottom: environment where $X_2$ is randomized. In this case, regressing on $X_2$ leads to the direct causal effect. In each case, the disturbance variables $\epsilon_i$ are jointly independent standard Gaussian.}
\label{fig:Ex1-2}
\end{figure*}
The disturbance variables $\epsilon_i$ are jointly independent standard Gaussian. We observe samples from the covariates $X_1, X_2$ and the response $Y$.  $H$ is an unobserved variable that jointly affects $X_1,X_2$ and $Y$. We assume we do not know such structure, and we want to identify the vector $\beta^0 = (\beta^0_1, \beta_2^0)$ that defines the linear structural equation of $Y$ with respect to $(X_1,X_2)$. This vector defines the causal effect on $Y$ of a joint intervention that acts on both $X_1$ and $X_2$ \citep{pearl2009causal}. Estimating $\beta^0$ from observational samples alone is not possible. A quick computation shows that the coefficient $\beta^{OLS}_{(X_1,X_2)}$ from ordinary least squares (OLS) regression of $Y$ on $\{X_1, X_2\}$ based on observational data is biased (cf. Figure~\ref{fig:Ex1-1}). Individually regressing the response on any of the covariates $X_1$ or $X_2$ (that we denote $\beta^{OLS}_{X_1},\beta^{OLS}_{X_2}$ respectively) also produces biased results. This is expected, as the latent variable $H$ is simultaneously affecting the covariates and the response. $\beta^0$ can be estimated by OLS in an environment where we fully randomize all covariates. However, suppose we can at most randomize one variable at a time, leading to two different environments. We represent those environments in Figure~\ref{fig:Ex1-2}, along with the regression coefficients. Graphically, randomization removes the effect of the latent variable on the randomized covariate, but unfortunately the regression coefficients $\beta^{OLS}_{(X_1,X_2)}$ are still biased in both environments. Regressing $Y$ on individual covariates does not necessarily lead to the corresponding coordinate of $\beta^0$ either, even when those covariates are the ones randomized. The ``total causal effect'' of $X_1$ on $Y$ \citep{pearl2009causal}, evaluated by regressing $Y$ only on $X_1$ in the environment where $X_1$ is randomized, is equal to $2$ because it takes into account the effect of $X_1$ on $Y$ mediated by $X_2$. This situation does not occur in the environment where $X_2$ is randomized, where $\beta^{OLS}_{X_2} = \beta^0_2$. However we assume we do not know the graphical structure of the causal model so we do not know a priori whether the effect of $X_1$ on $Y$ is mediated by $X_2$. Without such knowledge, can we aggregate information from experimentation on individual covariates in different environments? Based on partial experimentation on covariates and other types of causal information, we will discuss how to reconstruct $\beta^0$---the ``causal'' coefficient vector, as if we were simultaneously randomizing all covariates.

\end{example} 

The problem of causal discovery focuses on learning the causal structure, such as identifying the subset of covariates with non-zero coefficients in $\beta^0$. Aggregating ``causal information'' from different environments requires additional knowledge about the environment, besides the data samples. In the above example, we know which covariates are randomized in each environment. But other situations may provide us with other types of ``causal information'' that we can aggregate to the experimental data. Instrumental Variable (IV) methods is one such situation where causal information is embedded in the requirements for a variable to be an instrument. More generally, knowledge about the structural causal model can also be aggregated to the previously mentioned methods. In particular, knowing the parental set of a covariate is another type of causal knowledge that we can leverage.

\begin{example}
We observe an additional variable $I$ within the linear SEM from the observational environment. Several assumptions are needed for $I$ to be an instrument. Most importantly, $I$ must be exogenous (has no parents in the graph) and does not directly affect the response $Y$ (exclusion restriction). The last condition needed asks that $I$ is relevant, that is, it affects the covariates in the model. However, we may not know in practice which covariates are direct descendants of $I$. Also, in usual IV regression we need at least as many instruments as there are confounded covariates. Can we leverage and combine incomplete causal knowledge derived from IV methods with experimental data? Our causal aggregation methodology does not require a priori to know which of $X_1$ or $X_2$ are affected by $I$. In particular, our method recovers $\beta^0$ by combining the IV based information from the model in Figure~(\ref{fig:Ex1-3}) and the information from the environment where $X_2$ is randomized.
\begin{figure*}[ht]
\centering
\includegraphics[trim={5cm 7.5cm 5cm 7.5cm},clip,width=0.6\linewidth]{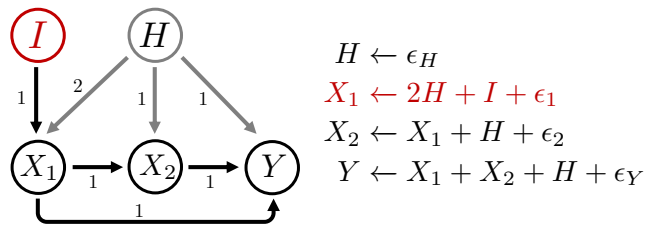}
\caption{Instrumental variables as causal information. A new observed variable is added to the graph: it corresponds to an instrument that leads to the second orthogonality constraint \eqref{equation:IVconstraint}. The noise contribution $\epsilon_i$ are jointly independent standard Gaussian.}
\label{fig:Ex1-3}
\end{figure*}
\end{example}

\textbf{Leveraging causal constraints arising from environments.} Example~\ref{example:LinearSEM} shows that recovering the causal model that defines $Y$ is a non-trivial task if the ability to intervene on the system is limited to manipulating small subsets of covariates, even in simple linear models. Fortunately, it is possible to extract from each experimental data set a set of constraints that partially identify $\beta^0$. Combining several of such constraints helps us identify $\beta^0$ and derive a procedure for constructing a consistent estimator. The intuition is that randomizing a covariate introduces exogenous randomness that is independent from all the other remaining elements in the system. For example, randomizing $X_1$ in Example~\ref{example:LinearSEM} removes the confounding effect induced by the latent variable $H$ when regressing $Y$ on $X_1$. This external manipulation of the covariate modifies the structural equation defining it (cf. top Figure~\ref{fig:Ex1-2}). Denoting by $(Y^e, \Xb^e)$ for $e\in \{1,2\}$ the random variables generated from the above model where $X_1^1$ (resp. $X_2^2$) have been randomized, we get the following system of equations in $\beta$ that is solved by $\beta^0$:
\begin{align}\label{eq:example-constraint}
\begin{cases}
& \mathbb{E}\big[X_1^1(Y^1- \Xb^{1,T}\beta)\big] = 0
\\ & \mathbb{E}\big[X_2^2(Y^2- \Xb^{2,T}\beta)\big] = 0
\end{cases}
\end{align}
under the assumption of a linear model for $Y|\Xb$ parametrized by $\beta = (\beta_1, \beta_2)$, where $\Xb = (X_1, X_2)$. These constraints reflect the independence between the randomized covariate and the residual term of the regression under the correct parameter value. Each equation imposes a ``causal constraint'' on the vector $\beta$. Given enough such constraints, $\beta^0$ is identified and the estimator $\hat{\beta}$ obtained as the solution to the empirical counterpart of the above system of equations is a consistent estimator of $\beta^0$. Causal constraints may originate from other assumptions on the data. If $I$ is an instrument for $X_1$ in the example above, then the following orthogonality constraint 
\begin{equation}\label{equation:IVconstraint}
    \mathbb{E}\big[I(Y-\Xb^T\beta)\big] = 0
\end{equation}
is satisfied whenever $\beta = \beta^0$. Other constraints can be constructed based on additional knowledge of the structural equations, in particular whenever we know the parental set of a given variable in the context of graphical representation of causal models. More generally, the structural equation of $Y$ may include non-linear terms in $\Xb$ that include interaction terms between the covariates, represented by a function $f^0(\Xb)$. Randomizing different subsets of covariates across different environments also leads to a joint system of equations similar to the system \eqref{eq:example-constraint} over the several environments: we estimate $f^0$ by constraining an estimator $\hat{f}$ so that the residuals $Y-\hat{f}(\Xb)$ and the randomized covariates are orthogonal.

\textbf{Intervening on subsets of $X$.} In this paper, we are interested in quantifying the total causal effect of \emph{joint} interventions on $X_1,\ldots,X_p$, which means that the variables $X_1,\ldots,X_p$ are set to a value at the same time. This is different from intervening just on one single variable. For example, in Figure~\ref{fig:Ex1-2}, intervening on $X_1$ changes the distribution of $X_2$ which then subsequently changes the distribution of $Y$. If on the other hand one intervenes on both $X_1$ and $X_2$ simultaneously, then randomization of $X_2$ destroys the causal pathway $X_1 \rightarrow X_2 \rightarrow Y$ and thus changes propagate through the system differently as in the case where only $X_1$ is randomized. If one is interested in the effect of an intervention on a subset $X_S \subseteq \{X_1,\ldots,X_p \}$, the methods of this paper still applies since one can simply set $\tilde X = X_S$ and apply the methods below for the subset of variables. 

\textbf{Our Contribution.} We define a general procedure for aggregating causal information from different environments where we leverage, in addition to the data samples, our knowledge of how the environment $e$ is generated. We present in Section~\ref{section:setting} a formal description of environments and how they relate via experimental manipulations. Starting with a linearity assumption in Section~\ref{section:linear}, we select \emph{constraint-inducing} variables $R$ based on such knowledge that define constraints of the form:
\begin{align*}
    0 = \mathbb{E}[R(Y - \Xb^T \beta^0 )]
\end{align*}
that must be satisfied by the true structural parameter $\beta^0$. In practice, constraints originate from randomized experiments, from the existence of an instrumental variable, but also if we have additional knowledge of the structure of the generative process. For example, if we know the parental set of a variable we can obtain another constraint via regression adjustment. These constraints are obtained with data from different sources that partially share some structure, in particular the structural parameter $\beta^0$. The general idea then is to construct estimators that simultaneously satisfy all the available constraints. We propose simple estimators that are asymptotically unbiased and normally distributed, provided that we have as many constraints as covariates that are not linearly dependent, under mild additional regularity assumptions such as finite moments. We also provide conditions for identifiability of $\beta^0$ based on these constraints. Finally, we analyze the case where the linear system is over-determined. Our solution relies on the method of moments (MM) theory, and we show how to optimally weight the constraints to obtain an asymptotically efficient estimator of $\beta^0$. Aggregating additional constraints obtained from new environments always reduces the asymptotic variance of our estimator.

A major challenge arises in the high-dimensional case, where there are not enough constraints or samples to accurately use the main methodology from the low-dimensional linear case. This is important whenever there is a very large number of experiments, but within each data set only a few samples are available. Another use case is when the system is under-determined, in the sense that we have fewer constraints than the number of covariates: multiple values for the estimator can simultaneously satisfy all constraints, and $\beta^0$ is not identifiable. Additional assumptions such as sparsity in $\beta^0$ can help: we propose in Section~\ref{section:high-dimensional} a regularized estimator similar to Causal Dantzig \citep{rothenhausler2019causal} that provably recovers $\beta^0$ under additional regularity assumptions. Whenever we have a large number of experiments, and the sample size in each experiment is large compared to the logarithm of the number of covariates, our method leads to a consistent and data-efficient procedure to estimate the causal parameter. In the purely under-determined case, additional constraints are required on the structure of the causal model for our method to work. Since these restrictions are relatively strong, in this high-dimensional case we recommend reducing the number of covariates through a pre-screening step before using regularized causal aggregation.

Building on the linear aggregation framework, we develop in Section~\ref{section:non-linear} a non-linear causal aggregation procedure following a boosting approach to construct estimators from a simple class of base learners. This procedure opens the possibility to learn complex non-linear causal models with potential interactions between covariates. We assume that the noise is additive but not necessarily independent of covariates. Using unconfounded (randomized) covariates, we construct individual base learners within each environment, and then combine these using a linear aggregation step as a sub-routine in the general boosting update. We assume that the randomized covariates are known in each environment. In addition to this flexible procedure, we characterize necessary conditions on the environments to identify the true non-linear causal structural equation of the response.

\section{Related Work}\label{section:related_work}
Learning a causal structure from multiple data sources has a long history in the literature. \citet{cooper1997simple} developed an algorithm for causal discovery based on independence tests. \citet{tian2001causal} combine a constraint based approach with background knowledge inferred from analyzing interventional data. \citet{sachs2005causal} use a score-based algorithm that searches through the space of directed acyclic graphs. \citet{eaton2007exact} use Bayesian inference to learn the causal structure from interventions with unknown effects.  \citet{eberhardt2010combining} describe a method to learn both the graph and the causal relations between a set of variables in presence of confounding based on experiments by first estimating a total effect matrix and then inferring direct effects from it. \citet{hauser2012characterization} present a modification of greedy equivalence search to learn the causal structure from interventional data. Similarly, \citet{hyttinen2012learning} discuss identifiability and present a search algorithm for learning linear cyclic models in presence of latent variables.  \citet{mooijheskes2013} learn cyclic causal models from equilibrium data collected under different experimental and observational contexts. \citet{mooij2016joint} propose an approach that allows different types of interventions and can be seen as a unification of several existing discovery methods. In contrast to these methods, we do not aim to reconstruct the causal graph or the overall structure but leverage constraints on causal effects.

Relatively recently, invariance principles have been exploited to estimate causal effects based on several data sets \citep{peters2016causal,heinze2018invariant,magliacane2018domain,pfister2019invariant,rothenhausler2019causal}. Our approach is similar in the sense that we can use data from different environments. In some of the work, it is possible to add background knowledge in the form of logical constraints to the optimization procedure. \citep{hyttinen2014constraint,magliacane2018domain}. In contrast, we allow to incorporate information about parental sets and experimental data not as logical constraints but as gradient information. 

Related to our work is do-calculus and the data-fusion framework \citep{bareinboim2013general,bareinboim2013meta,pearl2014external,bareinboim2016causal}. In this line of work, the authors present a powerful nonparametric framework to combine data sets to estimate a causal query, given the directed acyclic graph. Our approach is different in the sense that we do not assume that the graph is known but restrict the functional complexity of the structural equations. This distinction can be important in practice, since the graph is unknown in many cases. While the graph can be estimated using structure learning algorithms, such algorithms will usually make some errors, which propagate to the estimation step. In general, it is challenging to provide confidence intervals that take into account both the uncertainty in graph estimation and the uncertainty due to the estimation step. The proposed estimator allows to skip the graph estimation step which results in straightforward uncertainty quantification. In cases where the graph is known or can be estimated with high accuracy, the output of do-calculus and the data fusion framework can often be used in conjunction with the proposed method. This is explored in Section~\ref{sec:do-calc-sim}, where we demonstrate that causal constraints from the data fusion framework and other background knowledge can be used to improve precision.

In the potential outcome framework, there has been some recent work to combine data sets for causal inference. \citet{athey2016estimating} use surrogates to estimate long-term outcomes. \citet{kallus2018removing} use limited experimental data to remove the confounding on larger observational data under a linearity assumption. \citet{yang2020combining} consider a setting where the researcher is given a small unconfounded validation data set and a large confounded data set.

In applications, it is common to use two-stage least squares with multiple instruments, see for example \citet{mogstad2019identification}. We add to this literature by allowing for a wider range of causal constraints not necessarily based on linear models. Non-linear instrumental variable methods \citep{newey2003instrumental, carrasco2007linear, darolles2011nonparametric, singh2019}  provide a flexible framework for estimating non-linear causal effects based on exogenous instruments. Our non-linear aggregation method instead focuses on combining causal information based on different experimental datasets. A boosting version of non-linear instrumental variables has been recently proposed \citep{bakhitov2021causal} which is closer to our work. Again, these methods are devised for observational datasets with instruments, as opposed to considering the more general problem of leveraging different causal identification strategies across multiple data sets.

\section{Setting and Notation}\label{section:setting}

We assume that samples originate from different sources that we call environments, each characterized by their own data generating distribution. Let $\mathcal{E}$ denote the set of environments. We collect data from $E = |\mathcal{E}| \geq 2$ environments, where for each $e \in \mathcal{E}$ we have $n_e$ i.i.d. samples $(\Xb_i^e, Y_i^e)_{1\leq i \leq n_e}$, and let $n = \sum_e n_e$ the total number of samples. The $p$-dimensional random vector $\Xb^e \in \mathbb{R}^p$ (we denote by bold letters multivariate random variables) corresponds to the covariate vector in environment $e$, and $Y^e$ corresponds to the response variable. When needed, for simplicity we denote the response via a $p+1$-th indexed covariate $X_{p+1}^e := Y^e$. The environment $e$ is therefore characterized by $\mathbb{P}^e = \mathcal{D}(\Xb^e, Y^e)$ where $\mathcal{D}(U)$ denotes the distribution of $U$. In particular, we assume that we collect across environments the same real-valued covariates indexed by $[p] := \{1, \dots, p\}$, so that $\{\mathbb{P}^e, e \in \mathcal{E}\} \subset \mathcal{P}(\mathbb{R}^p\times \mathbb{R})$ where $\mathcal{P}(\mathcal{X})$ denotes the set of probability distributions over the space $\mathcal{X}$. Additionally, we assume that there are $p'\geq 1$ unobserved variables $\Hb= (H_l)_{l\in [p']}$ that jointly affect the covariates and the response. In order to define how the different $\mathbb{P}^e$ are related, we start with an observational base distribution $\mathbb{P}^0 \in \mathcal{P}(\mathbb{R}^p\times \mathbb{R})$. We assume that samples from $\mathbb{P}^0$ are generated by an acyclic linear Structural Equation Model (SEM) $\mathcal{M}^0$ with latent variables \citep{bollen1989structural,pearl2009causal}:
\begin{equation}\label{def:linearSEM}
\mathcal{M}^0 : \begin{cases}
    X_j \longleftarrow \sum_{k \in pa_0(j)}a_{jk}X_k + \sum_{l\in [p']}c_{jl}H_l + \epsilon_j \qquad \text{for all}\; 1\leq j \leq p
    \\ Y \longleftarrow \sum_{k \in [p]}\beta^0_{k}X_k + \sum_{l \in [p']}d_l H_l + \epsilon_{Y}
\end{cases}
\end{equation}
where $(a_{jk})_{1\leq j\leq p,1\leq k\leq p+1} \in \mathbb{R}^{p\times (p+1)}$, $\beta^0 \in \mathbb{R}^p$ and $\text{pa}_0(j):= \{k: a_{jk\neq 0}\} \subset [p+1]$. For concreteness, the effect of latent variables is assumed to be linear with coefficients $c_{jl}\in \mathbb{R}$, $d_l\in \mathbb{R}$, although our theory actually does not require linear latent effects on the observed variables. We associate to $\mathcal{M}^0$ a directed graph $G^{0} = (V^0, E^0)$ where the set of vertices $V^0 = [p+1]$ corresponds to the observed variables in the SEM, and a directed edge $(j,k) \in E^0$ iff $k\in \text{pa}_0(j)$. The subset $\text{pa}_0(j)$ corresponds to the \emph{parent nodes} of $j$ in $G^0$, which we assume is a Directed Acyclic Graph (DAG). This translates into constraints on the sets of coefficients $\{a_{jk}\}_{jk}, \beta^0$ that define equation~\ref{def:linearSEM}. We allow the response variable $Y$ to be a parent node of any covariate. The random variables $\epsb := \{\epsilon_j\}_{j \in [p]}\cup \{\epsilon_Y\}$---also referred as disturbance terms---are assumed to be centered, and have finite second moments. Furthermore, the disturbance terms are assumed to be independent of $\Hb$ and jointly independent. For simplicity, we may also denote $\epsilon_{p+1}:=\epsilon_Y$ whenever we refer to sets $\{\epsilon_j\}_{j\in J}$ with $J\subset[p+1]$. 
We can complete the graph $G^0$ by adding the vertices $\Hb = (H_l)_{ 1 \le l \le p'}$ to the set of nodes and define an extended graph $\bar{G}^0 := (\bar{V}^0, \bar{E}^0)$, where $\bar{V}^0 = V^0 \cup \{\Hb\}$ and $\bar{E}^0$ contains all edges in $E^0$ plus edges between a component $H_l$ and nodes of variables in $(\Xb, Y)$ whenever any latent variable $H_l$ has an effect on that variable. 
That is, unless we have explicit indication of the contrary, we assume that all covariates are potentially affected by the latent factors (i.e. $c_{jl}\neq 0, d_l \neq 0$). The graph $G^0$ encodes the input-output relations between observed variables given by the structural equations in $\mathcal{M}^0$. The structural model implicitly assumes that latent variables are not affected by observed variables. Figure~\ref{fig:Def-1} is an example of a most general $\bar{G}^0$ under our model. The joint distribution $\mathbb{P}^0$ over $(\Xb, Y)$ is properly defined given distributions $\mathcal{D}(\epsb), \mathcal{D}(\Hb)$ and the structural equations: we reformulate the model $\mathcal{M}^0$ as a linear system of equations.
\begin{equation}
    \mathcal{M}^0 : 
    \begin{pmatrix}
    \Xb \\ Y
    \end{pmatrix} \longleftarrow M^0\begin{pmatrix}
    \Xb \\ Y
    \end{pmatrix} + M_H\Hb + \epsb
\end{equation}
where $M^0\in \mathbb{R}^{(p+1)\times(p+1)}$ is a matrix that contains the structural parameters $a_{jk}, \beta^0_k$ between observed covariates and $M_H\in \mathbb{R}^{(p+1)\times p'}$ is the matrix containing coefficients $c_{jl}, d_l$. As $G^0$ is acyclic, the matrix $(I_{p+1} - M^0)$ is invertible. This guarantees that the distribution over the observed variables is well-defined.
\begin{figure*}[ht]
\centering
\includegraphics[trim={1cm 6.5cm 1cm 7cm},clip,width=.8\linewidth]{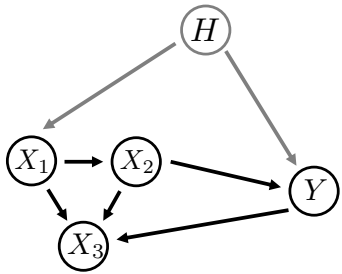}
\caption{Example of a SEM model: The SEM contains latent variables that have a confounding effect on the covariate-response relationship, and the response is potentially a parent of other covariates.}
\label{fig:Def-1}
\end{figure*}
As discussed in Section~\ref{section:introduction}, estimation of $\beta^0$ based only on samples from $\mathbb{P}^0$---for example, using least squares regression---might be subject to bias due to the presence of the latent variables. We want to impose the weakest possible assumptions on how the distributions $\mathbb{P}^e$ are generated, but that still allow us to recover $\beta^0$, which we assume remains invariant across environments. We now introduce causal models that encode the different types of interventions that lead to different environments. Importantly, in practice we only need to know which covariates receive the interventions. Knowledge of the graph structure of $\bar{G}^0$ is not needed, although we will show later how such information may be helpful in certain cases. However, we will in general estimate parameters that are related to those variables that are perturbed: having a flexible model for representing perturbations for as many variables as possible is therefore crucial in our framework.

New environments arise when subsets of covariates are manipulated. We use superscripts $e\in \mathcal{E}$ to indicate the elements of the structural equation that are environment-dependent. For environment $e$, we index by $\phi(e) \subset [p]$ the subset of randomized covariates, which may be an empty set. We can then write the corresponding structural equations that define the distribution $\mathbb{P}^e$:
\begin{align}\label{equation:perturbedSEM}
\mathcal{M}^e_{\phi(e)} : \begin{cases}
    X_j^e \longleftarrow \sum_{k \in pa_e(j)}a^e_{jk}X_k^e + \sum_{l\in [p']}c^e_{jl}H_l^e + \epsilon^e_j & \forall j \notin \phi(e)
    \\ X_j^e \longleftarrow \epsilon_j^e & \forall j \in \phi(e)
    \\ Y^e \longleftarrow \sum_{k \in [p]}\beta^0_{k}X_k^e + \sum_{l \in [p']}d_l^e H_l^e + \epsilon^e_Y
\end{cases}
\end{align}
In particular $(p+1) \notin \phi(e)$ for all $e\in \mathcal{E}$, i.e. no interventions allowed on the response variable. We now define a graph $G^e = (V^e, E^e)$ as above that encodes the model $\mathcal{M}^e$: we have that $V^e = V^0$, and we assume that the edges are such that $\text{pa}_e(j) \subset \text{pa}_0(j)$ for all $j\in [p+1]$, which implies that $G^e$ is also a DAG. The randomization intervention on covariates $\Xb_{\phi(e)}^e$ implies that the corresponding nodes have no incoming edges. That is, $E^e = E^0\setminus\{(j,k) \in E^0 | k\in \phi(e)\}$. The set of coefficients $\{a_{jk}^e\}_{jk}$ may change provided that no new dependencies between covariates are created in the structural equations. Variables $\epsb^e$ satisfy the same assumptions as in the observational environment: they are centered, have finite second moments, are independent of $\Hb^e$ and jointly independent. We also allow for arbitrary changes in the coefficients $c_{jl}^e$ and $d_l^e$ that are non-zero, as well as the latent variable distribution $\mathcal{D}(\Hb^e)$ unless otherwise indicated. The completed graph $\bar{G}^e := (\bar{V}^e, \bar{E}^e)$ still has $\bar{V}^e = \bar{V}^0$, but $\bar{E}^e = E^e \setminus \{(H_l,j)\}_{j \in \phi(e)}$: randomization removes edges connecting the latent confounders with the randomized covariates.

In conclusion, the assumptions on how this intervention mechanism acts on the system can be summarized as constraints on the extended graphs $\bar{G}^e$. We additionally assume that the different distributions $\{\mathbb{P}^e\}_{e\in \mathcal{E}}$ have the same support. 

This model for generating environment distributions $\mathbb{P}^e$ is closely related to ``surgical interventions'' in \cite{pearl2009causal}, also called ``ideal interventions'' \citep{spirtes2000causation} or ``structural interventions'' \citep{eberhardt2007interventions}. Such interventions replace those structural equations of the intervened variables $j\in \phi(e)$ by an independently generated $\epsilon^e_j$, leaving the rest unchanged. That is a strong assumption: manipulating variables could lead to ``spill-over'' effects on other variables whose disturbance distribution is shifted or structural equations perturbed. In contrast, our definition of $\mathbb{P}^e$ allows for changes in the structural equations of the other covariates as well as the joint distribution of the latent variables $\Hb$ and disturbance variables $\epsb$: we only impose invariance assumptions on $\beta^0$ and on the parental sets $pa_e(j)$ for all $e$ that are contained in $pa_0(j)$.

As long as we assume that $\beta^0$ is invariant, we can define additional models for generating environments. We introduce in the Appendix~\ref{section:additive_interventions} a model where environments are generated by additive shifts in the distribution $\mathcal{D}(\epsb)$ of the disturbance terms. These are a special case of so-called ``parametric'' interventions \citep{eberhardt2007interventions} or ``soft interventions'' \citep{eaton2007exact}. Randomizing covariates is a direct experimental intervention: the validity of the resulting constraints is thus verified by the scientist's intervention on the data generating procedure. On the other hand, the effects of soft interventions are sometimes not as easily verifiable, as shifts in latent factors are often not under the direct control of the scientist. Thus we need to trade-off this ``weaker'' causal knowledge with stronger assumptions. In particular we require the stability of the structural equations (i.e. the coefficients $a_{jk}$) across environments, as well as the base distributions $\mathcal{D}(\Hb)$, $\mathcal{D}(\epsb)$. To summarize, we can aggregate causal information originating from a flexible collection of models that represent environment heterogeneity, where side knowledge on the ``causal'' mechanism that generates the environment trades-off with our assumption on how stable such mechanism is across environments.

We additionally present in Section~\ref{section:non-linear} a non-linear extension of the above SEM. This increased flexibility restricts the types of causal information that we can aggregate to those arising from a randomized experiment only, as we can no longer leverage instrumental variables or covariate adjustment whenever the parental set of a variable is known. This model not only allows non-linear response structural equations, but can also include interaction terms between covariates. Identifying these effects is possible when an environment simultaneously randomizes the interacting covariates: the causal information contained in environments arising from other types of interventions---such as additive shifts---is harder to aggregate.

We use the following notation throughout the paper. For a vector $\ub \in \mathbb{R}^p$, we write for any subset $J \subset [p]$, $\ub_J := (u_j)_{j\in J}$. We denote by $|J|$ the cardinality of set $J$. For any $q\in [1,\infty]$, let $\Vert \ub\Vert_{q}$ be the $q$-norm of $\ub$, and given a positive semi-definite symmetric matrix $\Mb$, let $\Vert\ub\Vert_{\Mb}^2 := \ub^T\Mb\ub$. We denote by $e_j\in \mathbb{R}^p$ the $j$-th coordinate vector, a one-hot vector where the $j$-th coordinate is equal to $1$.

\section{Causal Aggregation in the Linear Case}\label{section:linear}

Assuming a linear structural equation model leads to an intuitive method for aggregating causal information where constraints are derived from each environment and we build an estimator to simultaneously satisfy these constraints. We start by characterizing constraints in Section~\ref{section:subsection:linear_constraints} and then formulate our aggregation procedure in Section~\ref{section:subsection:linear_aggregation}.

\subsection{Linear Constraints}\label{section:subsection:linear_constraints}

We can identify the true vector $\beta^0$ in our linear SEM model $\mathcal{M}^0$ by aggregating multiple sources of information about the causal structure. As illustrated in Example~\ref{example:LinearSEM}, randomization of a covariate leads to a linear constraint that should be satisfied by any estimator $\hat{\beta}$ of $\beta^0$. Linear constraints are also obtained if instrumental variable (IV) assumptions hold for some specific variable. Broadly speaking, under the linear SEM model the residual term $Y - \beta^T\Xb$ is equal to $\epsilon_Y$ only for the true value of $\beta = \beta^0$. Based on the assumption at hand, we formulate a linear equality that must be satisfied by the true parameter $\beta^0$, which is generally a consequence of the independence between $\epsilon_Y$ and other variables in the model. 

\subsubsection{Instrumental Variables}
We obtain an orthogonality constraint whenever we have access to an instrument $I$. Instrumental Variables (IV) methods \citep{wright1928tariff,heckman1990varieties,angrist1995identification, angrist1996identification} are based on several assumptions that we translate in graphical terms within our linear SEM. For those environments $e$ where $I$ is available, $\bar{G}^e$ has an additional node $I$. The instrument is \emph{exogenous}, meaning that the node has no parents: there is no directed edge from $H$ nor any other node to $I$. It is \emph{relevant} and satisfies the \emph{exclusion restriction}: these assumptions correspond to $I$ having only covariate nodes as potential children but not the response $Y$ - and at least one such child. For the true $\beta^0$, the following constraint holds:
\begin{equation}\label{eq:orthogonality_IV}
    0 = \mathbb{E}[I^e(Y^e - \beta^{0,T}\Xb^e)]
\end{equation}
The IV constraint is a direct consequence of the independence between the instrument and the response disturbance term. An instrument does not need to be measured in every single environment: the causal information does not rely on comparing different environments unlike other constraints later described. In practice, we use the sample average to build the constraint that must be satisfied by the estimator:
\begin{equation*}
    0 = \frac{1}{n_e}\sum_{i\in [n_e]}I^e_iY^e_i - \big(\frac{1}{n_e}\sum_{i\in [n_e]}I^e_i\Xb^e_i\big) \beta
\end{equation*}
Estimating $\beta^0$ via IV methods is usually done by solving a generalized method of moments (MM) problem based on constraints originating from several instruments. A general treatment of IV methods can be found in reference textbooks \citep{hall2005generalized} which provide necessary conditions for identification of $\beta^0$. In particular, there must be at least as many instruments as covariates to build an estimator $\hat{\beta}$ by solving the potentially overidentified system of equations. 

\subsubsection{Experimental Data from Randomization}

A more direct method for obtaining an orthogonality constraint is via randomization of a covariate. Following our definition of a model $\mathcal{M}_{\phi(e)}^e$ for environment $e$ where the subset $\phi(e)\subset [p]$ indexes the covariates that are randomized, the following equation holds:
\begin{equation}\label{eq:orthogonality_intervention_linear}
    0 = \mathbb{E}[X^e_j(Y^e - \beta^{0,T}\Xb^e)] \qquad \forall j\in \phi(e)
\end{equation}
Again, the equation above is a direct consequence of the independence between the randomized covariate and $\epsilon_Y^e$. Although data from experimental sources is expensive and limited, randomization provides a solid guarantee that the assumption leading to the orthogonality condition for identifying $\beta^0$ holds. Again, we use sample averages when building the constraint:
\begin{equation*}
    0 = \frac{1}{n_e}\sum_{i\in [n_e]}X^e_{j,i}Y^e_i - \big(\frac{1}{n_e}\sum_{i\in [n_e]}X^e_{j,i}\Xb^e_i\big)^T \beta
\end{equation*}
As indicated in Example~\ref{example:LinearSEM}, full simultaneous randomization of the covariates leads to a linear system of equations with $\beta^0$ as the unique solution. Our proposed method allows to individually treat the constraint from each randomized covariate and thus be able to aggregate them across environments.

\subsubsection{Regression Adjustment}\label{section:regression_adjustment}
Causal inference based on graphical models heavily relies on conditional independence statements between variables that are encoded by the DAG. Pearl's do-calculus identifies causal effects in a causal DAG by transforming conditional statements based on intervened distributions into conditional statements on the observational distribution \citep{pearl2009causal}. A fundamental assumption is the knowledge of the DAG, which often times needs to be estimated in practice. Errors in estimating the graph can drastically change the conclusions on causal effects, which is one motivation for developing this aggregation framework which circumvents estimating $\bar{G}^0$. However, we may have partial knowledge of the graph structure, which can be incorporated into our causal aggregation methodology in the form of constraints. If, for a given variable $X_j$, we assume that we know its parental set in $\bar{G}^0$ and that the latent confounders are not in such parental set, then based on the ``adjustment for directed causes'' property   \cite[Theorem 3.2.2]{pearl2009causal} we have the following constraint for variable $j$ :
\begin{equation}\label{eq:orthogonality_adjustment}
    0 = \mathbb{E}[X_j^e(Y^e - \beta^{0,T}\Xb^e)| \Xb_{pa_0(j)}^e]
\end{equation}
In essence, we use a conditional independence property based on the fact that the distribution $\mathbb{P}^e$ factorizes in $\bar{G}^0$, but as we can not condition on the unobserved confounder we need to additionally assume that it is not in the parental set of the node in $\bar{G}^0$, which is thus the same as the parental set in $G^0$. In practice, we consider the residual variable derived from regressing $X_j^e$ on its parents to derive the linear constraint:
\begin{align*}
    0 = \frac{1}{n_e}\sum_{i\in [n_e]}(X^e_{j,i}-\hat{X}^e_{j,i})Y^e_i - \big(\frac{1}{n_e}\sum_{i\in [n_e]}(X^e_{j,i}-\hat{X}^e_{j,i})\Xb^e_i\big)^T \beta,
\end{align*}
where $\hat{X}^e_{j,i} := \sum_{k\in pa_e(j)}\hat{a}_{jk} \Xb^e_{k,i}$ and $\hat{\gamma}_j := (\hat{a}_{jk})_{k\in pa_e(j)}$ is the regression coefficient of $X_j^e$ on $\Xb^e_{pa(j),i}$. Importantly, to simplify asymptotic deductions below, we assume that the estimator $\hat{\gamma}$ is computed on a different data set than the one used for constructing the orthogonality constraint (for notation simplicity we will use $\gamma$ to refer to the regression coefficient of a covariate on its parental set). In our framework, this is not very stringent, as different environments may provide enough samples to do this. Any two environments where we know that the structural equation of $X_j$ has not changed---this precludes any environment with $X_j$ randomized---can be used for estimating $\gamma$ and the orthogonality constraint independently (hence the notation without environment subscript).

\subsubsection{Additive interventions across environments}
We can derive constraints based on how different environments \emph{relate to each other}. In contrast with all previously mentioned examples of ``causal information'', we can derive orthogonality constraints based on the \emph{inner product invariance} \citep{rothenhausler2019causal} under $\beta^0$ for pairs of environments that are generated via additive interventions with respect to the observational base distribution $\mathbb{P}^0$. Assume that in environment $e$ covariates $X_j$ for $j\in \psi(e)$ have an additive intervention given by the model $\mathcal{M}_{\psi(e)}^e$ as defined in Appendix~\ref{section:additive_interventions}. The distributional shift induced by the additive intervention leaves the expression $\mathbb{E}[X^e_jY^e - X_j^e\Xb^{e,T} \beta^0]$ invariant across environments. The intuition behind this approach is that the covariance between $X^e$ and the residuals $Y^e - \Xb^{e,T} \beta^0$ is a measure of the strength of confounding. Under certain assumptions, the strength of confounding is invariant across settings, which can be leveraged for statistical inference. We thus have the following orthogonality constraint for variable $X_j$ for $j\in \psi(e)$ by merging data from the base distribution and the intervened one: 
\begin{equation}\label{eq:orthogonality_dantzig}
    0 = \mathbb{E}[X^e_j Y^e - X_j^0 Y^0 - (X_j^e\Xb^{e,T} - X_j^0\Xb^{0,T})\beta^{0}]
\end{equation}
This orthogonality constraint is crucial for the Causal Dantzig~\citep{rothenhausler2019causal}. We leave it as an additional way of constraining the parameter vector that leads to consistent estimators of $\beta^0$, but for readability it will not be included in finer analysis of the asymptotic behavior. We summarize in the following proposition the set of orthogonality constraints that we derived in the previous sections.
\begin{proposition}\label{proposition:linear_constraints}
The causal vector $\beta^0$ satisfies the linear constraints as defined above via either instrumental variables in eq.~\eqref{eq:orthogonality_IV}, randomization in eq.~\eqref{eq:orthogonality_intervention_linear}, regression adjustment in eq.~\eqref{eq:orthogonality_adjustment}, or inner product invariance in eq.~\eqref{eq:orthogonality_dantzig}. In all these cases the constraint in $\beta$ can be summarized 
\begin{equation}\label{eq:simple-linear-constraint}
    \gb^T \beta = z
\end{equation}
where $\gb \in \mathbb{R}^p$, $z\in \mathbb{R}$ are some specific transformations of population-level cross-covariances obtained at each environment.
\end{proposition}

As we see, different orthogonality constraints are derived from different types of prior causal information. We can potentially have more constraints than strictly necessary to estimate $\beta^0$, and we may be willing to discard those that rest on weaker foundations. In practice, the assumptions leading to constraints are not on an equal footing. Data obtained via covariate manipulation and the subsequent orthogonality constraint has a better standing than a constraint generated by assumptions on additive shifts in the covariance structure, which are less verifiable in practice. Conditioning on the parental set assumes an accurate knowledge of the (potentially estimated) graph, whereas there may be settings where one can be confident that the exogeneity and exclusion restriction assumptions in IV hold.

\subsection{Aggregating Linear Constraints under Just-Identification}\label{section:subsection:linear_aggregation}

One can recover $\beta^0$ by exclusively using one type of the orthogonality constraints previously defined. For example, if all the covariates in the model are randomized, then OLS is unbiased. Whenever there are as many instruments as covariates (and a full-rank condition holds) then usual IV methods apply. Finally, Causal Dantzig \citep{rothenhausler2019causal} provably recovers $\beta^0$ whenever for every covariate there is an environment where the given covariate has an additive intervention. The objective is now to combine constraints arising from multiple data sets into one single estimator. Let $\mathcal{C}$ be the set of constraints, denoted by $\gb^{c,T} \beta = z^c$ for each $c\in \mathcal{C}$ as in equation \eqref{eq:simple-linear-constraint}. We aggregate these via a linear matrix equality by first defining the vector $\Zb$ and matrix $\Gb$ as follows:
\begin{align*}
    & \Zb := \big( z^c\big)_{c\in \mathcal{C}} \in \mathbb{R}^{|\mathcal{C}|} \qquad \text{and} \qquad \Gb := \big(\gb^{c,T}\big)_{c\in \mathcal{C}} \in \mathbb{R}^{|\mathcal{C}|\times p}
\end{align*}
The vector $\beta^0$ is then a solution to the linear system:
\begin{equation}\label{equation:population-causal-aggregation-estimator}
\Zb = \Gb\beta
\end{equation}
Without any prior assumptions on the SEM $\mathcal{M}^0$, a necessary condition for identifying $\beta^0$ is to have at least as many constraints as covariates, i.e. $|\mathcal{C}| \geq p$, otherwise the system \eqref{equation:population-causal-aggregation-estimator} has multiple solutions. A sufficient condition so that $\beta^0$ is the unique solution to equation~\eqref{equation:population-causal-aggregation-estimator} above can be stated in purely mathematical terms. This is similar to the identifiability result in \cite{hyttinen2012learning} which we extend to any type of constraint, not only those based on randomization, although we constrain $G^0$ to be a DAG. With additional information about the constraints, we can formulate more practical necessary conditions. Consider instrumental variables in equation~\eqref{eq:orthogonality_IV}, experimental data through randomization from an interventional environment $e$ that follows the causal model $\mathcal{M}^e_{\phi(e)}$ in equation~\eqref{eq:orthogonality_intervention_linear}, or constraints from regression adjustment via the population regression vector of a covariate on its parental set in eq.~\eqref{eq:orthogonality_adjustment}. All these constraints feature a known specific random variable denoted $R^c$ that captures the prior knowledge about the causal structure of the data, which we refer to as the \emph{constraint-inducing} variable. We have respectively that $R^{c} = X^{e_c}_j$ for a randomized covariate $j\in \phi(e_c)$, or $R^{c} = I^{e_c}$ for an instrument $I^{e_c}$ whenever available, and $R^c$ is the residual term of regressing a covariate on its known parental set (using a different data set for estimating the regression adjustment and for estimating the orthogonality constraint). Additionally, each constraint-inducing variable relates to a covariate: either the randomized covariate itself, any covariate that is correlated to the instrument, or the covariate that is regressed on its parental set. A sufficient condition for identifiability of $\beta^0$ is then that each covariate has a distinct constraint related to it. This is again similar to \cite{hyttinen2012learning} where a necessary condition is to have an environment where each variable is randomized, and $Y$ is never intervened on.

\begin{proposition}\label{proposition:linear_identifiability}
$\beta^0$ is identified if there is a subset of $p$ linearly independent constraints in $\Gb$ in \eqref{equation:population-causal-aggregation-estimator}. If we only consider constraints derived from IV, randomization or regression adjustment, then it suffices to have at least one distinct constraint related to each covariate.
\end{proposition} 
Based on the empirical counterparts $\hat{\Zb}, \hat{\Gb}$ of $\Zb, \Gb$, we look for estimators $\hat{\beta}$ of $\beta^0$ that make the two sides of the equation \eqref{equation:population-causal-aggregation-estimator} as close as possible:
\begin{equation}\label{equation:empirical-causal-aggregation-estimator}
\hat{\Zb} \approx \hat{\Gb}\hat{\beta}
\end{equation}
This naïve approach only works in very limited situations, and we describe here the just-identified case. Whenever $|\mathcal{C}| = p$ and the corresponding square matrix $\Gb$ is invertible, $\beta^0$ is identified. Additionally, if $\hat{\Gb}$ is invertible, then the unique solution to equation~\eqref{equation:empirical-causal-aggregation-estimator} is given by the estimator 
\begin{equation}\label{def:estimator_just_identified}
\hat{\beta} := \hat{\Gb}^{-1} \hat{\Zb}
\end{equation}
which is consistent if $\hat{\Gb} \rightarrow \Gb$ and $\hat{\Zb} \rightarrow \Zb$ by continuity of the matrix inverse and product. Consistency of $\hat{\Gb}, \hat{\Zb}$ holds by the law of large numbers as soon as $n_e \xrightarrow[]{}\infty$ in every environment. In the following we analyze the asymptotic behavior of this estimator. We assume that sample sizes grow at the same rate, i.e. $n_e/n \rightarrow \rho_e$ for some $\rho_e \in (0,1)$. We discard constraints based on inner-product invariance for readability. Denote by $e_c\in \mathcal{E}$ the environment where constraint $c$ is generated, we can write: 
\begin{align*}
    & \Zb := \big( \mathbb{E}\big[R^{c}Y^{e_c}\big]\big)_{c\in \mathcal{C}} \in \mathbb{R}^{|\mathcal{C}|} \qquad \text{and} \qquad \Gb := \big(\mathbb{E}\big[R^{c}\Xb^{e_c,T}\big]\big)_{c\in \mathcal{C}} \in \mathbb{R}^{|\mathcal{C}|\times p}
\end{align*}
and the corresponding empirical counterparts
\begin{align}\label{equation:empirical-GZ}
\begin{split}
    & \hat{\Zb} := \big( (1/n_{e_c})\sum_{i\in [n_{e_c}]}R^{c}_iY_i^{e_c}\big)_{c\in \mathcal{C}} \in \mathbb{R}^{|\mathcal{C}|}
    \\ & \hat{\Gb} := \big( (1/n_{e_c})\sum_{i\in [n_{e_c}]}R^{c}_i\Xb_i^{e_c,T}\big)_{c\in \mathcal{C}} \in \mathbb{R}^{|\mathcal{C}|\times p}
\end{split}
\end{align}
Standard regularity conditions on the moments of the variables $R, \Xb, Y$ also lead to asymptotically valid confidence intervals for $\hat{\beta}$. Whenever a constraint is obtained by adjusting for the parental set we need that the estimates of $\gamma$ are obtained from a data set independent that the one used for constructing the constraints. Under these assumptions, we show that the estimator is asymptotically normally distributed:
\begin{equation*}
    \sqrt{n}\big(\hat{\beta} - \beta^0\big) \xrightarrow[n \to +\infty]{d} \mathcal{N}(\textbf{0}; \Sigma)
\end{equation*}
where $\Sigma$ is a positive definite matrix that we can consistently estimate by some $\hat{\Sigma}$. Therefore we can form asymptotically valid confidence intervals for $\beta^0_j$ by 
\begin{equation}\label{equation:confidence_interval}
    I_j = [\hat{\beta}_j \;\pm\; q_{\mathcal{N}}^{1-\alpha/2}\sqrt{\hat{\Sigma}_{jj}}]
\end{equation}
where $q_{\mathcal{N}}^{\alpha}$ is the $\alpha$ quantile of the standard Gaussian distribution, which has exact asymptotic coverage.
\begin{equation*}
    \mathbb{P}\big(\beta_j^0 \in I_j\big) \xrightarrow[]{+\infty} 1-\alpha
\end{equation*}
Let $\sigma_e^2 := \text{Var}(\epsilon_Y^e)$. We summarize these statements and give an expression for the asymptotic covariance in the following proposition.
\begin{proposition}\label{proposition:asymptotic_normality}
Consider the setting described above. Assume that for all $c$, all variables $R^c, \Xb^{e_c}, Y^{e_c}$ have finite fourth moments. We then have that
\begin{equation}
    \sqrt{n}\big(\hat{\beta} - \beta^0\big) \xrightarrow[n \to +\infty]{d} \mathcal{N}(\textbf{0}; \Sigma)
\end{equation}
where $\Sigma = \Gb^{-1}\text{Diag}\Big(\frac{\sigma_{e_c}^2}{\rho_{e_c}}\text{Var}(R^{c})\Big)_{c\in [p]}\Gb^{-1,T}$. A consistent estimator $\hat{\Sigma}$ is obtained via
\begin{equation*}
    \hat{\Sigma} = \hat{\Gb}^{-1}\text{Diag}\Big(\frac{n\hat{\sigma}^{e_c,2}}{n_{e_c}}\widehat{\text{Var}}(R^c)\Big)_{c\in [p]}\hat{\Gb}^{-1,T}
\end{equation*}
where $\hat{\sigma}^{e_c,2} = \frac{1}{n_{e_c}}\sum_{i\in [n_{e_c}]}\big(Y^{e_c}_i - \hat{\beta}^T\Xb^{e_c}_i\big)^2$.
\end{proposition}
This is a first asymptotic result for analyzing the convergence of an aggregation estimator to the true parameter $\beta^0$ based on multiple environments. For completeness, we now develop linear aggregation extensions to the over-identified case. Provided the model is correctly specified, additional constraints improve the efficiency of the estimator that we propose, based on a method of moments framework. In practice, some constraints potentially arise in several environments, but only one is kept in the just-identified case---if we assume we know the parental set of a covariate that is invariant across environments, then we have access to several regression adjustment constraints.  
\subsection{Linear Aggregation in the Over-Identified Case with More Constraints than Covariates}\label{section:overidentified}

The number of available constraints may surpass the dimension of the covariate vector. Provided that the model is well specified in the sense that the constraints are compatible with $\beta^0$ as solution, additional constraints improve the efficiency of the estimator by decreasing the asymptotic variance. In the ideal setting where in one environment all covariates are randomized, it may be tempting to discard data from other environments as direct least squares regression generates an unbiased estimator relying only on that environment. Other methods can also provide consistent estimators of the causal effects based on data from a single environment. The do-calculus is a method that, provided the underlying causal graph is properly estimated, is capable of assessing whether in a given environment the causal effects are identifiable. In such case, if identifiability holds, then the causal vector $\beta^0$ can be estimated via a regression over a judiciously chosen set of covariates. We show that, in any of these scenarios, we can still benefit from aggregating additional constraints from other environments into a single estimator. In this overidentified setting, $\hat{\Gb}$ is no longer a potentially invertible square matrix, so $\hat{\beta}$ can not be simply derived by inverting a matrix. We instead estimate it via the method of moments estimator (MM), we refer to \cite{hall2005generalized} for a general treatment of the subject. The standard MM conditions hold in our setting, where $\beta^0$ is identified if $\Gb$ is of rank $p$, and the estimators that we propose are consistent and asymptotically normal. The challenge is to construct one with minimal asymptotic variance, which requires adapting the usual MM framework to the multiple environment setting. In particular, to obtain the efficient two-step efficient estimator \citep{hansen1982large} we need to compute the covariance of the vector of constraints based on samples from different environments. Finally we show that incorporating additional constraints always lead to an improvement of the asymptotic variance.

We reformulate the multiple environment framework (cf. equation~\ref{equation:perturbedSEM}) to represent individual observations across environments as elements of a same space, sampled from a unique common distribution. We use an environment-indicator variable $E$ and then aggregate samples across all environments, the following collection of variables represents an individual observation:
\begin{align*}
    \big(E_i, Y_i^{E_i}, \Xb_i^{E_i}, \{\tilde{R}_i^{c}\}_{c\in \mathcal{C}} \big)
\end{align*}
The $i$-th observation originates from environment $E_i$, where $E_i \sim \text{Multinomial}\big((\rho_e)_{e\in \mathcal{E}}\big)$. Conditionally on the environment $E_i$, the values of $Y_i^{E_i}, \Xb_i^{E_i}$ correspond to the actual observations obtained in environment $E_i$. Additionally, each constraint in $c \in \mathcal{C}$ is obtained in a specific environment $e_c$. The aggregated $i$-th observation concatenates all $\tilde{R}_i^c = R_i^{c}\mathbb{1}_{E_i = e_c}$ for each $c\in \mathcal{C}$, which are equal to 0 whenever the constraint $c$ is not based on environment $E_i$. We again exclude the constraints arising from shift interventions. Samples across environments are thus merged into one data set such that samples can be considered as i.i.d. for $i\in [n]$. The $c$-th moment condition now becomes
\begin{equation}
    0 = \mathbb{E}\big[\tilde{R}_i^{c}(Y_i^{E_i} - \Xb_i^{E_i}\beta)\big] = \mathbb{E}\big[R_i^{c}(Y_i^{E_i} - \Xb_i^{E_i}\beta)\mathbb{1}_{E_i = e_c}\big]
\end{equation}
and therefore whenever $E_i$ is the environment where $R^c_i$ is obtained, we get that the above equality holds for $\beta^0$ by construction of our orthogonality constraints. In particular, after stacking the $|\mathcal{C}|=q$ constraints (for simplicity $\mathcal{C} = \{c_1,\dots,c_{q}\}$), we can reformulate the system of equations in terms of $\Zb$ and $\Gb$ as defined in the previous section:
\begin{align*}
    0 = \mathbb{E}\begin{bmatrix}
    R_i^{c_1}(Y_i^{E_i} - \Xb_i^{E_i}\beta)\mathbb{1}_{E_i = e_{c_1}}
    \\ \vdots
    \\ R_i^{c_{q}}(Y_i^{E_i} - \Xb_i^{E_i}\beta)\mathbb{1}_{E_i = e_{c_{q}}}
    \end{bmatrix} = \Db\big(\Zb - \Gb\beta\big)
\end{align*}
where $\Db = \text{Diag}\big((\rho_{e_c})_{c\in \mathcal{C}}\big)$. This shows that if $\Gb$ has rank $p$ then the solution to the system above is unique, equal to $\beta^0$. Given a positive definite weighting matrix $\Wb\in \mathbb{R}^{q\times q}$, the method of moments estimator of $\beta^0$ is the minimizer of the following loss with a closed-form solution:
\begin{align}\label{def:MM_optimization}
    \hat{\beta}^{MM}(\Wb) := & \argmin_{\beta\in \Theta} \big\Vert \hat{\Db}\big(\hat{\Zb} - \hat{\Gb}\beta\big)\big\Vert_{\Wb}^2
\end{align}
for the empirical values of $\hat{\Zb}$, $\hat{\Gb}$ and $\hat{\Db} := \text{Diag}\big((\frac{n_{e_c}}{n})_{c\in \mathcal{C}}\big)$. Under standard regularity assumptions on the moments of the variables, provided that $\Gb$ is of rank $p$, the estimator $\hat{\beta}^{MM}$ is consistent for any choice of weighting matrix. The choice of $\Wb$ characterizes the asymptotic covariance of the estimator. \cite{hansen1982large} proves that, to obtain the estimator with minimal asymptotic variance among all choices of weighting matrices, the optimal choice is given by $\Wb = \Sb^{-1}$, where
\begin{equation}
    \Sb := \text{Cov}\left(\big(\tilde{R}^{c}(Y^{E} - \Xb^{E}\beta^0)\big)_{c\in \mathcal{C}}\right) \in \mathbb{R}^{q\times q}
\end{equation}
Given that the residual term and the constraint-inducing variable $R^c$ are independent and that any two constraint-inducing variables are pairwise independent, we get that the constraints are uncorrelated and thus obtain the following simplification of $\Sb$ into a diagonal matrix:
\begin{equation}
    \Sb = \text{Diag}\Big(\rho_{e_c}\text{Var}(R^{c})\sigma_{e_c}^2\Big)_{c\in \mathcal{C}}
\end{equation}
$\Sb$ can only be estimated from data, but in order to do so one needs to estimate the residual variances. These, in turn, are given by $\sigma_{e_c}^2=\text{Var}(Y^{e_c} - \Xb^{e_c}\beta^0)$ which depend on the true unknown parameter $\beta^0$. \cite{hansen1982large} proposes a two-step estimator where an inefficient, consistent MM estimator $\hat{\beta}^{MM}(I_q)$ is obtained by setting $\Wb = I_q$ in equation~\eqref{def:MM_optimization}. $\hat{\Sb}$ is then derived based on $\hat{\beta}^{MM}(I_q)$, which then is used to construct a new weighting matrix for the final efficient MM estimator:
\begin{itemize}
    \item Compute a first consistent estimator $\tilde{\beta} := \hat{\beta}^{MM}(I_{q})$.
    \item Compute a consistent estimator $\hat{\Sb} :=  \text{Diag}\Big(\rho_{e_c}\widehat{\text{Var}}(R^{c})\hat{\sigma}^{e_c,2}\Big)_{c\in \mathcal{C}}$ of the weighting matrix, where \mbox{$\hat{\sigma}^{e_c,2} = \frac{1}{n_{e_c}}\sum_{i\in [n_{e_c}]}\big(Y^{e_c}_i - \tilde{\beta}^T\Xb^{e_c}_i\big)^2$}. 
    \item Return the two-step estimator $\hat{\beta}^{MM*} = \hat{\beta}^{MM}(\hat{\Sb}^{-1})$
\end{itemize}
If the data are i.i.d.\ within each environment with finite second moments, and $\tilde{\beta}$ is consistent, then $\hat{\sigma}^{e_c,2}$ and $\widehat{\text{Var}}(R^{c})$ are consistent. Thus, $\hat \Sb$ is consistent for $\Sb$ and $\hat{\beta}^{MM*}$ has optimal asymptotic variance. We summarize the above statements in the following proposition:
\begin{proposition}[\cite{hall2005generalized} Chapter 2]\label{proposition:hall2005generalized}
Assume that $\beta^0$ satisfies the constraints defined by equation~\eqref{equation:population-causal-aggregation-estimator}, that $\Gb$ is rank $p$, and that data samples are i.i.d. within each environment, with finite second moments so that any MM estimator is consistent. Additionally, assume that the plug-in estimator $\hat{\Sb}$ of $\Sb$ is consistent. The two-step MM estimator $\hat{\beta}^{MM*}$ based on our set of orthogonality constraints satisfies the following asymptotic limit
\begin{align*}
    \sqrt{n}\big(\hat{\beta}^{MM*} - \beta^0\big) \xrightarrow[]{} \mathcal{N}\Big(\mathbf{0}; \Sigma\Big)
\end{align*}
where convergence is in distribution, and where we get the optimal asymptotic covariance among all choices of weighting matrices: $\Sigma = \big(\Gb^T\text{Diag}\Big(\frac{\rho_{e_c}}{\text{Var}(R^{c})\sigma_{e_c}^2}\Big)_{c\in \mathcal{C}}\Gb\big)^{-1}$ 
\end{proposition}
We recover here the asymptotic covariance from Proposition~\ref{proposition:asymptotic_normality} whenever $\Gb$ is a square, invertible matrix. Also, this result automatically shows that adding constraints can not hurt the performance of the estimator. Given any subset of $p$ constraints, we can construct a MM estimator of the form \eqref{def:estimator_just_identified} by setting to $0$ the appropriate elements of $\Wb$ in \eqref{def:MM_optimization}. Therefore the asymptotic efficiency of $\hat{\beta}^{MM*}$ implies that the new constraints, properly weighted by $\Sb$, improve over the just-identified case. Incorporating all such information leads to more efficient estimators: all causal information helps. Conversely, in this over-identified setting one can wonder whether certain inconsistent constraints may be hurting the estimator performance: we refer to \cite{hall2005generalized} for further discussions on tests to detect such issues.

The following two sections address the high-dimensional and non-linear cases respectively, building upon the notation and theory presented thus far. We first present in the following section the high-dimensional case where additional assumptions are required to identify $\beta^0$, and we propose an estimator that favors sparse solutions $\hat{\beta}$ and converges to $\beta^0$. Within this framework, under additional assumptions, we can recover $\beta^0$ even in the under-identified case. Unfortunately these are not always satisfied, in which case a pre-screening step may be used to bypass this issue. The linear aggregation procedure in the just-identified case becomes a central sub-routine in the non-linear causal aggregation framework, which is addressed in the subsequent section. The underlying intuition is the same: based on the available additional knowledge about how environments are generated, we build estimators by enforcing constraints that represent the orthogonality between an unconfounded variable and the residuals.

\section{High-Dimensional Linear Aggregation under Sparsity Assumptions}\label{section:high-dimensional}

In high-dimensional settings we may not have enough constraints to construct an estimator via the methods presented above. Even with enough samples per environment, $\beta^0$ may not be identifiable based only on the orthogonality constraints. Conversely, we may have access to a large number of environments, enough for identifying $\beta^0$ at the population level, but containing very few samples in each. In the following section, we develop an estimator for such settings. Under additional assumptions on the structure of $\beta^0$, it may be possible to aggregate such causal information and obtain a reasonable estimator of $\beta^0$. Regularization methods based on leveraging the geometry induced by the $\ell_1$-norm are helpful to overcome this issue under the assumption that the actual vector $\beta^0$ is sparse. Among these very popular techniques we mention the Lasso penalty for linear regression \citep{tibshirani1996regression} and basis pursuit \citep{chen2001atomic}. These methods were subsequently adapted to address multiple other problems based on sparsity assumptions (eg. precision matrix estimation in \cite{friedman2008sparse}), deal with additional structure in the regressors (eg. when the order of the covariates matters as in the fused lasso in \citet{tibshirani2005sparsity}, where differences between consecutive coefficients are penalized) or expanded with additional penalties (eg. \citet{zou2005regularization} combine Lasso and $\ell_2$-norm penalties). Whenever we have experimental data for every covariate---even if very few samples per experiment---, we derive in Section~\ref{section:subsection:dantzig} an estimator based on these regularization techniques by directly solving a constrained risk minimization problem. We discuss this in more detail after Proposition~\ref{proposition:support_recovery}. Alternatively, in Section~\ref{section:subsection:pre-screening} we propose running a two-step procedure where a subset of covariates is first selected based on the Lasso regression, and then the proposed aggregation procedure is run on the subset of covariates. 

\subsection{Estimation by Constrained Optimization}\label{section:subsection:dantzig}
We present an estimator that mirrors the formulation of the Dantzig Selector \citep{candes2007dantzig} designed for the problem of high-dimensional linear regression $y = \xb^T\beta^0 + \epsilon$, where the dimension of $\xb$ is larger than the number of available observations. The Dantzig Selector is the solution to the following problem:
\begin{align*}
    \text{minimize}\; \Vert\beta\Vert_1 \qquad \text{subject to} \; \frac{1}{n}\Vert\Xb^T(Y - \Xb\beta)\Vert_{\infty}\leq \lambda
\end{align*}
where $Y = (y_1,\dots, y_n)$, and $\Xb = (x_1,\dots,x_n)^T$ is the design matrix. \cite{candes2007dantzig} derived a probabilistic upper bound for the $\ell_2$ error of the estimator under some conditions on $\Xb$. Additional work by \cite{bickel2009simultaneous}, \cite{ye2010rate}, among others lead to sharper bounds. We now present our $\ell_1$-norm minimization based causal aggregation estimator, and we then follow \cite{ye2010rate} to derive upper bounds on the $\ell_q$ loss. We follow the exposition in \cite{rothenhausler2019causal} and adapt it to our setting, leaving all proofs to the Appendix. Our alternative formulation promotes sparsity by solving an $\ell_1$-norm minimization problem. We denote it by $\hat{\beta}(\lambda)$ where $\lambda>0$ is a hyper-parameter, and the definition mirrors the above minimization problem:
\begin{align}\label{equation:high-dim}
    \text{minimize}\; \Vert\beta\Vert_1 \qquad \text{subject to} \; \Vert\hat{\Zb} - \hat{\Gb}\beta\Vert_{\infty}\leq \lambda
\end{align}
where $\hat{\Zb}$ and $\hat{\Gb}$ are defined as in equation~\eqref{equation:empirical-GZ}. This is a convex problem, in particular a linear programming problem: Let $B^\lambda$ be the feasible set to \eqref{equation:high-dim} and define $\Gamma^\lambda$ as the feasible set to the linear programming problem
\begin{align*}
    &\text{minimize } \mathbf{c}^T \gamma \\
   & \text{subject to } \mathbf{A} \gamma \le \mathbf{b} \text{ and } \gamma \ge 0, \text{ where } \\
    &\mathbf{A} = \begin{pmatrix}
    - \hat{\Gb} & \hat{\Gb} \\
    \hat{\Gb} & - \hat{\Gb}
    \end{pmatrix};  \mathbf{b} = \begin{pmatrix}
    - \hat{\Zb} \\ \hat{\Zb}
    \end{pmatrix} + \begin{pmatrix}
    \lambda \\
    \lambda \\
    \vdots \\
    \lambda
    \end{pmatrix}; \text{ and } \mathbf{c} = \begin{pmatrix}
    1 \\
    1 \\
    \vdots \\
    1
    \end{pmatrix}.
\end{align*}
Then, as shown in Lemma 5 in \cite{rothenhausler2019causal}, $B^\lambda = \{  \gamma_{1:p} - \gamma_{(p+1):2p} : \gamma \in \Gamma^\lambda \}$. Assimilating $\gamma_{1:p}$ to the vector $\beta_+$, the positive part of the $\beta$ coefficients --analogously for $\gamma_{p+1:2p}$ and $\beta_-$, we get the equivalence.

Higher values of $\lambda$ relax the constraint based on the causal orthogonality constraints, leading to solutions $\hat{\beta}(\lambda)$ with smaller $\ell_1$ norm. Because of the geometry of the $\ell_1$-norm unit ball, sparser solutions are favored. We can prove bounds on the $\ell_q$ loss of our estimator with high probability. To this end, we will impose assumptions on $\hat{\Gb}$. Following \cite{ye2010rate}, we define the cone invertibility factor (CIF) as follows: for $1\leq q \leq +\infty$, $J\subset[p]$, and $J^{c} := [p]\setminus J$, define
\begin{align*}
    \text{CIF}_q(J, \Mb) := \inf\Big\{ \frac{|J|^{1/q}\Vert\Mb \ub\Vert_{\infty}}{\Vert \ub\Vert_q} ; \;\ub\in \mathbb{R}^p\setminus \{0\}; \; \Vert \ub_{J^c}\Vert_1 \leq \Vert \ub_{J}\Vert_1\Big\}
\end{align*}
where $\Mb \in \mathbb{R}^{\tilde{p}\times p}$ is not necessarily a square matrix, and we set by convention $|J|^{1/\infty} = 1$. \cite{ye2010rate} show how this quantity plays a similar role and relates to the sparse eigenvalue condition, where generally the matrix $\Mb$ is an estimate of the covariance matrix of $\Xb$. Intuitively, estimation in high dimensions is difficult as $\hat \Gb$ is not invertible. However, under a sparsity assumption on the coefficients, one only needs $\hat \Gb$ to be invertible on the set of sparse vectors. The cone invertibility factor (CIF) captures whether $\hat \Gb$ is non-invertible on the set of sparse vectors. The proof technique then proceeds by showing that the cone invertibility factor for $\hat{\mathbf{G}}$ is close to the cone invertibility factor for the population matrix $\mathbf{G}$, which is assumed to be invertible. As we will show, estimation of $\beta^0$ is possible in this setting and we can control the $\ell_q$ error of the estimator $\hat{\beta}(\lambda)$. This derives from an upper bound for $\Vert\hat{\beta}(\lambda) - \beta^0\Vert_q$ by a ratio with the CIF in the denominator as a critical quantity that must be positive. 

One may ask whether under a sparsity assumption it is actually necessary to have one constraint per covariate, or whether it is sufficient to have far fewer constraints than covariates. The following example gives a negative answer to this question.  Intuitively speaking, sparsity assumptions allows us to get away with few observations per constraint, but we still need at least as many constraints as covariates. This issue might be mitigated by pre-screening, which will be discussed further below. 
More specifically, Figure~\ref{fig:Ex3} provides an example where $\ell_1$-regularization leads to the wrong solution, even for $n \rightarrow \infty$. Consider two separate models where the only difference is in the structural equation of $X_2$ which has no effect on $Y$. The true causal parameter is $\beta^0=(1,0)$. Assuming $X_1$ is intervened on both settings, we obtain an estimator solving equation~\eqref{equation:high-dim} on each separate environment, both satisfying their corresponding orthogonality constraints. Even though $\beta^0$ also satisfies the orthogonality constraints in both cases, the estimator on the right-hand side model is a solution with smaller $\ell_1$ norm than $\beta^0$, and thus is not consistent. The minimal $\ell_1$ norm solution to the linear constraint for the left-hand side model is however the true $\beta^0$.
\begin{figure*}[ht]
\centering
\includegraphics[trim={1cm 9cm 9cm 6cm},clip,width=0.6\linewidth]{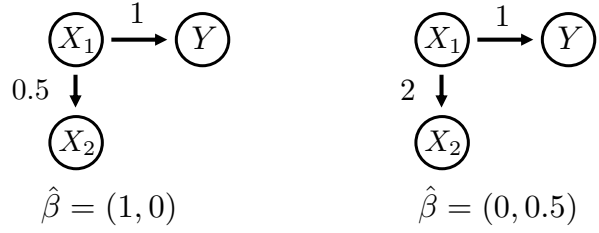}
\caption{Regularization by shrinking the $\ell_1$ norm of the coefficient vector requires additional structural assumptions. In this example, both $\beta^{(1)} = (1,0)$ and $\beta^{(2)}=(0,.5)$ satisfy $\Gb \beta = \Zb$ in their respective environments. Since $ \| \beta^{(1)} \|_1 > \| \beta^{(2)} \|_1$, the proposed method would estimate $\beta = \beta^{(2)}$ in the population case for the right-hand model, which is not the correct solution. This issue can be mitigated by pre-screeing, as we will discuss below.}
\label{fig:Ex3}
\end{figure*}

We now refer back to the two scenarios briefly described at the beginning of this section. Consider environments generated via randomization. The population level matrix $\Gb$ is invertible following Proposition~\ref{proposition:linear_identifiability} whenever we have access to a very large number of environments, such that for any covariate there is an environment where it is randomized. In this case, even if we only have few samples per environment, our high-dimensional causal aggregation procedure will be of practical use as the CIF is positive. On the other hand, whenever the number of constraints is small,  
we will assume that the entries in the connectivity matrix between covariates are small enough so that the matrix $\Gb$ is invertible on the set of sparse vectors with same support as $\beta^0$, leading to a positive CIF. This assumption is not verifiable in practice, and therefore we recommend pre-screening in these scenarios which we discuss in Section~\ref{section:subsection:pre-screening}.

We thus now assume that the CIF value is positive, and in particular does not decrease too fast in the high-dimensional regime where $p, n_e$ simultaneously grow. We formalize our main result in the following proposition.

\begin{proposition}\label{proposition:lq-bound}
Denote by $S^0 := \{j: \beta^0_j\neq 0\}$ the active set of covariates. Assume that $X^{e}_j$ are $\sigma_X^2$ sub-Gaussian, $\epsilon_Y^e$ are $\sigma_E^2$ sub-Gaussian, and that $R^c$ are $\sigma^{2}_C$ sub-Gaussian for all $e\in \mathcal{E}, j \in [p], c\in \mathcal{C}$ and some fixed $\sigma_X^2, \sigma_E^2, \sigma_C^2>0$. Additionally, assume that 
\begin{align*}
    \frac{1}{\text{CIF}_q(S^0, \Gb)}\sqrt{\frac{\log p}{\min_{e\in \mathcal{E}}n_e}} \xrightarrow[\{n_e\}_e,p \rightarrow +\infty]{} 0
\end{align*}
so that, in particular, $\text{CIF}_q(S^0, \Gb) > 0$. There exists a constant $K>0$, that depends only on $\sigma_X^2, \sigma_E^2, \sigma_C^2>0$ and $K_0$, another universal constant, such that, for the following choice of $\lambda$:
\begin{equation*} 
\lambda := K\sqrt{\frac{\log p}{\min_{e\in \mathcal{E}}n_e}}
\end{equation*}
we get
\begin{align*}
    \mathbb{P}\Bigg(\Vert\hat{\beta}(\lambda) - \beta^0 \Vert_{q} \leq \frac{K|S^0|^{1/q}}{\text{CIF}_q(S^0, \Gb)}\sqrt{\frac{\log p}{\min_{e\in \mathcal{E}}n_e}}\Bigg) \xrightarrow[\{n_e\}_e,p \rightarrow +\infty]{} 1
\end{align*}
\end{proposition}
Our assumptions require in particular that the dimension $p$ does not grow too fast compared to the number of samples in the environments:
\begin{equation*}
    \frac{\log p}{\min_{e\in \mathcal{E}}n_e} \xrightarrow[\{n_e\}_e,p \rightarrow +\infty]{} 0
\end{equation*}
The estimator $\hat{\beta}(\lambda)$ is sparse and we denote $\hat{S}(\lambda):= \{j: \hat{\beta}(\lambda)\neq 0\}$ its active set. Our estimation procedure leads to a feature selection procedure as the following result holds under an additional assumption on the minimum value of the non-zero coordinates of $\beta^0$, called beta-min assumption \cite{buhlmann2013statistical}. This condition is required for support recovery as it provides the required separation between the non-zero coordinates of $\beta^0$ and the null vector with respect to the CIF-derived upper bound of the $\ell_q$ loss. Note that support recovery is different from model selection consistency, which usually needs much stronger assumptions \citep{zhao2006model}.
\begin{proposition}\label{proposition:support_recovery}
Assume that the conditions of the result above hold with $q = \infty$, and that 
\begin{align*}
    \min_{j\in S^0} |\beta^0_j| > \frac{K}{\text{CIF}_{\infty}(S^0, \Gb)}\sqrt{\frac{\log p}{\min_{e\in \mathcal{E}}n_e}}
\end{align*}
We then have
\begin{align*}
    \mathbb{P}\big(S^0\subset \hat{S}(\lambda)\big)\xrightarrow[\{n_e\}_e,p \rightarrow +\infty]{} 1
\end{align*}
\end{proposition}
In conclusion, our regularized estimator recovers the support of $\beta^0$ and converges in probability to $\beta^0$ under the $\ell_q$ norm under the assumption that the CIF value is positive. This holds whenever we have a very large number of environments where every covariate is randomized in at least one environment, even if we have few samples per environment. In cases where the number of randomized constraints is too small this assumption may not hold in practice. In this setting, we recommend the pre-screening estimation procedure described below.

\subsection{Estimation via Pre-Screening}\label{section:subsection:pre-screening}

Given the limitations of the direct shrinkage approach, which requires few interventional samples per environment, but a large number of experiments, we propose an alternative formulation for estimation of $\beta^0$ in the high-dimensional setting based on a two-step procedure. However, we instead assume that some observational data is available: we can start with a pre-screening step that chooses a subset of covariates by running a Lasso regression of $Y$ on $X$ on the observational data set. Under regularity assumptions, the set $\hat{S}$ of covariates selected by running a Lasso regression on observational data contains the Markov blanket of the response variable $Y$: in particular, $\mathbb{P}[\hat S \supseteq S^0 ] \rightarrow 1$ under some conditions on the non-zero regression coefficients \citep[Section 2.5]{buhlmann2011statistics}. This procedure assumes that we are able to choose covariates to intervene on: randomizing variables in $\hat{S}$ in (potentially) multiple environments generates orthogonality constraints that allow us to estimate $\beta^0$. The estimator is computed on a different dataset from the one used to pre-select covariates, which guarantees asymptotically valid confidence intervals if $\hat S$ contains the Markov blanket. We report in Appendix~\ref{section:exp-high-dimensional} simulations based on synthetic data to validate our procedure based on a pre-screening step.

\section{Non-Linear Causal Aggregation}\label{section:non-linear}

Here we extend our causal aggregation procedure beyond the linear case. This methodology allows us to recover non-linear interaction terms between covariates if we have access to environments where those covariates are simultaneously randomized. If we assume that the response variable model has no such interactions, environments where a single variable is randomized still allow for estimation of non-linear responses. Compared to previous sections, in this section we focus on causal constraints that arise through randomization. 

We first define in Section~\ref{section:subsection:nonlinear_definition} a non-linear SEM extension, and then characterize in Section~\ref{section:subsection:nonlinear_identification} sufficient conditions on the set of environments that allow us to identify the non-linear function of the response variable structural equation. Finally, in Section~\ref{section:subsection:causal_boosting} we propose a causal aggregation procedure inspired from the Boosting methodology.

\subsection{Non-linear Structural Equation Models}\label{section:subsection:nonlinear_definition}

Estimating causal effects is more challenging in presence of interactions, and linear approximations may not capture true causal relationships. We define a non-linear extension of the SEM in \eqref{def:linearSEM} via the following causal model $\mathcal{M}^0$:
\begin{equation}\label{equation:SEM}
\mathcal{M}^0 : \begin{cases}
    X_j \longleftarrow \mathfrak{f}_j(\Xb_{\text{pa}_0(j)}) + \epsilon_j \qquad \text{for all}\; 1\leq j \leq p
    \\ Y \longleftarrow f^0(\Xb_{pa_0(p+1)}) +  \epsilon_{Y}
\end{cases}
\end{equation}
 where we define structural equations as real-valued functions $\mathfrak{f}_j$ for $1\leq j \leq p$ over a subset of covariates indexed by $\text{pa}_0(j) \subset [p+1]$. The structural equations from the causal model $\mathcal{M}^0$ given by \eqref{equation:SEM} define a DAG $G^0=(V^0, E^0)$ over the $p+1$ nodes denoted by $V^0$. As in the linear case, we assume that $G^0$ is a DAG and that the response variable $Y$ can belong to the parental sets of covariate nodes. Unlike the linear case, we do not make explicit the confounding effect via latent variables, instead allowing disturbance terms $\epsb = \{\epsilon_i\}_{i\in [p+1]}$ to be dependent. Therefore, due to such confounding the observational distribution $\mathbb{P}^0$ factorizes in an extended graph $\bar{G}^0=(V^0, \bar{E}^0)$ with additional edges in $\bar{E}^0$. Conditional independence statements based only on the structure of $G^0$ may not hold if a latent variable simultaneously influences several nodes. The function $f^0$ represents the causal effect of the covariates on the response, which takes as argument a subset $S^0:=pa_0(p+1)\subset[p]$ of covariates and potentially contains non-linear effects of covariates on the response as well as interaction effects between covariates. We assume that all functions $(\mathfrak{f}_j)_j , f^0$ are square integrable under any environment distribution $\mathbb{P}^e$. The function $f^0$ represents the unconfounded relationship between covariates and response: under an interventional distribution $\mathbb{P}^e$ where $\Xb_{S^0}^e$ are randomized, we assume that $\mathbb{E}[\epsilon_Y^e | \Xb_{S^0}^e] = 0$. The model for generating distributions $\mathbb{P}^e$ from interventions on covariates indexed by $\phi\subset[p]$ is defined analogously to \eqref{equation:perturbedSEM} by:
\begin{align}\label{equation:perturbed-nonlinear-SEM}
\mathcal{M}^e_{\phi(e)} : \begin{cases}
    X_j^e \longleftarrow \mathfrak{f}_j^e(\Xb^e_{\text{pa}_e(j)}) + \epsilon^e_j & \forall j \notin \phi(e)
    \\ X_j^e \longleftarrow \epsilon_j^e & \forall j \in \phi(e)
    \\ Y^e \longleftarrow f^0(\Xb^e_{S^0}) + \epsilon^e_Y
\end{cases}
\end{align}
where randomized covariate disturbance terms $(\epsilon_j^e)_{j\in \phi}$ are jointly independent, independent of  non-randomized covariate ones $(\epsilon_j^e)_{j\notin \phi}$. As in the linear case, an interventional causal model $\mathcal{M}^e_{\phi(e)}$ still factors in a simplified DAG $\bar{G}^e=(V^0, \bar{E}^e)$ where incoming edges into $\Xb_{\phi(e)}$ nodes are deleted. Our proposed non-linear aggregation procedure relies on the structure of $\bar{G}^e$.

\subsection{Identification of the Response Structural Equation}\label{section:subsection:nonlinear_identification}

Our previous approach for aggregating information across environments does not carry on to this scenario as we can no longer stack vectors representing linear orthogonality constraints into one system of equations that aggregates the causal information derived from each environment. Similar to the linear case, estimating $f^0$ is possible whenever all the covariates in the parental set of $Y$ are randomized: in environment $e$ where $S^0\subset\phi(e)$, we have $\mathbb{E}[Y^e | \Xb_{\phi(e)}^e] = f^0(\Xb_{S^0}^e) + \mathbb{E}[\epsilon_Y^e | \Xb_{\phi(e)}^e] = f^0(\Xb_{S^0}^e)$. Any non-parametric regression method can be used to estimate $f^0$ in this setting, where non-randomized covariates are ignored. However, as we indicated in the introduction, simultaneous randomization may not be feasible in practice, so that instead of one fully randomized experiment we may only have access to several datasets where different subsets of covariates are randomized. Our method for constructing an estimator of $f^0$ still relies on orthogonality constraints that we obtain environment-wise and then aggregate into a single estimator $\hat{f}$. In the example above, denoting by $\sigma(U)$ the Borel $\sigma$-algebra generated by the random variable $U$, we get that $f^0(\Xb^e_{S^0})$ is characterized as the orthogonal projection of $Y^e$ on the subspace $L_2(\sigma(\Xb^e_{\phi(e)}))$ of square-integrable $\sigma(\Xb^e_{\phi(e)})$-measurable random variables: for any square-integrable Borel function $h$, we have
\begin{align}\label{eq:orthogonality-general}
    \mathbb{E}\big[h(\Xb^e_{\phi(e)})(Y^{e} - f^0(\Xb^e_{S^0}))\big] = 0
\end{align}
Conversely, this projection does not capture the true $f^0$ if some covariates in $S^0$ are confounded: the $L_2$ projection by conditioning over $\Xb_{S^0}$ leads to a biased estimate of $f^0$. However, under additional assumptions on $f^0$ and the set of causal models $(\mathcal{M}^e_{\phi(e)})_e$, exact recovery of $f^0$ is still possible based on multiple environments, where only a few covariates in $S^0$ are simultaneously randomized in each environment. Our first result shows that, given a collection of environments $(\mathcal{M}^e_{\phi(e)})_e$, assuming $f^0$ can be decomposed as follows:
\begin{equation*}
    f^0(\xb_{S^0}) := \sum_{e \in \mathcal{E}}f_e(\xb_{\phi(e)\cap S^0})
\end{equation*}
then we can identify $f^0$. We do not know $S^0$ in practice, and thus we assume that $f^0$ is a sum of functions representing each an interaction term between subsets of covariates in $S^0$ which are simultaneously randomized in a given environment: 
\begin{equation}\label{decomposition_function}
    f^0(\xb_{S^0}) = \sum_{e \in \mathcal{E}}f_e(\xb_{\phi(e)})
\end{equation}
Let $\mathcal{F}$ be the set of square-integrable Borel functions for all $\mathbb{P}^e$, and $\mathcal{F}_{\mathcal{E}} \subset \mathcal{F}$ the set of functions $f$ that can be decomposed as above \eqref{decomposition_function}, where $f_e: \mathbb{R}^{|\phi(e)|}\rightarrow\mathbb{R}$ and $f^e \in \mathcal{F}$ for all $e\in \mathcal{E}$.
\begin{proposition}\label{proposition:identifiability}
Assume that the distribution $\mathbb{P}^e$ defining environment $e\in \mathcal{E}$ is generated through interventions following model~\eqref{equation:perturbed-nonlinear-SEM}, and that all $\mathbb{P}^e$ have the same support. Then there exists at most one function $\bar{f} \in \mathcal{F}_{\mathcal{E}}$ such that, for all $e\in \mathcal{E}$, for all $h\in \mathcal{F}$, we have 
\begin{equation}\label{orthogonality_intervention}
    0 = \mathbb{E}[h(\Xb^e_{\phi(e)})(Y^e - \bar{f}(\Xb^e))]
\end{equation}
Additionally, $f^0$ satisfies equation~\eqref{orthogonality_intervention} above over the shared support of $(\mathbb{P}^e)_e$ whenever $f^0 \in \mathcal{F}_{\mathcal{E}}$.
\end{proposition}
This effectively corresponds to projecting the response $Y$ on a smaller subspace of random variables. However, this projection is only on the unconfounded covariates within each environment: the objective function varies between environments when fitting any regression model to estimate such projections.

\subsection{Causal Aggregation via Boosting}\label{section:subsection:causal_boosting}

Based on the available environments, we propose a causal aggregation estimator $\hat{f} \in \mathcal{F}_{\mathcal{E}}$ of $f^0$ following decomposition in equation~\eqref{decomposition_function}:
\begin{equation*}
    \hat{f} = \sum_{e \in \mathcal{E}}\hat{f}_e(\xb_{\phi(e)})
\end{equation*}
by individually estimating each term $\hat{f}_e$ via samples from the corresponding environment---specifically, the randomized covariates of those samples. Our ability to recover $f^0$ is therefore limited by the overall ``causal information'' derived from each environment: the more covariates are simultaneously randomized in some experiment, the better our procedure will be at capturing potential interactions. The identification result can inspire a naive fitting procedure that closely resembles the backfitting algorithm for additive models \citep{breiman1985estimating}, but uses samples from environment $e$ to estimate the different additive terms $\hat{f}_e$. We first initialize all the additive term estimators $\hat{f}_e$ (for example, setting these to $0$). We then randomly pick an environment $e$, compute the residual term given by 
\begin{align*}
    R^e := Y^e - \sum_{\tilde{e} \neq e}\hat{f}_{\tilde{e}}(\Xb_{\phi(\tilde{e})}^e)
\end{align*}
and update $\hat{f}_e$ based on the orthogonality constraint for that environment where $h$ is any square-integrable random variable:
\begin{align}\label{equation:orthogonality-residual}
    0 = \mathbb{E}\Big[h(\Xb^e_{\phi(e)})\big(R^e - \hat{f}_e(\Xb_{\phi(e)}^e)\big)\Big]
\end{align}
This in turn corresponds to choosing $\hat{f}_e$ as the minimizer of 
\begin{align}\label{equation:orthogonality-loss}
    \hat{f}_e = \argmin_{f_e} \mathbb{E}\Big[\big(R^e - f_e(\Xb_{\phi(e)}^e)\big)^2\Big]
\end{align}
We then loop until some convergence criterion is attained, for example, taking the supremum over $e$ of an $\ell_{\infty}$ loss on the difference between the two last most recent updates of $\hat{f}_e$. We describe this procedure in Algorithm~\eqref{algorithm:backfitting}. The connection between this procedure and the proof of the identification result in Proposition~\ref{proposition:identifiability} is the following. Assume the sequential training of the component $\hat{f}_e$ in $\hat{f}$ starts with the environment index $e$ whose input variables are the most downstream in the DAG derived from $\mathcal{M}^0$. Then this estimator is set to the true population $f_e$ when fitted on the environment $e$ if the model is well specified---similar to the observation leading to \eqref{eq:orthogonality-general} in the fully randomized case. Iterating this procedure we recover $f^0$ by sequentially cancelling out the contribution of each $f_e$ in $Y^e$ when computing \eqref{equation:orthogonality-loss}. Unfortunately, this procedure would require knowing the graph $G^0$, and other challenges arise when using this method that showcase the differences with respect to the original backfitting algorithm. In particular, the loss~\eqref{equation:orthogonality-loss} is minimized on different datasets. Therefore, the sequence of loss values for subsequent iterations of our estimator on one given environment is not decreasing. After updating any given term in decomposition~\eqref{decomposition_function} based on minimizing the loss over the corresponding environment, the loss evaluated at another environment may increase instead.
\begin{algorithm}
    Initialize $\hat{f}_e = 0$, set $\delta_0>0$ convergence threshold, $\delta = 2*\delta_0$ current update gap.
    \\ \While{$\delta > \delta_0$}{
    Sample uniformly: $e \sim \mathcal{U}(|\mathcal{E}|)$
    \\Compute the residual: $R^e \xleftarrow{} Y^e - \sum_{\tilde{e} \neq e}\hat{f}_{\tilde{e}}(\Xb_{\phi(\tilde{e})}^e)$
    \\ Compute estimator over environment: $\hat{g}_e \xleftarrow{} \argmin_{f_e \in \mathcal{F}_{\mathcal{E}}} \mathbb{E}\Big[\big(R^e - f_e(\Xb_{\phi(e)}^e)\big)^2\Big]$
    \\ Update gap: $\delta \xleftarrow{} \sup_e \sup_{\Xb^e}|\hat{f}_e(\Xb^e) - \hat{g}_e(\Xb^e)|$
    \\ Update estimator component: $\hat{f}_e \xleftarrow{} \hat{g}_e$\label{eq:replace-component}
    }
    \Return $\hat{f} = \sum_{e\in \mathcal{E}}\hat{f}_e$
    \caption{Non-linear Causal Aggregation via Backfitting}\label{algorithm:backfitting}
\end{algorithm}

This algorithm then suffers from unstable training, in particular when fitting models with small sample sizes. These issues have the same underlying source: exactly enforcing the orthogonality constraint on one environment to update one component can in turn make the orthogonality conditions fail on the remaining components.

We propose a procedure that addresses this issue, called non-linear causal aggregation via boosting. Boosting is a greedy algorithm where weak learners are aggregated to form a regression estimator or a classifier. Weak learners correspond to individual simple estimators derived from some base procedure such as regression or decision trees. Initial boosting formulations for classification sequentially fit weak learners to samples re-weighted in such way that the mis-classified ones by previously generated weak learners get higher weights in the next iteration. The procedure thus increasingly focuses on the ``hard'' training samples  \citep{FREUND1995256,FREUND1997119, schapire1990strength}. The final output is then a linear weighted average of the weak learners. In \citet{friedman2000additive}, boosting is reformulated as a sequential additive modeling procedure where fitted weak learners correspond to the negative gradient of the classification loss in the function space. Weak learners are then sequentially added to the current function estimator. This characterization was first observed in \citet{breiman1998arcing, breiman1999prediction}, but using this additive modeling perspective boosting methods were subsequently developed beyond classification problems, in particular for regression problems \citep{friedman2001greedy}. At each iteration a weak learner is trained on top of the current iteration of the full estimator: using the square loss, this corresponds to training the weak learner on the residual term obtained based on the current regression model \citep{buhlmann2003boosting}. Within our framework, instead of replacing altogether the individual component $\hat{f}_e$ with an updated fitted model trained on environment $e$ alone, we incrementally update all components $\hat{f}_e$ simultaneously by adding weighted function estimates $\hat{h}_e$ defined over the same set of covariates $\xb_{\phi(e)}$ as the corresponding $\hat{f}_e$. This can also be understood as aggregating ``weakly informative datasets'', as each boosting step only takes into account for each sample those covariates that are unconfounded. Those $\hat{h}_e$ are first estimated independently: within each environment, we use the current aggregate function $\hat{f}$ to generate the environment residuals
\begin{align*}
    R^e := Y^e - \sum_{\tilde{e}}\hat{f}_{\tilde{e}}(\Xb_{\phi(\tilde{e})}^e) = R^e - \hat{f}(\Xb^e)
\end{align*}
Then, independently for each environment, we estimate $\hat{h}_e$ by solving a minimization problem as in \eqref{equation:orthogonality-loss} based on samples from the corresponding environment $e$. We can view this in light of the gradient boosting methodology \citep{friedman2001elements}. We similarly fit a set of candidate learners to the residuals derived from a quadratic loss function. However, individual components $\hat{h}_e$ use only information pertaining to the randomized covariates in each individual data set. Recent work has developed the boosting framework for non-linear instrumental variables regression \citep{bakhitov2021causal}. Instead of using instruments from a single data set, our procedure deals with heterogeneous datasets with different distributions generated via different experimental settings. The second step in this procedure scales the candidate $(\hat{h}_e)_{e\in \mathcal{E}}$ by weights constructed to simultaneously satisfy all orthogonality constraints. That is, we update each component as follows:
\begin{align*}
    \hat{f}_e \xleftarrow{} \hat{f}_e + \hat{\alpha}_e\hat{h}_e
\end{align*}
where $(\hat{\alpha}_e)_{e\in \mathcal{E}}$ are chosen to satisfy the constraints \eqref{equation:orthogonality-residual} for the corresponding aggregated $\sum_{e} \hat{f}_e + \hat{\alpha}_e\hat{h}_e$:
\begin{equation*}
    0 = \mathbb{E}\left[\hat{h}_e(\Xb^e_{\phi(e)})\Big(Y^e - \big(\sum_{\tilde{e}}\hat{f}_{\tilde{e}}(\Xb^e_{\phi(\tilde{e})})+\hat{\alpha}_e\hat{h}_{\tilde{e}}(\Xb^e_{\phi(\tilde{e})}\big)\Big)\right] \qquad \forall e\in \mathcal{E}
\end{equation*}
One choice of loss function to implement this weight parameter search leads to the convex minimization problem: 
\begin{align}\label{eq:linear-subroutine}
    (\hat{\alpha}_e)_{e\in \mathcal{E}} := \argmin_{\alpha_e \in \mathbb{R}} \sum_{e\in \mathcal{E}}\left|\mathbb{E}\left[\big(R^e-\sum_{\tilde{e}\in \mathcal{E}}\alpha_{\tilde{e}}\hat{h}_{\tilde{e}}(\Xb^e_{\phi(\tilde{e})})\big)\hat{h}_e(\Xb^e_{\phi(e)})\right]\right| + \nu\Vert\alpha\Vert_2^2
\end{align}
where we regularize the weight vector $\alpha$ with a quadratic penalty scaled by $\nu\geq 0$ that we empirically observed improves the training behavior. This step is key to guarantee the stability of the procedure. Although individual $\hat{h}_e$ are fitted through individual environment $e$ samples, the global re-weighting over the joint data across environments prevents that individual updates by a single $\hat{h}_e$ perturb the orthogonality constraints on environments $\tilde{e}\neq e$. Finally, we empirically observe that shrinking the updates by a learning rate $1>\eta>0$ improves convergence, so our final proposed update is given by:
\begin{equation*}
    \hat{f} \xleftarrow{} \hat{f} + \eta\sum_{\tilde{e}\in \mathcal{E}}\hat{\alpha}_{\tilde{e}}\hat{h}_{\tilde{e}}
\end{equation*}
We again stop the training procedure as soon as the new term $\sum_{\tilde{e}\in \mathcal{E}}\hat{\alpha}_{\tilde{e}}\hat{h}_{\tilde{e}}$ no longer substantially updates the aggregated $\hat{f}$, and hyper-parameters such as the learning rate and the penalty weight $\nu$ can be chosen by keeping a separate validation data set within each environment and evaluating the unpenalized orthogonality constraint loss in the minimization problem \eqref{eq:linear-subroutine}. 
\begin{algorithm}
    Initialize $\hat{f}_e = 0$, $\hat{f} = \sum_e \hat{f}_e$, set $\eta$ learning rate, $\nu$ penalty weight, set $\delta_0>0$ convergence threshold, $\delta = 2*\delta_0$ current update gap.
    \\ \While{$\delta > \delta_0$}{
    \For{$e \in \mathcal{E}$}{
    Compute residual: $R^e \xleftarrow{} Y^e - \hat{f}(\Xb^e)$
    \\ Compute individual component over environment: $\hat{h}_e \xleftarrow{} \argmin_{h_e \in \mathcal{F}_{\mathcal{E}}} \mathbb{E}\left[\big(R^e - h_e(\Xb_{\phi(e)}^e)\big)^2\right]$
        }
    Enforce orthogonality constraints: $(\hat{\alpha}_e)_{e\in \mathcal{E}} := \argmin_{\alpha_e \in \mathbb{R}} \sum_{e\in \mathcal{E}}\left|\mathbb{E}\left[\big(R^e-\sum_{\tilde{e}\in \mathcal{E}}\alpha_{\tilde{e}}\hat{h}_{\tilde{e}}(\Xb^e_{\phi(\tilde{e})})\big)\hat{h}_e(\Xb^e_{\phi(e)})\right]\right| + \nu\Vert\alpha\Vert_2^2$
    \\ Update gap: $\delta \xleftarrow{} \sup_{e\in\mathcal{E}}\sup_{\Xb^e}|\sum_{\tilde{e}\in \mathcal{E}}\hat{\alpha}_{\tilde{e}}\hat{h}_{\tilde{e}}(\Xb^e_{\phi(\tilde{e})})|$
    \\ Update model: $\hat{f} \xleftarrow{} \hat{f} + \eta\sum_{\tilde{e}\in \mathcal{E}}\hat{\alpha}_{\tilde{e}}\hat{h}_{\tilde{e}}$
    }
    \Return $\hat{f}$
    \caption{Causal Aggregation Boosting}\label{algorithm:boosting}
\end{algorithm}

We summarize the procedure in Algorithm~\eqref{algorithm:boosting}. The proposed procedure leverages the same causal aggregation principle as in the linear aggregation framework, but as a subroutine of the fitting procedure. Our procedure can be understood as gradient descent on the space of functions $\mathcal{F}_{\mathcal{E}}$, where at each step we linearize the space of functions by computing the environment-wise direction of maximal descent, and then fit the best global linear approximation before taking the gradient step. The crucial point is that the environment-wise components are fitted based on the unconfounded covariates of each sample, as the variables $\Xb_{\phi(e)}^e$ used to estimate $\hat{h}_e$ in environment $e$ are precisely those that are randomized.

\section{Numerical Simulations}\label{section:simulations}
We run simulations based on synthetic and semi-synthetic data. We start by validating in Section~\ref{section:subsection:exp-simulations-just} the main results of linear causal aggregation in the just-identified case, its extension to the over-identified case and compare it to do calculus. We then validate non-linear causal aggregation via boosting in Section~\ref{section:subsection:exp-non-linear}. Finally, we then generate in Section~\ref{section:subsection:exp-semi-synthetic} a semi-synthetic data set based on a gene perturbation experiment where we know the true causal relationship between $\Xb$ and $Y$ by design and use it to validate our causal aggregation method.

\subsection{Aggregation in the linear case}\label{section:subsection:exp-simulations-just}

\subsubsection{Aggregation in the just-identified case}

\begin{figure*}%
\centering
\includegraphics[trim={5cm 6.5cm 15cm 6.2cm},clip,width=.3\linewidth]{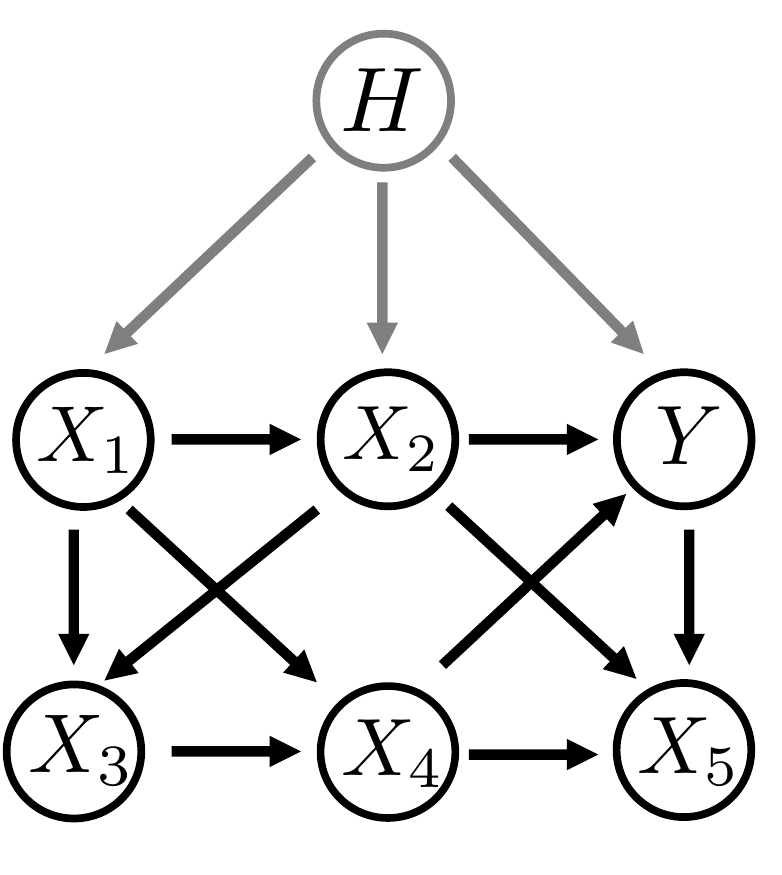}
\caption{SEM for observational environment: Experimental environments are obtained by randomizing covariates in the above SEM, for which their incoming edges are removed. 
}
\label{fig:Ex2-1}
\end{figure*}

We build on Example~\ref{example:LinearSEM} to illustrate a simple case where we aggregate causal information from a diverse set of environments, with different types of constraints built within each environment. We add new covariates as represented in Figure~\ref{fig:Ex2-1}, and generate samples according to a SEM $\mathcal{M}^0$ as in \eqref{def:linearSEM}. The structural equations defining the model are as follows: 
\begin{align}\label{eq:SEM-complex}
\begin{split}
    H & \xleftarrow[]{} \epsilon_H
\\[-0.5\jot] X_1 & \xleftarrow[]{} 2H + \epsilon_1
\\[-0.5\jot] X_2 & \xleftarrow[]{} X_1 + H + \epsilon_2
\\[-0.5\jot] X_3 & \xleftarrow[]{} -X_1 + 2X_2 + \epsilon_3
\\[-0.5\jot] X_4 & \xleftarrow[]{} X_1 + X_3 + \epsilon_4
\\[-0.5\jot] X_5 & \xleftarrow[]{} 2X_2 + X_4 - Y + \epsilon_5
\\[-0.5\jot] Y & \xleftarrow[]{} X_2 + 2X_4 + H + \epsilon_Y
\end{split}
\end{align}
This model contains latent factors that simultaneously affect the covariates and the response. The causal vector is given by \mbox{$\beta^0 = (0,1,0,2,0)$}, but the presence of a latent variable $H$ biases the least squares estimates of $Y$ on $\Xb$. Additionally, the vector $\epsb$ is sampled from a standard Gaussian distribution (where for simplicity we include $\epsilon_H$ in the vector $\epsb$ previously defined). We consider four different environments. First, samples are collected from an observational environment $e_1$ where an instrument $I$ is available for $X_1$, so that its structural equation becomes $X_1 \xleftarrow[]{} 2H + I + \epsilon_1$. An experimental environment $e_2$ is generated by intervening on $X_3$ and $X_5$. This intervention is translated into a SEM with the same equations above except for $X_3, X_5$ that follow $X_i \xleftarrow[]{} \epsilon_i$. We generate samples for the third environment $e_3$ from an interventional data set where we assume we intervene on $X_2$ and also we know the parental set of $X_4$. These interventional environments follow the the causal model $\{\mathcal{M}^e_{\phi(e)}\}_{e}$ where $\phi(e_2)= \{3,5\}$ and $\phi(e_3) = \{2\}$ (cf. equation~\ref{equation:perturbedSEM}). Randomization in the experimental datasets removes spurious correlations due to the latent variable $H$ and truncate the direct contribution of parental variables of intervened covariates. Last, we generate an environment $e_4$ where all covariates are randomized and thus the latent variable $H$ has no longer any confounding effect, i.e. $\phi(e_4)=\{1,2,3,4,5\}$.

\begin{figure*}[ht]
\centering
\includegraphics[trim={3cm 6.5cm 3cm 6.2cm},clip,width=.65\linewidth]{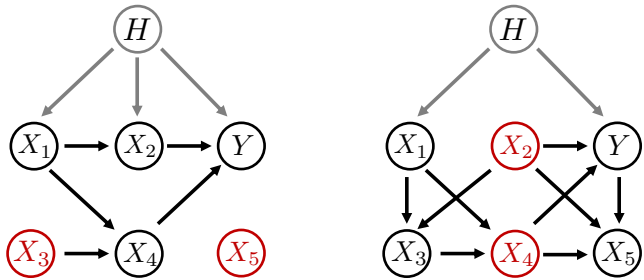}
\caption{Graphical representation of experimental environments: In environment $e_2$ (left), $X_3$ and $X_5$ are randomized. In environment $e_3$ (right) only $X_2$ is randomized. Nodes in red correspond to covariates that are associated to an orthogonality constraint.}
\label{fig:Ex2-2}
\end{figure*}

We run several experiments and evaluate the performance when using different subsets of constraints derived from the environments above. The first criterion we look at is whether the actual coverage from the confidence intervals matches the nominal pre-specified level. We compute confidence intervals with nominal coverage 0.95 for each coordinate and report actual coverage for increasing sample sizes $n \in \{50,100,200,500,1000\}$, averaging across 500 repetitions. In those cases where several methods attain the correct coverage, we compare them based on the mean length of the individual coordinate-wise confidence intervals. We begin comparing our causal aggregation estimator $\hat{\beta}$ presented above and simple OLS, and we later add other methods with competitive performance.

Our first experiment (referred to as experiment A) relies on constraints from environments $e_1$, $e_2$ and $e_3$ that are directly derived from the constraint inducing variables. Each variable has exactly one such constraint across all environments, this corresponds to the just-identified case. We report in Table~\ref{table:ex1:coverage} the coverage of the confidence intervals derived in equation~\eqref{equation:confidence_interval} for different values of sample size, based on 500 repetitions of the simulation. As expected, our causal aggregation estimator has proper coverage following our theoretical results, and OLS, being inconsistent, does not achieve the targeted coverage.

\begin{table}
\caption{We construct confidence intervals with .95 nominal coverage based on the SEM in Figure~\ref{fig:Ex2-1}. We report the actual coverage of those confidence intervals for increasing sample sizes, over 500 repetitions. Bold values indicate that the target coverage is achieved.}
\label{table:ex1:coverage}
\resizebox{\textwidth}{!}{
\begin{tabular}{@{}llrrrrc@{}}
\toprule
 & Estimator & \multicolumn{1}{c}{$n=50$}
& \multicolumn{1}{c}{$n=100$} & \multicolumn{1}{c}{$n = 200$}& \multicolumn{1}{c}{$n = 500$}
& \multicolumn{1}{c}{$n = 1000$} \\
\cmidrule{1-7}
\multirow{2}{7em}{Experiment A} & Causal Aggregation    &  $\boldsymbol{0.98\pm0.02}$  &  $\boldsymbol{0.97\pm0.02}$   &  $\boldsymbol{0.96\pm0.02}$   &  $\boldsymbol{0.95\pm0.02}$   &  $\boldsymbol{0.96\pm0.02}$  \\
\cmidrule{2-7}
& Pooled data OLS    &  $0.35\pm0.04$  &  $0.23\pm0.04$   &  $0.14\pm0.03$   &  $0.04\pm0.02$   & $0.01\pm0.01$ \\
\bottomrule
\end{tabular}
}
\end{table}

\subsubsection{Simulations in the over-identified case}
We define additional experiments built upon the previous example based on the SEM \eqref{eq:SEM-complex}. We already assumed the parental set for $X_4$ is known, but previously we only constructed the corresponding constraint in $e_3$, leaving information on the side. Based on our assumptions on the causal model $\mathcal{M}^e_{\phi(e)}$, the parental set for $X_4$ remains unchanged in those environments where $X_4$ is not intervened on, and thus we have access to two additional constraints. We define experiment B by including those, obtained in environments $e_1$ and $e_2$, on top of those already available in experiment A. Experiment C is defined though the randomization constraints from environment $e_4$ only. Last, experiment D combines all constraints from experiments B and C.

We first check the actual coverage of causal aggregation and OLS and report the results in Table~\ref{table:ex2:coverage} where each row corresponds to a different experiment. Again, causal aggregation achieves the target coverage in all experiments, but OLS only does so in experiment C, where full randomization removes any confounding due to the latent variable and thus the OLS estimator becomes consistent and the derived confidence intervals attain the nominal coverage. Actually, OLS and our causal aggregation method are the same in the just-identified setting of experiment C. We also notice a slight over-coverage in certain settings: the achieved coverage may be higher than the target. We conjecture that this is due to the two-step nature of the regression adjustment constraints: the estimated variance of the constraints may be slightly upwardly biased due to the adjustment step previous to the construction of the orthogonality constraint, which leads to wider-than-expected confidence intervals in finite samples.

\begin{table*}
\caption{Extending Table~\ref{table:ex1:coverage} to experiments B, C and D. As expected, the confidence intervals derived from our causal aggregation method always achieve the target coverage.}
\label{table:ex2:coverage}
\resizebox{\textwidth}{!}{
\begin{tabular}{@{}llrrrrc@{}}
\toprule
  & Estimator & \multicolumn{1}{c}{$n=50$}
& \multicolumn{1}{c}{$n=100$} & \multicolumn{1}{c}{$n = 200$}& \multicolumn{1}{c}{$n = 500$}
& \multicolumn{1}{c}{$n = 1000$} \\
\cmidrule{1-7}
\multirow{2}{7em}{Experiment B}
& Causal Aggregation    &  $\boldsymbol{0.97\pm0.02}$  &  $\boldsymbol{0.96\pm0.02}$   &  $\boldsymbol{0.96\pm0.02}$   &  $\boldsymbol{0.96\pm0.02}$   &  $\boldsymbol{0.96\pm0.02}$  \\
\cmidrule{3-7}
& Pooled Data OLS    &  $0.36\pm0.04$  &  $0.23\pm0.04$   &  $0.13\pm0.03$   &  $0.05\pm0.02$ &  $0.01\pm0.01$  \\
\cmidrule{1-7}
\multirow{2}{7em}{Experiment C}
& Causal Aggregation    &  $\boldsymbol{0.94\pm0.02}$  &  $\boldsymbol{0.94\pm0.02}$   &  $\boldsymbol{0.96\pm0.02}$   &  $\boldsymbol{0.95\pm0.02}$   &  $\boldsymbol{0.95\pm0.02}$ \\
\cmidrule{3-7}
& Pooled Data OLS    &  $\boldsymbol{0.95\pm0.02}$  &  $\boldsymbol{0.95\pm0.02}$   &  $\boldsymbol{0.95\pm0.02}$   &  $\boldsymbol{0.95\pm0.02}$   &  $\boldsymbol{0.95\pm0.02}$   \\
\cmidrule{1-7}
\multirow{2}{7em}{Experiment D}
& Causal Aggregation    &  $\boldsymbol{0.95\pm0.02}$  &  $\boldsymbol{0.94\pm0.02}$   &  $\boldsymbol{0.94\pm0.02}$   &  $\boldsymbol{0.95\pm0.02}$   &  $\boldsymbol{0.95\pm0.02}$   \\
\cmidrule{3-7}
& Pooled Data OLS    &  $0.39\pm0.04$  &  $0.20\pm0.04$   & $0.06\pm0.03$   &  $0.01\pm0.01$   &  $0.00\pm0.00$   \\
\bottomrule
\end{tabular}
}
\end{table*}

\begin{figure*}[ht]
\centering
\includegraphics[clip,width=.95\linewidth]{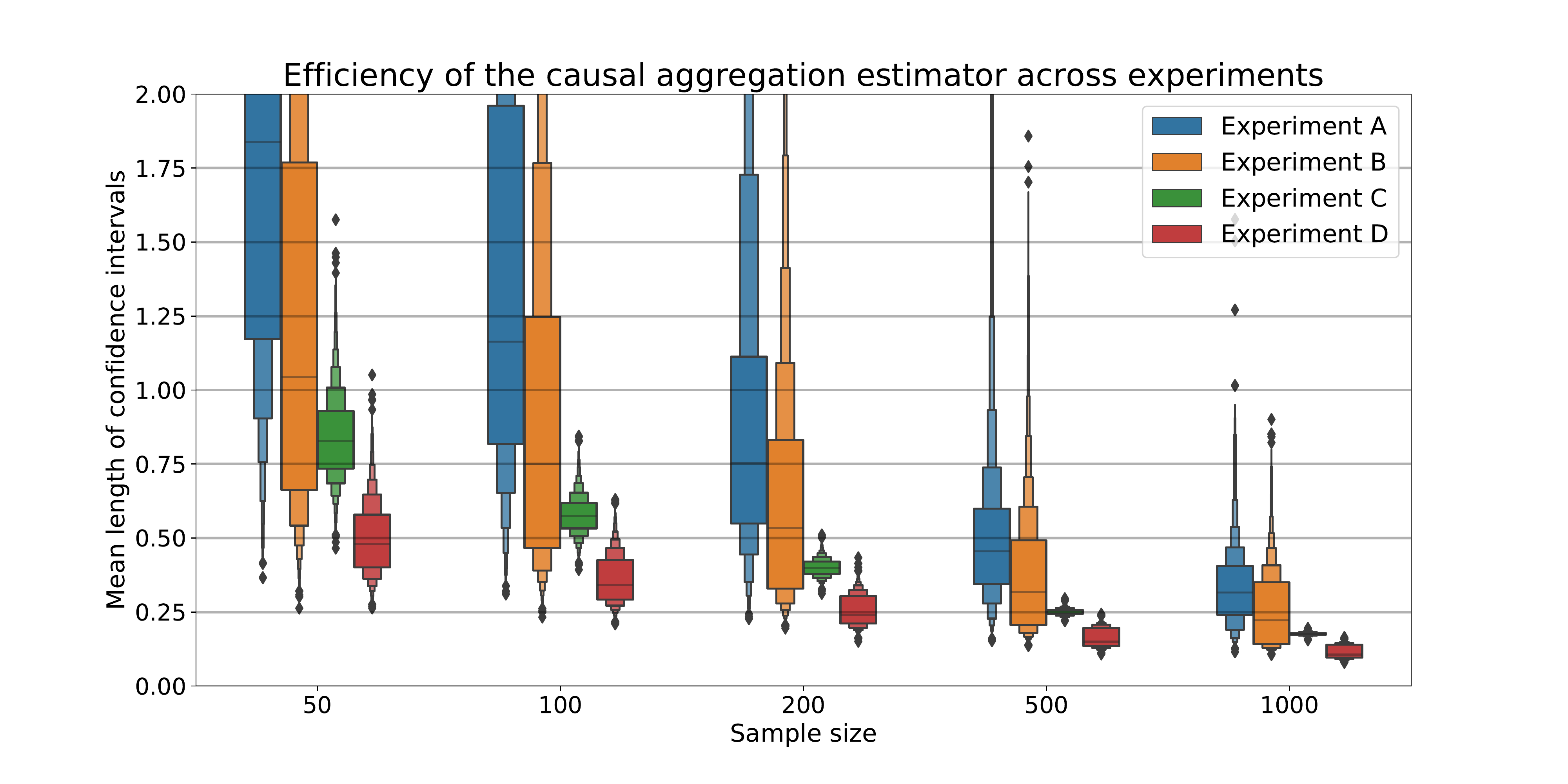}
\caption{Box plots of the mean length of the confidence intervals across 500 repetitions for a range of sample sizes in different experimental settings. The causal aggregation estimator is consistent and the confidence intervals achieve the target coverage. Our method can additionally leverage information from every environment that generates a valid constraint to tighten the confidence interval length.}
\label{fig:section7-3}
\end{figure*}

Having checked that the confidence intervals from our causal aggregation method achieve the nominal coverage, we now turn to comparing size across different experiments. Figure~\ref{fig:section7-3} shows the main point of these simulations: increasing the number of constraints improves the asymptotic efficiency of the causal aggregation estimator. Experiments A and C correspond to two just-identified cases, and additional constraints define experiments B and D respectively, which are thus over-identified cases where we use the two-step MM aggregation method. In both cases the confidence intervals become tighter: experiment B shows an improvement over A, and experiment D improves over C and B. Even if causal aggregation (and OLS in experiment C) are able to properly estimate $\beta^0$, our method leverages additional information to improve efficiency.

\begin{table*}
\caption{We construct confidence intervals with .95 nominal coverage and report the actual coverage of those confidence intervals for increasing sample sizes, over 500 repetitions. Bold values indicate that the target coverage is achieved.}
\label{table:ex3:coverage}
\centering
\resizebox{.95\textwidth}{!}{
\begin{tabular}{@{}lrrrrc@{}}
\toprule
 & \multicolumn{1}{c}{$n=50$}
& \multicolumn{1}{c}{$n=100$} & \multicolumn{1}{c}{$n = 200$}& \multicolumn{1}{c}{$n = 500$}
& \multicolumn{1}{c}{$n = 1000$} \\
\cmidrule{1-6}
Do-calculus    &  $\boldsymbol{0.94\pm0.02}$  &  $\boldsymbol{0.95\pm0.02}$   &  $\boldsymbol{0.95\pm0.01}$   &  $\boldsymbol{0.96\pm0.02}$   &  $\boldsymbol{0.95\pm0.02}$  \\
\cmidrule{2-6}
Causal Aggregation   &  $\boldsymbol{0.94\pm0.03}$  &  $\boldsymbol{0.96\pm0.02}$   &  $\boldsymbol{0.95\pm0.01}$   &  $\boldsymbol{0.97\pm0.01}$   &  $\boldsymbol{0.98\pm0.01}$  \\
\bottomrule
\end{tabular}
}
\end{table*}

\subsubsection{Simulations with additional Do-calculus constraints}\label{sec:do-calc-sim}

Within environment $e_3$, we report the performance of another standard, well-known estimation procedure of the causal effects: regression after applying do-calculus on a known causal graph. To showcase the advantage of our methodology that leverages all the available data from each environment, we assume that the causal graphs in each environment are \emph{known} to the do-calculus method. We have applied the causal transportability formula (a generalization of the do-calculus as implemented in the \texttt{R}-package \texttt{causaleffect}) to the environments $\{e_1,e_2,e_3\}$. Among the environments $\{e_1,e_2,e_3\}$ (or combinations of environments), only $e_3$ leads to identifiable causal effects on $Y$ of intervening on all $X_1,X_2,X_3,X_4,X_5$. These are obtained by regressing $Y$ on the set of variables $X_1,X_2,X_3,X_4$ and reporting the coefficients of $X_2,X_4$. Therefore, based on samples from $e_3$ only, we can estimate $\beta^0$ based on two constraints derived from the do-calculus approach and the perfect knowledge of the causal graph structure. We omit any analysis with samples from the fully randomized environment $e_4$, which could be used in an equivalent manner by both methods.

The causal transportability formula does not recommend using any samples from $e_1$ and $e_2$. This is potentially due to the fact that the approach is non-parametric. Samples from $e_1$ and $e_2$ help improving the efficiency of the estimator whenever we use causal aggregation. Indeed, the previous causal aggregation constraints in those environments can be updated with the knowledge derived from the causal graph: certain coefficients in $\beta^0$ are equal to 0 regardless of the environment. We can therefore restrict the problem to a regression of $Y$ over $X_2, X_4$ and build a causal aggregation estimator combining the two do-calculus constraints with the updated ones from environments $e_1$ and $e_2$. We verify in Table~\ref{table:ex3:coverage} that both methods achieve the target coverage as expected. We then compare the length of the confidence intervals in Figure~\ref{fig:section7-4}, showing again an improvement in efficiency as we include additional constraints from multiple environments with causal aggregation.

\begin{figure*}[ht]
\centering
\includegraphics[clip,width=.95\linewidth]{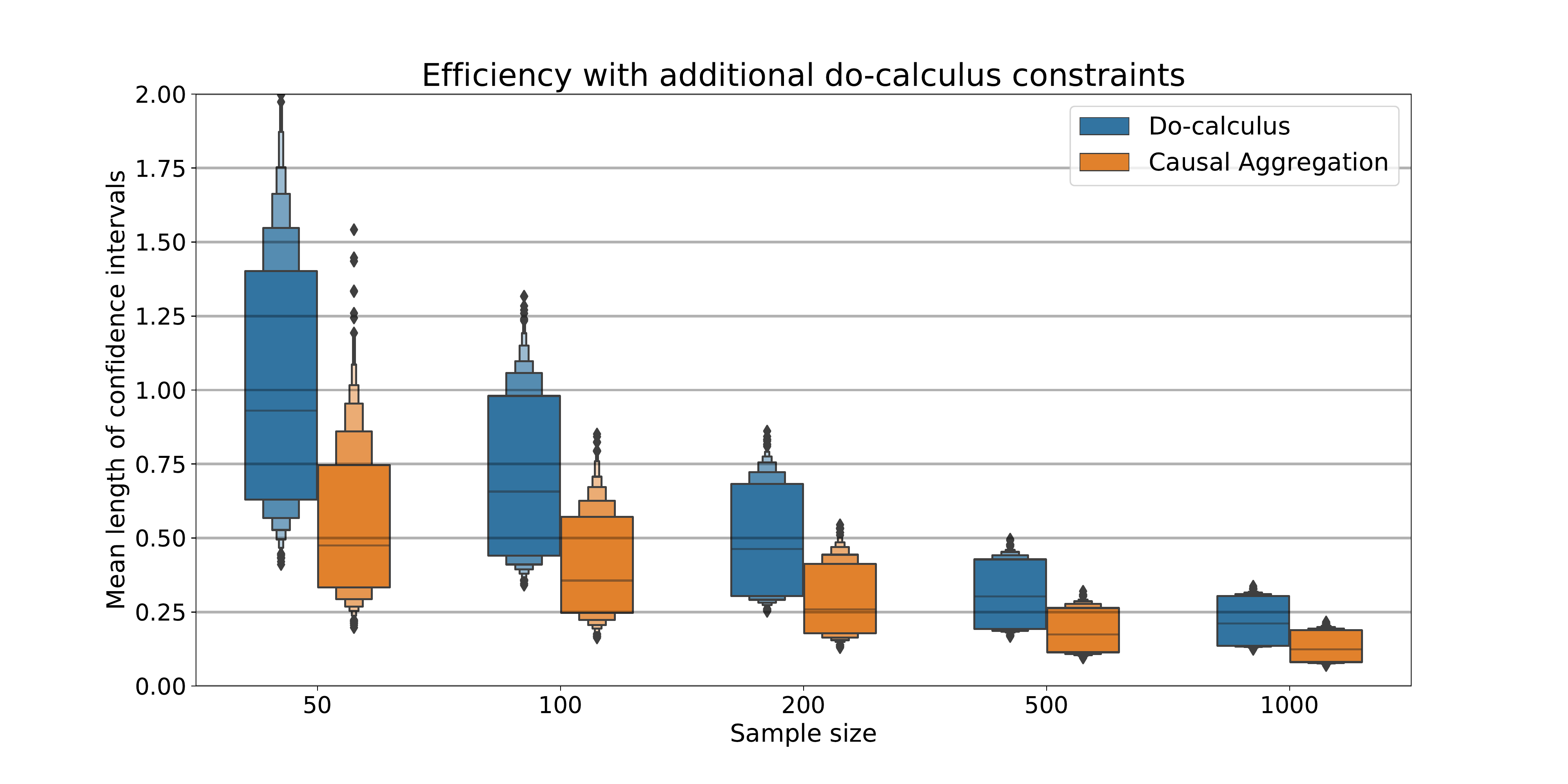}
\caption{Box plots of the mean length of the confidence intervals across 500 repetitions for a range of sample sizes.}
\label{fig:section7-4}
\end{figure*} 

In conclusion, even in situations where OLS and do-calculus lead to consistent estimators of the causal effects, we can benefit from collecting data from other environments with confounding where those previous methods may fail to work: if we can extract partial information based on constraints derived from valid assumptions, we can improve the efficiency of our estimators.

\subsection{Non-Linear Aggregation}\label{section:subsection:exp-non-linear}

Consider the same SEM used for simulations in the low dimensional case, whose DAG structure is given in Figure~\ref{fig:Ex2-1}, but where the structural equation of $Y$ is given by a complex non-linear function. We also modify structural equations of $\Xb$ to validate the fact that causal boosting performs well when allowing flexible non-linear covariate structural equations. In particular, we consider the following SEM:
\begin{align*}
    H & \xleftarrow[]{} \epsilon_H
\\[-0.5\jot] X_1 & \xleftarrow[]{} 2H + \epsilon_1
\\[-0.5\jot] X_2 & \xleftarrow[]{} X_1 + H + \epsilon_2
\\[-0.5\jot] X_3 & \xleftarrow[]{} -2X_1 X_2 + \epsilon_3
\\[-0.5\jot] X_4 & \xleftarrow[]{} \log(1+|X_1|) + X_3 + \epsilon_4
\\[-0.5\jot] X_5 & \xleftarrow[]{} 2X_2 + X_4 - Y + \epsilon_5
\\[-0.5\jot] Y & \xleftarrow[]{} f^0(X_1,X_2,X_3,X_4) + 2H + \epsilon_Y
\end{align*}
where all the disturbance terms $\epsilon_H, (\epsilon_j)_{j\in [5]}, \epsilon_Y$ are jointly independent, standard Gaussian. This model satisfies the assumptions on our non-linear SEM framework (cf. equation~\ref{equation:SEM}). We make confounding explicit by introduced the latent factor $H$ that makes the aggregated residual $H+\epsilon_Y$ of the response structural equation non independent of the covariates---as some of these covariates are also affected by $H$. However, whenever we randomize the covariates in the parental set $S^0=\{1,2,3,4\}$ of the response $Y$, we get that the residual $H+\epsilon_Y$ is centered conditionally on the randomized covariates. We run several simulations where we analyze the performance of the proposed causal aggregation boosting algorithm \eqref{algorithm:boosting} for different choices of function $f^0$ and sets of environments $\mathcal{E}$. Our simulations show the correctness of our method whenever the model is well specified in the sense that for any interaction term in $f^0$ there is an environment in $\mathcal{E}$ where all the input covariates to the interaction term are randomized.

We fit the model on training data and put aside a test set for evaluating the performance of our method. We report the $L_2$ loss of our estimator $\hat{f}$ given by an oracle with access to $f^0$: 
\begin{align*}
    \mathbb{E}\big[(f^0(\Xb^e) - \hat{f}(\Xb^e))^2\big]
\end{align*}
where the expectation is taken over the test set samples $\Xb^e$ averaged across all environments. 

We use two different response models in our simulations, denoted by $f^0_I, f^0_{II}$. We first use a piece-wise constant function given by: 
\begin{align*}
    f^0_I(x_1,x_2,x_3,x_4) & =  \mathbb{1}_{x_1>0} + \mathbb{1}_{x_2>0} - 2\mathbb{1}_{x_2>0,x_3>-1} + 2\mathbb{1}_{x_1<0, x_4<-1} + 3\mathbb{1}_{x_1<0,x_2<1,x_3<-1} 
\end{align*}
We choose as the base model to fit the environment-wise $\hat{h}_e$ a decision tree: $f^0_I$ should be simple enough that the causal aggregation boosting algorithm recovers it efficiently. If we instead use another response function that is no longer piece-wise constant, then the number of samples required to achieve the same level of accuracy increases, and we run simulations for the following loss function $f^0_{II}$ that additively combines a linear component and a piece-wise constant component:
\begin{align*}
    f^0_{II}(x_1,x_2,x_3,x_4) & = 2x_1 - 2x_2 - 2\mathbb{1}_{x_2>0,x_3>-1} + 2\mathbb{1}_{x_1<0, x_4<-1} + 3\mathbb{1}_{x_1<0,x_2<1,x_3<-1} 
\end{align*}
Let a set of two environments $\mathcal{E}^0=\{e_1,e_2\}$ given by $\phi(e_1)=\{1,2,3\}$ and $\phi(e_2) = \{1,4\}$, so that $f^0_I, f^0_{II}\in \mathcal{F}_{\mathcal{E}^0}$. We can therefore recover these target functions using such set of environments as our data source. Conversely, given that the set of environments $\mathcal{E}$ constrains the class of functions $\mathcal{F}_{\mathcal{E}}$ our estimator $\hat{f}$ belongs to, we analyze how the performance degrades as the approximation error increases when limiting the interventions across environments. 

We define four simulations each corresponding to a different set of environments in Table~\ref{table:environments}. For each simulation we characterize the set of environments as a set of subsets of indices, each representing one environment. Within such environment, the indices represent the simultaneously randomized covariates. Simulation A corresponds to collecting data from $\mathcal{E}^0$ above. Simulations B, C and D correspond to fitting our causal boosting procedure in cases where not all the interacting covariates have an environment where they are simultaneously randomized, leading to an approximation bias that increases as we limit the amount of simultaneous randomization. We report the results of fitting our causal boosting procedure to data generated with the response function given by $f^0_I$ in Figure~\ref{fig:boosting-envs}, and as expected we observe that causal boosting correctly recovers the non-linear function: the $L_2$ loss decreases towards 0 with sample size. However, the loss reaches a plateau in those other simulations where not all the interactions in the target function have corresponding randomized data.

\begin{table}
\begin{center}
\begin{tabular}{c c c c} 
Simulation & Sets of randomized covariates \\ [0.5ex] 
\hline
Simulation A & $\{1,2,3\},\{1,4\} $\\ 
Simulation B & $\{1,2,3\} $\\ 
Simulation C & $\{1\},\{2,3\} $\\
Simulation D & $\{1\},\{2\},\{3\} $
\end{tabular}
\end{center}
\caption{Four different experimental settings: for each simulation, a set of subsets of indices represents the subsets of simultaneously randomized covariates within an environment, across the set of environments.} \label{table:environments}
\end{table}

\begin{figure*}[ht]
\centering
\includegraphics[clip,width=\linewidth]{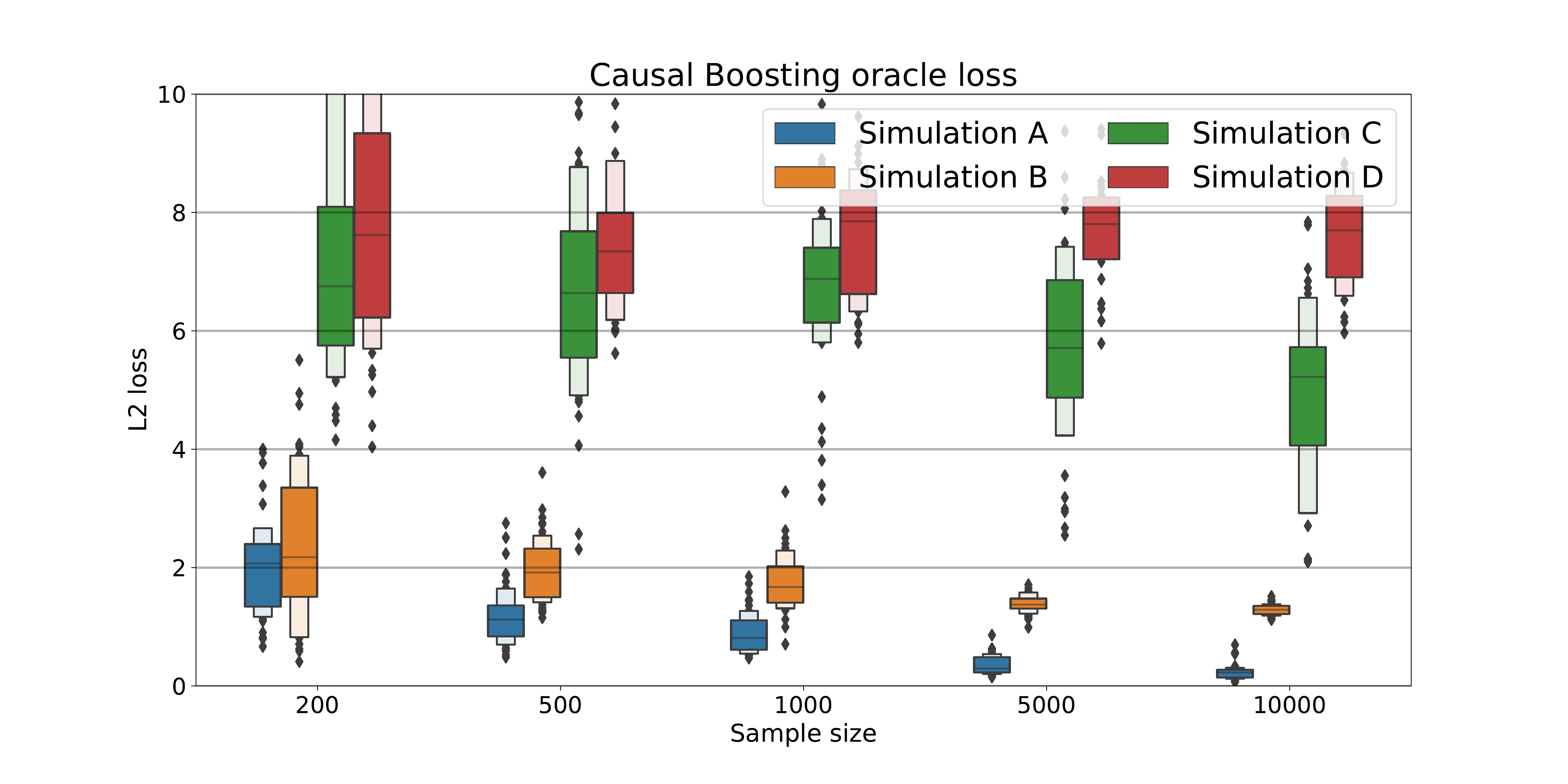}
\caption{Causal boosting loss under different simulation settings. Only Simulation A has environments with properly randomized data such that there is no approximation error.}
\label{fig:boosting-envs}
\end{figure*}

We now turn to analyzing the performance of causal boosting with respect to alternative methods. We focus on the set of environments defined in Simulation A, so that there is no approximation error. A first naive baseline corresponds to fitting a non-parametric regression method to the pooled data across environments, similar to the OLS baseline previously used in the linear case. We choose random forests for that purpose, given that the base estimator in the boosting procedure is a decision tree. We also compare the boosting procedure to the causal backfitting procedure defined in Algorithm~\ref{algorithm:backfitting}. Both these procedures will perform poorly: the naive pooled method will indeed fail to recover the true non-linear response function. The backfitting procedure will suffer from instability during training, and will often fail to converge. We define a simpler one-pass backfitting procedure as an alternate baseline that does recover the response function, albeit with an additional key assumption similar to the one in the above section. If we assume that the causal graph is \emph{known}, then we can implement a single-pass backfitting procedure that has competitive performance. The individual components of the estimator $\hat{f}$ following equation~\ref{decomposition_function} are fitted sequentially, starting by those components whose variables are most downstream in the DAG.

We report the result of the simulations in Figure~\ref{fig:boosting-competing}, comparing the causal aggregation boosting procedure versus the competing methods, where the response function is either $f^0_I$ or $f^0_{II}$ and $\mathcal{E}^0$ is the set of environments. Again, our aggregation procedure recovers a good estimate of the true causal function for both $f^0_I, f^0_{II} \in \mathcal{F}_{\mathcal{E}^0}$. However, for a given sample size, the recovery loss for $f^0_I$ is smaller than that of $f^0_{II}$, as expected from the inclusion of linear components in the latter. Random forests fitted on the pooled data perform poorly. This is expected, as this method is oblivious to the confounding. Training the causal backfitting procedure is unstable and does not converge towards the true non-linear response, and this issue is even more problematic for $f^0_{II}$. Finally, the single pass backfitting performs reasonably well but requires larger sample sizes to achieve the same performance as the causal boosting procedure.

\begin{figure*}[ht]
\centering
\includegraphics[trim={4cm 0 4cm 0cm},clip,width=1\linewidth]{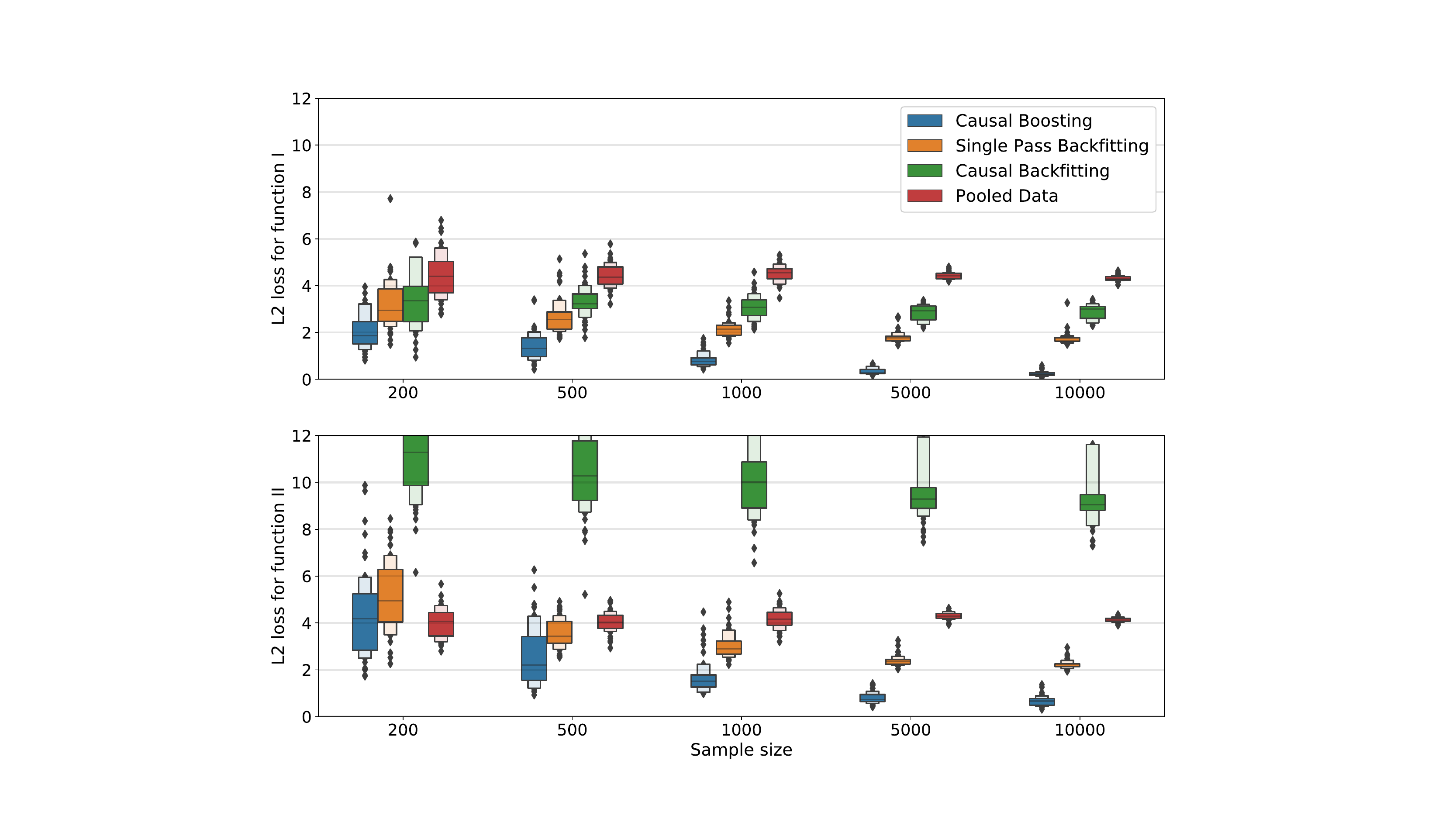}
\caption{Comparative performance of different aggregation methods for two separate target functions}
\label{fig:boosting-competing}
\end{figure*} 

In conclusion, our causal boosting procedure performs very well. The per-environment model fitting, coupled with the linearization step that re-weights the components of the estimator by drawing from our linear causal aggregation theory is a competitive method for estimating non-linear causal responses from partially randomized data from multiple environments.

\subsection{Semi-Synthetic Data Set}\label{section:subsection:exp-semi-synthetic}

We finally validate our procedure on semi-synthetic data set that we create from a single-cell RNA sequencing (scRNA-seq) data set published in \cite{gasperini2019genome}. This data set contains mRNA counts of $207,324$ individual cells measuring the expression of around $10,000$ genes. Each of those genes is in a neighborhood of several potential DNA regulatory elements---a section of non-protein coding DNA that modulates the expression of a gene. This experiment in particular focuses on enhancers, regulatory elements that promote gene expression. The goal of the experiment is to find true biological regulatory associations between candidate enhancers pre-selected based on their chemical characteristics, and genes that they potentially up-regulate. This particular data set is obtained by perturbing cells via CRISPR interference \citep{jinek2012programmable}. This technology allows to directly intervene in the gene regulation mechanism of by targeting a section of the genome with a CRISPR/Cas9 molecule attached to a specifically designed guide RNA (gRNA). That gRNA determines the locations where CRISPR/Cas9 introduces DNA mutations that perturb gene regulation. In this data set the targeted DNA sections are the candidate enhancers, therefore genes that were actually regulated by a CRISPR-intervened candidate enhancer region are differentially down-regulated.

We focus on a particular gene called PRKCB, with 10 selected candidate enhancers. Cells are then intervened on a random subset of those enhancers: that is, each cell receives a random combination of CRISPR molecules with gRNAs that target the 10 enhancers. In addition to measuring the response $Y$ corresponding to PRKCB gene expression, the data set has a vector of binary variables $(X_1, \dots, X_{10})$ for each cell indicating the presence of CRISPR molecules, corresponding to the treatment covariates. Additionally, one needs to control for technical factors $W$ such as temperature and batch ID that simultaneously affect the treatment variables and response, usually introducing those as covariates in a regression model. 

\begin{table*}
\caption{Our aggregation MM estimator $\hat{\beta}^{MM}$ on the perturbed samples closely matches the original regression estimates $\hat{\beta}^{OLS}$ based on the original unperturbed data set (first two columns). However, the confidence intervals are much larger, owing to the additional noise introduced in the perturbed data set by the latent factor. The OLS estimates $\hat{\beta}^{OLS}(e)$ on any single perturbed data set $e=1,2,3$ leads to biased estimates of the regression vector due to the additional synthetic confounding we introduce. Thus, some of the OLS estimates are far from the estimates on the original, unperturbed data set. On the other hand, the estimates $\hat \beta^{MM}$ are close to the estimates from the original data set.}
\label{table:semi-synthetic}
\resizebox{\textwidth}{!}{
\begin{tabular}{@{}lc|cccc@{}}
& \multicolumn{1}{c}{Original data set} & \multicolumn{4}{c}{Perturbed datasets}\\
\cmidrule{2-6}
Targeted Candidate Enhancer &  \multicolumn{1}{c}{$\hat{\beta}^{OLS}$} & \multicolumn{1}{c}{$\hat{\beta}^{MM}$}
& \multicolumn{1}{c}{$\hat{\beta}^{OLS}(1)$} & \multicolumn{1}{c}{$\hat{\beta}^{OLS}(2)$} & \multicolumn{1}{c}{$\hat{\beta}^{OLS}(3)$} \\
\hline
1822 top       & $\phantom{-}0.00 \pm 0.04$ & $\phantom{-}0.29\pm 0.29$     & $\phantom{-}0.05\pm 0.07$ & $\phantom{-}0.67 \pm 0.07$    & $\phantom{-}0.35 \pm 0.07$ \\
1856 top       & $\phantom{-}0.01 \pm 0.04$ & $-0.04\pm 0.31$               & $\phantom{-}0.05\pm 0.08$ & $\phantom{-}0.76 \pm 0.07$    & $\phantom{-}0.35 \pm 0.07$ \\
1857 top       & $\phantom{-}0.06 \pm 0.05$ & $\phantom{-}0.02\pm 0.32$     & $\phantom{-}1.09\pm 0.09$ & $\phantom{-}0.04 \pm 0.09$    & $\phantom{-}0.60 \pm 0.08$ \\
1863 second    & $\phantom{-}0.03 \pm 0.04$ & $-0.29\pm 0.26$               & $\phantom{-}0.56\pm 0.07$ & $-0.04 \pm 0.08$              & $\phantom{-}0.39 \pm 0.07$ \\
1863 top       & $\phantom{-}0.00 \pm 0.04$ & $-0.24\pm 0.25$               & $\phantom{-}0.58\pm 0.07$ & $-0.03 \pm 0.07$              & $\phantom{-}0.30 \pm 0.06$ \\
1865 top       & $-0.10 \pm 0.04$             & $\phantom{-}0.00\pm 0.13$   & $\phantom{-}0.30\pm 0.06$ & $\phantom{-}0.41 \pm 0.06$    & $-0.11\pm 0.06$ \\
1866 top       & $-0.50 \pm 0.04$             & $-0.46\pm 0.13$             & $-0.02\pm 0.06$           & $\phantom{-}0.08 \pm 0.06$    & $-0.51\pm 0.07$ \\
1866 second    & $-0.46 \pm 0.05$             & $-0.42\pm 0.17$             & $\phantom{-}0.29\pm 0.08$ & $\phantom{-}0.53 \pm 0.08$    & $-0.43\pm 0.09$ \\
1867 top       & $-0.16 \pm 0.05$             & $-0.13\pm 0.15$             & $\phantom{-}0.47\pm 0.08$ & $\phantom{-}0.78 \pm 0.08$    & $-0.12\pm 0.08$ \\
1897 top       & $-0.03 \pm 0.05$             & $\phantom{-}0.17\pm 0.17$   & $\phantom{-}0.72\pm 0.08$ & $\phantom{-}0.76 \pm 0.07$    & $-0.02\pm 0.08$ \\
\hline  
\end{tabular}
}
\end{table*}

For the purpose of our semi-synthetic data example we assume a linear model for the response, which is regressed on the binary vector of covariates to identify the actual regulatory elements of the gene conditionally on the technical factors, though other models can better capture these effects \cite{gasperini2019genome, katsevich2020conditional}. We obtain a estimate of the effects of perturbing each candidate enhancer via the regression vector. We now perturb the original data set and create several environments with synthetic confounding that leads to biases in estimation at each individual environment. We then recover the original regression vector through our aggregation procedure. We generate three environments $\mathcal{E}= \{1,2,3\}$ by randomly partitioning the initial data set samples in three smaller datasets, and for each partition we pick a subset of covariates $S_e$, $e=1,2,3$. We then introduce in each environment $e$ a confounding latent factor $H$ that simultaneously affects the subset of selected covariates $(X_i)_{i\in S_e}$ and the response $Y$. We also assume that no covariate is selected in all three new environments. We represent in Figure~\ref{fig:Ex5-1} graphical models corresponding to two potential environments. Recovering the original regression estimate based on just one environment is no longer possible, as we assume that in practice $H$ is not observed. However, if we assume that for each environment we know the subset of unconfounded covariates, we can then construct orthogonality constraints for these, and use the aggregation estimator to recover an estimate that is closer to the original one based on the unperturbed data set. 
\begin{figure*}[ht]
\centering
\includegraphics[trim={2.5cm 6.5cm 1.3cm 6.2cm},clip,width=.65\linewidth]{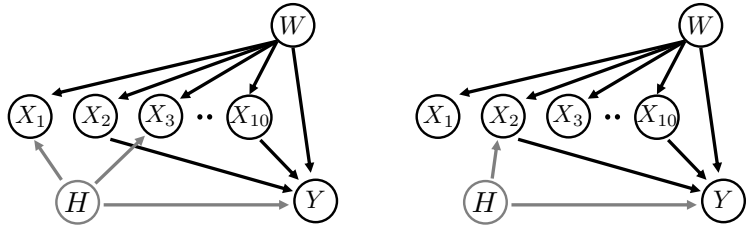}
\caption{Semi-synthetic environments with added confounders: We split the initial data set into three environments and sample three subsets of covariates. In each environment we add a confounding term between the corresponding set of covariates and the response. No covariate is perturbed in all three environments. Graphical models above represent two potential environments.}
\label{fig:Ex5-1}
\end{figure*}
In practice, given one of the environments $e\in \mathcal{E}$, we perturb samples as follows:
\begin{align*}
    H & \xleftarrow[]{} \text{Ber}(0.5)
\\[-0.5\jot] \tilde{X}_i & \xleftarrow[]{} X_i + H \quad \forall i \in S_e
\\[-0.5\jot] \tilde{Y} & \xleftarrow[]{} Y  -4*(H-1)
\end{align*}
leaving the other covariates unchanged. The data set covariates are coded via the set of enhancers that are targeted by the CRISPR gRNA. In particular, for $e=1$ we perturb the set of covariates $\{$1857 top, 1863 second, 1863 top, 1865 top, 1866 top, 1866 second, 1867 top, 1897 top$\}$, for $e=2$ we perturb the set of covariates $\{$1822 top, 1856 top, 1865 top, 1866 top, 1866 second, 1867 top, 1897 top$\}$ and for $e=3$ we perturb the set of covariates $\{$1822 top, 1856 top, 1857 top, 1863 second, 1863 top$\}$. The OLS estimator $\hat{\beta}^{OLS}(e)$ on any environment $e\in \{1,2,3\}$ is thus biased. However, aggregated estimator $\hat{\beta}^{MM}$ provides us with estimates that broadly match those of the original OLS estimator, albeit with wider confidence intervals due to the increased variance because of the latent factor. In particular, confidence intervals for each covariate overlap between these two methods, which is not the case for OLS estimators built on just one environment.

\section{Discussion}\label{section:discussion}

We have introduced a method for aggregating causal information across data sets, with
the goal of estimating the effect of simultaneous interventions. The method is based on causal constraints, which arise from experimental manipulations and background knowledge. On observational data, instrumental variables and knowledge about parental sets can be used to define causal constraints. In the low-dimensional case, we discuss a two-stage procedure that allows for asymptotically efficient estimation and inference of causal effects. In the high-dimensional case we provide an $\ell_1$-regularized estimator and derive finite sample bounds. These finite sample bounds rely on a cone invertibility factor, which play a similar role as the sparse eigenvalue condition in high-dimensional linear regression. As our high-dimensional theory indicates, the method might be of use whenever a very large number of experiments are available, but only very few samples per experiment are observed. Whenever few covariates are randomized, we recommend instead a pre-screening step on observational data to reduce the dimensionality of the problem. In addition, we provide a non-linear version of causal aggregation of experimental data that flexibly estimates interactions between covariates and non-linearities in the response. This non-linear method uses the linear aggregation step as a sub-routine when training the model following a boosting-like procedure, using at every step only those covariates in each sample that are unconfounded. On synthetic and semi-synthetic data sets we show that the proposed method outperforms naive methods that do not take into account the special structure induced by the causal constraints. Most of the causal constraints individually make use of data from only one environment. However, as in the case of cross-product invariance \eqref{eq:orthogonality_dantzig}, a causal constraint can leverage data from several environments. Looking ahead, it would be interesting to explore whether novel constraints can be derived that make use of data from different environments simultaneously.

\newpage

\bibliography{main.bib}

\appendix

\section{SEM Extension to Additive Shifts}\label{section:additive_interventions}

In Section~\ref{section:setting} we defined how new environments are generated by randomization of subsets of covariates. In addition to this intervention mechanism, we consider additive noise interventions where the distribution of the disturbance variable changes across environments by an additive shift that is independent of the base distribution. Assume that within environment $e$ there is a set $\psi(e) \subset [p]$ of covariates that have an additive intervention. With respect to the base model $\mathcal{M}^0$, the following structural equations are modified.
\begin{equation}\label{equation:perturbedSEMsoft}
\mathcal{M}^e_{\psi(e)}: \begin{cases}
    X_j^e \longleftarrow \sum_{k \in pa_0(j)}a_{jk}X_k^e + \epsilon_j + \delta^e_j \qquad \forall j \in \psi(e)
    \\ Y^e \longleftarrow \sum_{k \in [p]}\beta^0_{k}X_k^e + \epsilon_{Y}
    \\ \{\epsilon_j\}_{j\in [p+1]} \independent \{\delta_j^e\}_{j\in [p]}
\end{cases}
\end{equation}
where $\delta^e_j = 0$ for all $j \notin \psi(e)$. Let us give a justification for this model of environments with a practical example. Experimentation by randomization provides a very concrete way of perturbing a system, by direct manipulation of a covariate that modifies the structural mechanism that generates it. A concrete example of this is a gene knock-out experiment via CRISPR-Cas9: the targeted genes are no longer expressed, shifting the expression of other downstream genes in the regulation pathway. Additionally, environments may differ by some change in the overall environment. For example, different cell lines may have different baseline expressions of some subset of genes that shift the overall gene expression distribution. In practice, we want to assume that we know the subset of variables that are affected by this background shift, and that they precede in the causal mechanism any other variable that is experimentally manipulated.

\section{High-Dimensional Aggregation Simulations: Pre-screening Before Collecting Experimental Data}\label{section:exp-high-dimensional}

Whenever the number of covariates is too large, obtaining orthogonality constraints for each covariate may become prohibitive. In practice we may be allowed to choose a subset of covariates to randomize based on an initial observational data set. Assuming that the connectivity matrix is sparse, we can run a pre-selection step to select a few variables, then based on an experimental environment construct the appropriate orthogonality constraints. We then assume the direct effects of the discarded covariates is $0$, and thus we construct an estimator in the low-dimensional framework. However, to construct valid confidence intervals for such estimator based on pre-selected covariates we need to avoid using the same samples for model selection and inference. Our description above does indeed use two different sets of samples, where (easily accessible) observational data is used for covariate selection and then experimental data is obtained at a second stage. Otherwise we would split the data set to solve this issue.

Consider the following SEM where the covariate dimension is $200$, where the connectivity matrix has a sparse structure given by the graph in Figure~\ref{fig:Ex4-1}. The only covariate with a non-zero direct effect on $Y$ is $X_{99}$. The values of the entries in the connectivity matrix are sampled from a Gaussian distribution centered in 1 and with variance $0.5$. We run a Lasso regression \citep{tibshirani1996regression} on a standardized observational data set to select a subset of the covariates. Under some conditions on the coefficient sizes, as the sample size increases the selected set of covariates contains the the Markov blanket of $Y$ \citep[Section 2.5]{buhlmann2011statistics}, which is equal to the subset $\{X_1,X_2,X_{99},X_{100}, X_{199}, X_{200}\}$. In practice, this pre-selection step is adding a few random variables to the set of those that are randomized, our aggregation method works as long as the variables in the Markov blanket are selected and randomized in the experimental data. We report in our simulations the average number of Markov blanket covariates selected by Lasso (out of 6) as sample size increases.

\begin{table*}
\caption{Results of simulation from Appendix~\ref{section:exp-high-dimensional}. We estimate the regression coefficient on a subset of covariates that are pre-selected with a Lasso regression on an initial observational data set. We split the set of selected covariates into two, and build two experimental environments where covariates from each split are randomized in the corresponding environment. We then generate orthogonality constraints and run our aggregation procedure, as well as the pooled OLS. We additionally report the average number of selected covariates by Lasso.}
\label{table:ex4:coverage}
\resizebox{\textwidth}{!}{
\begin{tabular}{@{}lrrrrrc@{}}
\toprule
 & Sample Size & \multicolumn{1}{c}{$n=50$}
& \multicolumn{1}{c}{$n=100$} & \multicolumn{1}{c}{$n = 200$}& \multicolumn{1}{c}{$n = 500$}
& \multicolumn{1}{c}{$n = 1000$} \\
Estimator & Selected MB vars. & $2.72\pm0.26$ & $2.98\pm0.26$ & $3.74\pm0.16$ & $4.70\pm0.27$ & $5.3\pm0.12$ \\
\cmidrule{1-7}
\multirow{2}{7em}{Causal Aggregation}& Coverage    &  $0.99\pm0.01$  &  $0.99\pm0.01$   &  $0.97\pm0.04$   &  $0.97\pm0.05$   &  $0.95\pm0.06$  \\
& Average length   &  $136.6\pm171.9$   &  $3.84\pm2.04$  &  $0.51\pm0.03$  &  $0.32\pm0.01$   &  $0.12\pm0.01$   \\
\cmidrule{2-7}
\multirow{2}{7em}{Pooled data OLS}& Coverage   &  $0.80\pm0.11$  &  $0.81\pm0.11$   &  $0.75\pm0.12$   &  $0.72\pm0.12$   &  $0.58\pm0.14$  \\
& Average length   &  $0.43\pm0.04$   &  $0.31\pm0.02$  &  $0.19\pm0.02$  &  $0.13\pm0.01$   &  $0.05\pm0.01$
\\
\bottomrule
\end{tabular}
}
\end{table*}

Based on this selection step, we partition the selected covariates and for each subset we generate experimental datasets where those covariates are randomized. We do so by removing the dependence of the randomized covariates in their parents and assigning it a random standard Gaussian variable. We run the procedure for the just-identified setting based on the orthogonality constraints obtained from the experimental environments, as well as the OLS on the pooled data from the experimental environments. We run this procedure for different sample sizes ($n \in \{100,200,500,1000,5000\}$) and report the confidence intervals coverage with nominal coverage 0.95 and average length. We report the results in Table~\ref{table:ex4:coverage}. As the sample size increases, our aggregation procedure correctly estimates the sparse regression coefficient. For small sample sizes, the confidence intervals are meaningless as with few samples errors propagate in both the screening step and the aggregation step in the just-identified case, which is too imprecise with small sample sizes as previously seen. Pooled OLS coverage is well below the nominal value.

\begin{figure*}[ht]
\centering
\includegraphics[trim={4cm 6.5cm 4.5cm 9cm},clip,width=.60\linewidth]{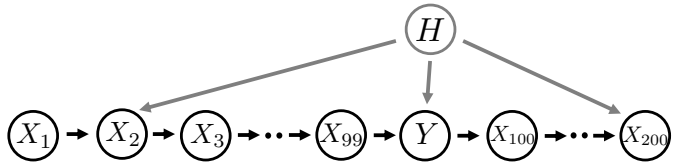}
\caption{SEM for high-dimensional example: Observational samples are generated via the above SEM with 200 covariates, with a confounder affecting the response and two covariates.}
\label{fig:Ex4-1}
\end{figure*}

\section{Proofs}\label{section:proofs}

\subsection{Proof of Proposition~\ref{proposition:linear_constraints}}

\begin{proof}
The linear constraint is a consequence of the independence between the residual term obtained $\epsilon_Y = Y - \Xb^T\beta^0$ and an exogenous variable. For an instrumental variable $I$ or a randomized covariate $X_j$, independence arises by definition. This leads to constraints of the type:
\begin{align*}
    & \mathbb{E}[I(Y-\Xb^T\beta^0)]=0
    \\ & \mathbb{E}[X_j(Y-\Xb^T\beta^0)]=0
\end{align*}
In a DAG, a variable is independent of its non-descendent nodes conditionally on its parental set \citep[Theorem 3.2.2]{pearl2009causal}. Our constraint is based on conditioning on the parental set of a covariate $X_j$. This parental set corresponds is obtained from the DAG $\bar{G}^0$, where the observational distribution factorizes. However, in practice we can not condition on the unobserved variables, hence we must assume that the parental set of $X_j$ in $\bar{G}^0$ is the same as the parental set in $G^0$. When adjusting $X_j$ for the parental set we get a random variable given by $X_j -\sum_{k \in pa_0(j)}a_{jk}X_k = \epsilon_j$ as we assumed $c_{jk}=0$: i.e. the latent factors do not affect $X_j$. Additionally, given that the graph has no cycles whenever we know that the response $Y$ is in the parental set of $X_j$ we immediately get as a constraint the fact that the regression coefficient corresponding to $X_j$ is 0. Therefore from now on we assume that $Y$ is not in the parental set. The residual is not represented in the graph on its own, but we can derive the orthogonality constraint as follows. Replacing $Y$ by $(Y-\Xb^T\beta^0) = \sum_{k \in [p']}d_k H_k + \epsilon_{Y}$, we get that this term is independent of the adjusted term $X_j - \hat{X}_j$ where $\hat{X}_j = \sum_{k \in pa_0(j)}a_{jk}X_k$. Hence we obtain the result. 
Finally, the last constraint derived from additive interventions (called \emph{inner-product invariance}) is proved in \citet[Proposition 1]{rothenhausler2019causal}.
\end{proof}

\subsection{Proof of Proposition~\ref{proposition:linear_identifiability}}

\begin{proof}
We can assume without loss of generality that the ordering $\{1,\dots,p\}$ is a topological ordering for the DAG $G^0$. We stack constraint vectors that form $\Gb$ following the same topological ordering in the constraint related variable. We will now show that stacking the vertical vectors in such ordering leads to an upper triangular matrix that is invertible. The idea is that each constraint is related to one covariate, and we show that the corresponding constraint inducing variable $R$ is independent of the previous covariates in the topological ordering. Consider the $j$-th vector in $\Gb$. In the IV setting, the instrument is independent of all the non-descendant variables of $X_j$ in the $\bar{G}^0$ graph. Therefore the constraint vector has its first $j-1$ entries equal to 0. Whenever randomizing $X_j$, that variable is now independent of all its non-descendants in the $\bar{G}^0$ graph. Adjusting for direct causes (i.e. conditioning on the parental set), assuming there is no latent variable effect, we get that the residual after adjusting $X_j$ given by $R=X_j - \hat{X_j}$ is independent of the non-descendants of $X_j$. Again, the constraint vector has its first $j-1$ entries equal to 0. 
\end{proof}

\subsection{Proof of Proposition~\ref{proposition:asymptotic_normality}}

\begin{proof}
For notation simplicity, we associate $\mathcal{C}=[p]$. For $c\in [p]$, let $u_c\in \mathbb{R}^p, v_c \in \mathbb{R}$, and let $u=(u_c)_{c\in [p]}, v = (v_c)_{c\in [p]}$. Consider the mapping 
\begin{align*}
    \xi:\begin{cases}
    \mathbb{R}^{(p+1)\times p} \rightarrow \mathbb{R}^p
    \\ (u,v) \mapsto \Gb^{-1}(u)\Zb(v)
\end{cases}    
\end{align*}
where  
\begin{align*}
\Gb(u) = \big[ u_c^T\big]_{1\leq c \leq p}  := \begin{pmatrix}
u_1^T
\\ \vdots
\\ u_p^T
\end{pmatrix} \in \mathbb{R}^{p \times p} \; ,
\qquad \Zb(v) = \begin{pmatrix}
v_1 \\ \vdots \\ v_p
\end{pmatrix} \in \mathbb{R}^p
\end{align*}
where the notation $\big[u_c^T\big]_{1\leq c \leq p}$ denotes the matrix obtained by stacking the row vectors $u^T_j$. If we consider the restriction of $\xi$ to those vectors $(u,v)$ such that the matrix $\Gb(u)$ is invertible, then we have that $\xi$ is continuously differentiable at $(u,v)$ and its derivative is given by 
\begin{align*}
    d\xi_{u,v}(\tilde{u}, \tilde{v}) = -\Gb(u)^{-1}\Gb(\tilde{u})\Gb(u)^{-1}\Zb(v) + \Gb(u)^{-1}\Zb(\tilde{v})
\end{align*}
Therefore, via a Taylor approximation, we have that 
\begin{align*}
    & \xi(u',v') - \xi(u,v) - d\xi_{u,v}(u'-u, v'-v)
    \\ = & \xi(u',v') - \xi(u,v) + \Gb(u)^{-1}\big(\Gb(u') - \Gb(u)\big)\Gb(u)^{-1}\Zb(v) - \Gb(u)^{-1}\big( \Zb(v') -\Zb(v)\big)
    \\ = & \xi(u',v') - \xi(u,v) + \Gb(u)^{-1}\big(C(u',v') - C(u,v)\big)\begin{pmatrix}
    \Gb(u)^{-1}\Zb(v)
    \\ -1
    \end{pmatrix}
    \\ = & o_{\substack{u'\rightarrow u \\v' \rightarrow v}}\big(\Vert(u',v')- (u,v) \Vert \big)
\end{align*}
where 
\begin{align*}
C(u,v) = \big[u^T_c,v_c\big]_{1\leq c \leq p} \in \mathbb{R}^{p \times (p+1)}
\end{align*}
We collect $n_e$ samples $(\Xb^e_i, Y^e_i)_{i\in [n_e]}$ in environment $e$, where $\Xb^e_i$ is the $i$-th vector sample of the covariates $\Xb^e_{i} = (X^e_{1,i}, \dots, X^e_{p,i})$. Each constraint $c\in [p]$ is based on samples from environment $e_c\in \mathcal{E}$: for each $c\in [p]$ we collect a constraint inducing variable $(R^{c}_i)_{i\in [n_{e_c}]}$. Such $R^{c}$ may correspond to $X_c^{e_c}$ if such covariate is randomized, or an instrument in environment $e_c$ for covariate $X_c^{e_c}$. $R^c$ can also be the residual variable when adjusting for the parental set of a given covariate (although as indicated in Section~\ref{section:regression_adjustment}, regressing on the parental set and estimating the orthogonality constraint must be done with distinct datasets). Recall that $n = \sum_{e\in \mathcal{E}}n_e$ is the total number of samples, and that the sample sizes from different environments grow at the same rate: $\frac{n_e}{n}\longrightarrow \rho_e \in (0,1)$. The vector $\beta^0$ is characterized as the solution to the system of equations:
\begin{align*}
    \Zb(v) = \Gb(u)\beta
    \Longleftrightarrow  \begin{cases}
    \mathbb{E}[R^{e_1}(Y^{e_1}- \beta^{0,T}\Xb^{e_1})] = 0
    \\ \dots
    \\ \mathbb{E}[R^{e_p}(Y^{e_p}- \beta^{0,T}\Xb^{e_p})] = 0
    \end{cases} 
\end{align*} 
so that $\beta^0 = \xi(u,v)$ where $u = (\mathbb{E}[R^{c}\Xb^{e_c}])_{c\in [p]}$, and $v = (\mathbb{E}[R^{c}Y^{e_c}])_{c\in [p]}$. Analogously, the estimator $\hat{\beta}$ is given by plugging in the previous equation the sample averages: $\hat{\beta} = \xi(\hat{u}, \hat{v})$ where 
\begin{equation}
    \begin{cases}
    \hat{u} = \big( (1/n_{e_c})\sum_{i\in [n_{e_c}]}R^{c}_i\Xb_i^{e_c}\big)_{c\in [p]}
    \\ \hat{v} = \big( (1/n_{e_c})\sum_{i\in [n_{e_c}]}R^{c}_iY_i^{e_c}\big)_{c\in[p]}
    \end{cases}
\end{equation}
We have by the strong law of large numbers that $(\hat{u},\hat{v})- (u,v) = O_{P}(1/\sqrt{n})$. Therefore we get that
\begin{align*}
    & \xi(\hat{u},\hat{v}) - \xi(u,v) + \Gb(u)^{-1}\big(C(\hat{u},\hat{v}) - C(u,v)\big)\begin{pmatrix}
    \xi(u,v)
    \\ -1
    \end{pmatrix}
    \\ = & \hat{\beta} - \beta^0 + \Gb(u)^{-1}\big(C(\hat{u},\hat{v}) - C(u,v)\big)\begin{pmatrix}
    \beta^0
    \\ -1
    \end{pmatrix}
    \\ = & o_P(\Vert (\hat{u},\hat{v})- (u,v)\Vert)
\end{align*}
and thus
\begin{align*}
   \sqrt{n}\big(\hat{\beta} - \beta^0\big) + \Gb(u)^{-1}\sqrt{n}\big(C(\hat{u},\hat{v}) - C(u,v)\big)\begin{pmatrix}
    \beta^0
    \\ -1
    \end{pmatrix}\big) = o_P(O_P(1)) = o_P(1)
\end{align*}

Let $\Ub_i^{c} := R^{c}_i(\Xb^{e_c}_i,Y^{e_c}_i)\in \mathbb{R}^{p+1}$ for $c\in [p]$. We have the following convergence in distribution by the central limit theorem in environment $e$ by combining all the constraints $c$ based on environment $e$:
\begin{align*}
    & \sqrt{n_e}\; \text{Vect}\Bigg[ \Big(\frac{1}{n_e}\sum_{i=1}^{n_e}\Ub^{c}_i - \mathbb{E}[\Ub^{c}] \Big)^T \Bigg]_{e_c = e} \xrightarrow[n \to +\infty]{d} \mathcal{N}\Big( \textbf{0}; \big[ \text{Cov}(\Ub^{c},\Ub^{\tilde{c}})\big]_{\substack{e_c = e\\e_{\tilde{c}} = e}} \Big)
\end{align*}
Now, given the assumptions on the environment $e$, we have that 
\begin{align*}
    (\Ub^{c})^T\begin{pmatrix}
\beta^0
\\ -1
\end{pmatrix} = &  R^{c}(\beta^{0,T}\Xb^{e} - Y^{e}) 
\\ = & -R^{c}\epsilon^{e}_Y
\end{align*}
As $R^c$ is an instrument, a randomized covariate, or the residual from regressing the covariate on its parental set, we have $R^c$ and $\epsilon_Y^e$ are independent as indicated in Proposition~\ref{proposition:linear_constraints}. Therefore, as $\epsilon_Y^e$ is centered, we get: 
\begin{align*}
\begin{pmatrix}
\beta^0
\\ -1
\end{pmatrix}^T \text{Cov}(\Ub^{c},\Ub^{\tilde{c}})\begin{pmatrix}
\beta^0
\\ -1
\end{pmatrix} = & \text{Cov}(R^{c}\epsilon^{e}_Y, R^{\tilde{c}} \epsilon^{e}_Y) = \sigma_{e}^2\text{Cov}(R^{c}, R^{\tilde{c}})
\end{align*}
where $\sigma_{e}^2 := \mathbb{E}[\epsilon_Y^{e,2}]$. We have by Slutsky's theorem, given that $\frac{n_e}{n}\rightarrow \rho_e$,
\begin{align*}
    \sqrt{n}\Bigg[ \Big(\frac{1}{n_e}\sum_{i=1}^{n_e}\Ub^{c}_i - \mathbb{E}[\Ub^{c}] \Big)^T\begin{pmatrix}
\beta^0
\\ -1
\end{pmatrix} \Bigg]_{e_c = e} \xrightarrow[n \to +\infty]{d} \mathcal{N}\Big( \textbf{0}; \frac{\sigma_{e}^2}{\rho_e}\text{Cov}\big(R^{c}\big)_{e_c = e}\Big)
\end{align*}
Also, the covariance matrix $\text{Cov}\big(R^{c}\big)_{e_c = e}$ is diagonal and invertible, as within environment $e$ constraint inducing variables that are either randomized covariates or instruments are jointly independent. We now concatenate the results for different environments. Given the independence of samples across environments, we have that:
\begin{align*}
    & \sqrt{n}\big(C(\hat{u},\hat{v}) - C(u,v)\big)\begin{pmatrix}
    \beta^0
    \\ -1
    \end{pmatrix} \xrightarrow[n \to +\infty]{d} \mathcal{N}\Big(\textbf{0}; \text{Cov}\big(\frac{\sigma_{e_c}}{\sqrt{\rho_{e_c}}}R^{c}\big)_{c\in [p]} \Big)
\end{align*}
where the covariance matrix $\text{Cov}\big(\frac{\sigma_{e_c}}{\sqrt{\rho_{e_c}}}R^{c}\big)_{c\in [p]}$ is still diagonal. We conclude: 
\begin{align*}
    \sqrt{n}(\hat{\beta} - \beta^0) \xrightarrow[n \to +\infty]{d} \mathcal{N}\Big( \textbf{0}; \Sigma \Big)
\end{align*}
where $\Sigma = \Gb^{-1}(u)\text{Diag}\Big(\frac{\sigma_{e_c}^2}{\rho_{e_c}}\text{Var}(R^{c})\Big)_{c\in [p]}\Gb^{-1,T}(u)$.
\end{proof}

\subsection{Proof of Proposition~\ref{proposition:lq-bound}}

\begin{proof}

Letting $z_\infty := \Vert\hat{\Zb} - \hat{\Gb}\beta^0\Vert_{\infty}$, we show that in the event $\{z_{\infty} \leq \lambda\}$ our regularized estimator $\hat{\beta}(\lambda)$ satisfies
\begin{align}\label{equation:bound-CIF}
    \Vert\hat{\beta}(\lambda) - \beta^0 \Vert_{q} \leq \frac{|S^0|^{1/q}\Vert \hat{\Gb}(\hat{\beta}(\lambda)-\beta^0)\Vert_{\infty}}{\text{CIF}_q(S^0, \hat{\Gb})} \leq \frac{2|S^0|^{1/q}\lambda}{\text{CIF}_q(S^0, \hat{\Gb})}
\end{align}
We follow \cite{ye2010rate} and show the following inequality:
\begin{align*}
    \Vert \hat{\beta}(\lambda)_{(S^0)^c} - \beta^0_{(S^0)^c} \Vert_1 = & \Vert \hat{\beta}(\lambda)_{(S^0)^c}\Vert_1
    \\ = & \Vert \hat{\beta}(\lambda)\Vert_1 - \Vert \hat{\beta}(\lambda)_{S^0}\Vert_1
    \\ \leq & \Vert \beta^0\Vert_1 - \Vert \hat{\beta}(\lambda)_{S^0}\Vert_1
    \\ \leq & \Vert \beta^0_{S^0}\Vert_1 - \Vert \hat{\beta}(\lambda)_{S^0}\Vert_1
    \\ \leq & \Vert \beta^0_{S^0} - \hat{\beta}(\lambda)_{S^0}\Vert_1
\end{align*}
where we used the fact that $S^0$ is the support of $\beta^0$ and that in the event $\{z_{\infty} \leq \lambda\}$ the true vector $\beta^0$ is feasible, therefore we get $\Vert\hat{\beta}(\lambda)\Vert_1 \leq \Vert \beta^0 \Vert_1$. If the upper bound is equal to $0$, then we get $\hat{\beta}(\lambda) =  \beta^0$ as these vectors coincide over $S^0$ and $(S^0)^c$ and the inequality above holds. Otherwise, this shows that $\hat{\beta}(\lambda) - \beta^0$ belongs to the cone $\mathfrak{C} := \{\ub: \Vert\ub_{(S^0)^c}\Vert_1 \leq \Vert \ub_{S^0}\Vert_1\neq 0\}$ and we get the first inequality in equation~\eqref{equation:bound-CIF} by definition of the CIF. Furthermore, we show that $\Vert \hat{\Gb}(\hat{\beta}(\lambda)-\beta^0)\Vert_{\infty} \leq 2\lambda$ in the event $\{z_{\infty} \leq \lambda\}$:
\begin{align*}
    \Vert \hat{\Gb}(\hat{\beta}(\lambda)-\beta^0)\Vert_{\infty} = & \Vert (\hat{\Zb}-\hat{\Gb}\beta^0)-(\hat{\Zb} - \hat{\Gb}\hat{\beta}(\lambda))\Vert_{\infty} \leq 2\lambda
\end{align*}
as $\beta^0, \hat{\beta}(\lambda)$ are in the feasible set.

We need a high probability bound for the event $\{z_{\infty} \leq \lambda\}$, and we also need to control the CIF value for the empirical matrix $\hat{\Gb}$ in equation~\eqref{equation:bound-CIF}, which entails using concentration inequalities to control the deviation of $\hat \Gb$ from $\Gb$. Lemma 3 in \cite{rothenhausler2019causal} provides the following bound for the gap between the CIF under the estimator matrix $\hat{\Gb}$ and the CIF value under $\Gb$:
\begin{align*}
    \big|\text{CIF}_q(S^0, \Gb) - \text{CIF}_q(S^0, \hat{\Gb})\big| \leq 2|S^0|\Vert\Gb - \hat{\Gb}\Vert_{\infty}
\end{align*}
therefore in the event $\{2|S^0|\Vert\Gb - \hat{\Gb}\Vert_{\infty} \leq \frac{1}{2}\text{CIF}_q(S^0, \Gb)\}$ the following upper bound holds:
\begin{align}
    \Vert\hat{\beta}(\lambda) - \beta^0 \Vert_{q} \leq & \frac{4|S^0|^{1/q}\lambda}{\text{CIF}_q(S^0, \Gb)}
\end{align}
We now apply lemma~\ref{lemma:high-prob-bounds} to obtain a bound with high probability for $\Vert\Gb - \hat{\Gb}\Vert_{\infty}$ and the event $\{z_{\infty}\leq \lambda\}$. Let $t = 2\log p$ and set
\begin{equation*}
    \lambda := k\sigma_C\sigma_E\sqrt{\frac{t+\log p}{\min_{e\in \mathcal{E}}n_e}} = k\sqrt{3}\sigma_C\sigma_E\sqrt{\frac{\log p}{\min_{e\in \mathcal{E}}n_e}} \xrightarrow[]{} 0
\end{equation*}
where the universal constant $k$ is defined in the lemma. Given that choice of $\lambda$ we get by lemma~\ref{lemma:high-prob-bounds}: 
\begin{align*}
    & \mathbb{P} \Big(z_{\infty} \leq \lambda \Big) \geq 1 -\frac{2}{p^2} \qquad \text{and} \qquad \mathbb{P} \Big(\Vert\hat{\Gb}-\Gb\Vert_{\infty} \leq \frac{\sigma_X}{\sigma_E}\lambda\Big) \geq 1 -\frac{2}{p}
\end{align*}
In addition, as by assumption we have $\frac{1}{\text{CIF}_q(S^0, \Gb)}\sqrt{\frac{\log p}{\min_{e\in \mathcal{E}}n_e}} \xrightarrow[]{} 0$, we get that eventually
\begin{align*} 
\sqrt{\frac{\log p}{\min_{e\in \mathcal{E}}n_e}} & \leq \frac{\text{CIF}_q(S^0, \Gb)}{4\sqrt{3}k|S^0|\sigma_X\sigma_C}
\\ \frac{\sigma_X}{\sigma_E}\lambda & \leq \frac{\text{CIF}_q(S^0, \Gb)}{4|S^0|}
\end{align*}
The high probability bounds above then control the two events leading to inequality~\eqref{equation:bound-CIF}. Therefore with probability at least $1-\frac{4}{p}$, we have that 
\begin{align*}
    \Vert\hat{\beta}(\lambda) - \beta^0 \Vert_{q} \leq & \frac{4|S^0|^{1/q}\lambda}{\text{CIF}_q(S^0, \Gb)}
    \\ \leq & \frac{4\sqrt{3}k\sigma_C\sigma_E|S^0|^{1/q}}{\text{CIF}_q(S^0, \Gb)}\sqrt{\frac{\log p}{\min_{e\in \mathcal{E}}n_e}}
\end{align*}
hence the result.
\end{proof}

\begin{lemma}\label{lemma:high-prob-bounds}
Assume that $X^{e}_j$ are $\sigma_X^2$ sub-gaussian, $\epsilon_Y^e$ are $\sigma_E^2$ sub-gaussian, and that $R^c$ are $\sigma^{2}_C$ sub-gaussian for all $e\in \mathcal{E}, j \in [p], c\in \mathcal{C}$ and some fixed $\sigma_X^2, \sigma_E^2, \sigma_C^2>0$. There exists a universal constant $k>0$, such that for any $t>0$:
\begin{align*}
    & \mathbb{P} \Bigg(z_{\infty} \leq k\sigma_C\sigma_E\max\Big(\frac{t + \log p}{\min_{e\in \mathcal{E}}n_{e}},\sqrt{\frac{t + \log p}{{\min_{e\in \mathcal{E}}n_{e}}}}\Big)\Bigg) \geq 1 -2e^{-t}
    \\ & \mathbb{P} \Bigg(\Vert\hat{\Gb}-\Gb\Vert_{\infty} \leq k\sigma_C\sigma_X\max\Big(\frac{t + 2\log p}{\min_{e\in \mathcal{E}}n_{e}},\sqrt{\frac{t + 2\log p}{{\min_{e\in \mathcal{E}}n_{e}}}}\Big)\Bigg) \geq 1 -2e^{-t}
\end{align*}
\end{lemma}

\begin{proof}
We will prove the result by relying on concentration inequalities for sub-gaussian and sub-exponential random variables. We use Orlicz spaces and norms since this allows us to bound the products of quantities easily, by invoking inequalities that we discuss in the following. We refer to \cite{vershynin2018high} for further details on the use of Orlicz spaces in concentration inequalities. We also introduce universal constants, finite positive real numbers that do not depend on the other elements in the problem. The Orlicz norm of a random variable $X$ with respect to an Orlicz function $\psi$ is defined as 
\begin{equation*}
    \Vert X\Vert_{\psi} := \inf\Big\{t>0 \; : \; \mathbb{E}\Big[\psi\Big(\frac{X}{t}\Big)\Big] \leq 1 \Big\}
\end{equation*}
The Orlicz space with respect to $\psi$ is the space of random variables with finite Orlicz norm. Given the choices of $\psi_1, \psi_2$ defined below, we get that the corresponding Orlicz spaces are the families of sub-gaussian and sub-exponential random variables respectively.
\begin{align*}
    & \psi_1 : x\mapsto e^{x} - 1 \qquad \psi_2 : x\mapsto e^{x^2} -1
\end{align*}
These two spaces are connected by the following inequality that applies for any variables $X,Y$:
\begin{equation*}
    \Vert XY \Vert_{\psi_1} \leq \Vert X\Vert_{\psi_2}\Vert Y \Vert_{\psi_2}
\end{equation*}
Hence the products of random variables we have are sub-exponential as products of sub-gaussian random variables. The following inequalities hold for an universal constant $k_0$ that does not depend on the random variables.
\begin{align*}
    & \Vert R^c \epsilon^{e_c}\Vert_{\psi_1} \leq  \Vert R^{c}\Vert_{\psi_2}\Vert \epsilon^{e_c}\Vert_{\psi_2} \leq k_0\sigma_C\sigma_E
    \\ & \Vert R^c X^{e_c}_j - \mathbb{E}\big[R^c X^{e_c}_j\big]\Vert_{\psi_1} \leq \Vert R^{c}\Vert_{\psi_2}\Vert X^{e_c}_j\Vert_{\psi_2} \leq k_0\sigma_C\sigma_X
\end{align*}
where we relied on the fact that there exists a universal constant $k_{00}$, independent of the choice of the random variable, such that for any $X$ we have $\Vert X -\mathbb{E}[X]\Vert_{\psi_1} \leq k_{00}\Vert X\Vert_{\psi_1}$, and, for a $\sigma^2$ sub-gaussian random variable, there is another universal constant $k_{01}$ such that $\Vert X \Vert_{\psi_2} \leq k_{01}\sigma$. We apply Bernstein's inequality to the product $R^c\epsilon^{e_c}$ (cf. Theorem 2.8.1 in \cite{vershynin2018high}):
\begin{align*}
    \mathbb{P}\Big(\Big|\frac{1}{n_{e_c}}\sum_{i\in [n_{e_c}]}R^c_i\epsilon_i^{e_c}\Big| \geq t\Big) & \leq 2\exp\Big(-k_1 \min\Big(\frac{(n_{e_c}t)^2}{n_{e_c}\Vert R^c\epsilon^{e_c}\Vert_{\psi_1}^2}, \frac{n_{e_c}t}{\Vert R^c\epsilon^{e_c}\Vert_{\psi_1}}\Big)\Big)
    \\ & \leq 2\exp\Big(-k_1 \min\Big(\frac{n_{e_c}t^2}{k_0^2\sigma_C^2\sigma_E^2}, \frac{n_{e_c}t}{k_0\sigma_C\sigma_E}\Big)\Big)
    \\ & \leq 2\exp\Big(-k_1 n_{e_c} \min\Big(\frac{1}{k_0^2}, \frac{1}{k_0})\min\Big(\frac{t^2}{\sigma_C^2\sigma_E^2}, \frac{t}{\sigma_C\sigma_E}\Big)\Big)
    \\ & \leq 2\exp\Big(-k_2 n_{e_c} \min\Big(\frac{t^2}{\sigma_C^2\sigma_E^2}, \frac{t}{\sigma_C\sigma_E}\Big)\Big)
\end{align*}
Therefore we get by inverting the term in the exponential bound:
\begin{align*}
   2\exp(-t) \geq & \mathbb{P}\Big(\Big|\frac{1}{n_{e_c}}\sum_{i\in [n_{e_c}]}R^c_i\epsilon_i^{e_c}\Big|  \geq \sigma_C\sigma_E\max\Big(\frac{t}{n_{e_c}k_2},\sqrt{\frac{t}{n_{e_c}k_2}}\Big)\Big)  
   \\ \geq & \mathbb{P}\Big(\Big|\frac{1}{n_{e_c}}\sum_{i\in [n_{e_c}]}R^c_i\epsilon_i^{e_c}\Big| \geq \sigma_C\sigma_E\max\Big(\frac{t}{n_{e_c}},\sqrt{\frac{t}{n_{e_c}}}\Big)\max\Big(\frac{1}{k_2},\sqrt{\frac{1}{k_2}}\Big)\Big)
   \\ = & \mathbb{P}\Big(\Big|\frac{1}{n_{e_c}}\sum_{i\in [n_{e_c}]}R^c_i\epsilon_i^{e_c}\Big| \geq k_3\sigma_C\sigma_E\max\Big(\frac{t}{n_{e_c}},\sqrt{\frac{t}{n_{e_c}}}\Big)\Big)
\end{align*}
where $k_0, k_1, k_2, k_3$ are universal constants. Analogously for the product $R^cX^{e_c}_j$ we get that:
\begin{align*}
   \mathbb{P}\Big(\Big|\frac{1}{n_{e_c}}\sum_{i\in [n_{e_c}]}R^c_iX^{e_c}_{j,i} - \mathbb{E}\big[R^c X^{e_c}_j\big]\Big|  \geq k_3\sigma_C\sigma_X\max\Big(\frac{t}{n_{e_c}},\sqrt{\frac{t}{n_{e_c}}}\Big)\Big) \leq 2\exp(-t)
\end{align*}
We conclude in both cases by applying an union bound. For the term $z_{\infty}$ we have:
\begin{align*}
    & \mathbb{P}\Big(z_{\infty} \geq k_3\sigma_C\sigma_E\max\Big(\frac{t + \log p}{\min_{e\in \mathcal{E}}n_{e}},\sqrt{\frac{t + \log p}{{\min_{e\in \mathcal{E}}n_{e}}}}\Big)\Big)
    \\ = &\mathbb{P}\Big(\bigcup_{c\in [p]}\Big\{\Big|\frac{1}{n_{e_c}}\sum_{i\in [n_{e_c}]}R^c_i\epsilon_i^{e_c}\Big| \geq k_3\sigma_C\sigma_E\max\Big(\frac{t + \log p}{\min_{e\in \mathcal{E}}n_{e}},\sqrt{\frac{t + \log p}{{\min_{e\in \mathcal{E}}n_{e}}}}\Big)\Big\}\Big)
    \\ \leq & \sum_{c\in [p]} \mathbb{P}\Big(\Big|\frac{1}{n_{e_c}}\sum_{i\in [n_{e_c}]}R^c_i\epsilon_i^{e_c}\Big| \geq k_3\sigma_C\sigma_E\max\Big(\frac{t + \log p}{n_{e_c}},\sqrt{\frac{t + \log p}{{n_{e_c}}}}\Big)\Big)
    \\ \leq & 2p \exp\big(-(t + \log p)\big)
    \\ \leq & 2\exp(-t)
\end{align*}
For the term $\Vert \hat{\Gb} -\Gb\Vert_{\infty}$ we have:
\begin{align*}
    & \mathbb{P}\Big(\Vert \hat{\Gb} -\Gb\Vert_{\infty} \geq k_3\sigma_C\sigma_X\max\Big(\frac{t + 2\log p}{\min_{e\in \mathcal{E}}n_{e}},\sqrt{\frac{t + 2\log p}{{\min_{e\in \mathcal{E}}n_{e}}}}\Big)\Big)
    \\ = &\mathbb{P}\Big(\bigcup_{\substack{c\in \mathcal{C}\\j\in [p]}}\Big\{\Big|\frac{1}{n_{e_c}}\sum_{i\in [n_{e_c}]}R^c_iX^{e_c}_{j,i}- \mathbb{E}\big[R^c X^{e_c}_j\big]\Big| \geq k_3\sigma_C\sigma_X\max\Big(\frac{t + 2\log p}{\min_{e\in \mathcal{E}}n_{e}},\sqrt{\frac{t + 2\log p}{{\min_{e\in \mathcal{E}}n_{e}}}}\Big)\Big\}\Big)
    \\ \leq & \sum_{{\substack{c\in \mathcal{C}\\j\in [p]}}} \mathbb{P}\Big(\Big|\frac{1}{n_{e_c}}\sum_{i\in [n_{e_c}]}R^c_iX^{e_c}_{j,i}- \mathbb{E}\big[R^c X^{e_c}_j\big]\Big| \geq k_3\sigma_C\sigma_X\max\Big(\frac{t + 2\log p}{n_{e_c}},\sqrt{\frac{t + 2\log p}{{n_{e_c}}}}\Big)\Big)
    \\ \leq & 2p|\mathcal{C}| \exp\big(-(t + 2\log p)\big)
    \\ \leq & 2\frac{|\mathcal{C}|}{p}\exp(-t) \leq 2\exp(-t)
\end{align*}
Therefore we get the result.
\end{proof}

\subsection{Proof of Proposition~\ref{proposition:support_recovery}}

\begin{proof}
For completeness, we adapt the proof in \cite{rothenhausler2019causal} to show that 
\begin{equation*}
    \lim \mathbb{P}\big(\min_{j\in S^0}|\hat{\beta}(\lambda)_j| > 0\big) \xrightarrow[]{} 1.
\end{equation*}
In the event
\begin{equation*}
    \Vert\hat{\beta}(\lambda) - \beta^0 \Vert_{\infty} \leq \frac{K\sigma_C\sigma_E}{\text{CIF}_{\infty}(S^0, \Gb)}\sqrt{\frac{\log p}{\min_{e\in \mathcal{E}}n_e}}
\end{equation*}
the beta-min condition implies
\begin{align*}
    0 < & \min_{j\in S^0} |\beta^0_j| - \frac{K\sigma_C\sigma_E}{\text{CIF}_{\infty}(S^0, \Gb)}\sqrt{\frac{\log p}{\min_{e\in \mathcal{E}}n_e}}
    \\ \leq & \min_{j\in S^0}|\hat{\beta}(\lambda)_j|.
\end{align*}
This completes the proof.
\end{proof}

\subsection{Proof of Proposition~\ref{proposition:identifiability}}

Our proof is based on two key properties of the graphical structure of the interventional DAGs $(\bar{G}^e)_e$: they all share a same topological ordering---the topological ordering from $\bar{G}^0$ still holds when intervening on covariates---and the nodes of randomized covariates $\phi(e)$ in environment $e$ have no incoming edges by definition of randomization. These two properties are independent of the response node, and therefore to simplify our proofs we consider a different graph structure by marginalizing out the response node as follows. Given the causal models $\mathcal{M}^0$ and $\mathcal{M}^e_{\phi(e)}$, we define the marginal distributions of $\mathbb{P}^0, \mathbb{P}^e$ over $\Xb$, denoted $\mathbb{P}^0_X, \mathbb{P}^e_X$. These factorize in DAGs $\bar{G}^0_X=([p], \bar{E}^0_X)$ and $\bar{G}^e_X=([p], \bar{E}^e_X)$ where the set of nodes $[p]$ represents covariate nodes. The edges in $\bar{E}^0_X, \bar{E}^e_X$ are the same as in $\bar{E}^0, \bar{E}^e$ for those that are not connecting $Y$ to another node. Previous edges in $(\bar{E}^0, \bar{E}^e)$ that connected $Y$ to other nodes are replaced by edges in $\bar{E}^0_X, \bar{E}^e_X$ between covariates that connect in $\bar{G}^0, \bar{G}^e$ every parent of $Y$ to every child of $Y$. Finally, graphs $\bar{G}^e_X$ share the same previously mentioned two properties as $\bar{G}^e$: they all share a same topological ordering given by $\bar{G}^0_X$ and nodes of randomized covariates $\phi(e)$ in environment $e$ have no incoming edges in $\Bar{G}^e_X$.

Assume without loss of generality that the order $[p] = \{1,\dots, p\}$ is a topological order of $\bar{G}^0_X$. We consider the reversed lexicographical order between two subsets $A, B \subset [p]$, $A = \{a_1, \dots, a_{n_{A}}\}$ and $B = \{b_1, \dots b_{n_B}\}$, where $a_1 > a_2 > \dots > a_{n_A}$ (and analogously for $B$), as the total order given by
\begin{equation}
    A \succ B \Longleftrightarrow \big(\exists i | (a_j = b_j \forall j<i) \;\text{and}\; (a_i>b_i )\big) \;\text{or}\; (n_A > n_B \;\text{and}\; \{a_1, \dots,a_{n_B}\} = B)
\end{equation}
For a set of environments $\mathcal{E}$ we define the set $\phi(\mathcal{E}):=\{\phi(e) \subset [p], e\in \mathcal{E}\}$ of subsets $\phi(e)$ of indices in $[p]$ that index the variables intervened in environment $e \in \mathcal{E}$. Such set has a maximal element $\max\big(\phi(\mathcal{E})\big)$. For two sets of environments $\mathcal{E}_1, \mathcal{E}_2$, we define an order $\mathcal{E}_1 \succeq \mathcal{E}_2$ by comparing their maximal elements: $\max\big(\phi(\mathcal{E}_1)\big) \succeq \max\big(\phi(\mathcal{E}_2)\big)$. We now prove Proposition~\ref{proposition:identifiability}.

\begin{proof}

\paragraph{Uniqueness}
We first prove the uniqueness statement. Assume that there are two functions $\bar{f}_1$ and $\bar{f}_2$ that follow the decomposition given in equation~\eqref{decomposition_function} such that the orthogonality conditions given by equation~\eqref{orthogonality_intervention} hold and consider the difference $f := \bar{f}_1 - \bar{f}_2$. By linearity $f$ follows the decomposition given in equation~\eqref{decomposition_function}. Additionally, for all $e\in \mathcal{E}$, for all square-integrable $h:\mathbb{R}^{|\phi(e)|}\longrightarrow \mathbb{R}$, by subtracting the two orthogonality constraints for $\bar{f}_1$ and $\bar{f}_2$ given by equation~\eqref{orthogonality_intervention}, we have
\begin{equation}\label{recursion_hyp}
    0 = \mathbb{E}^e[h(\Xb_{\phi(e)})f(\Xb)]
\end{equation}
Applying Proposition~\ref{prop:orthogonality-identifiability} from the Appendix we get that $f = 0$ over the support of the random variables which by assumption is the same across environments, hence the uniqueness of $\bar{f}$. 

\paragraph{Identifiability}
We finally show that $f^0$ satisfies equation~\eqref{orthogonality_intervention} whenever $f^0 \in \mathcal{F}_{\mathcal{E}}$. We have that the structural equation defining the response variable is given by $Y = f^0(\Xb_{S^0}) + \epsilon_Y$. Then, for all $e\in \mathcal{E}$, we have that the residual $\epsilon_Y$ in environment $e$ is independent of the intervened variables $\Xb_{\phi(e)}$ under $\mathbb{P}^e$. Therefore, given that the residuals are centered random variables, we get that 
\begin{equation}
    \mathbb{E}^e[h(\Xb_{\phi(e)})(Y - f^0(\Xb_{S^0}))] = 0
\end{equation}
This implies that $f^0 = \bar{f}$ whenever $f^0 \in \mathcal{F}_{\mathcal{E}}$. 
\end{proof}

\begin{proposition}\label{prop:orthogonality-identifiability}
Assume the model of environments defined in \eqref{equation:perturbed-nonlinear-SEM}, let $f \in \mathcal{F}_{\mathcal{E}}$, i.e. it can be decomposed as follows:
\begin{equation}\label{f-decomposition}
    f(\xb) = \sum_{e \in \mathcal{E}} f_{e}(\xb_{\phi(e)})
\end{equation}
for some functions square-integrable $(f_{e})_{e \in \mathcal{E}}$ defined over the subsets of covariates indexed by $\phi(e)$. Assume that all $\mathbb{P}^e$ have the same support. If for all $e\in \mathcal{E}$, for all $h:\mathbb{R}^{|\phi(e)|}\rightarrow \mathbb{R}$ square-integrable we have
\begin{equation}\label{f-orthogonality}
    \mathbb{E}^e[h(\Xb_{\phi(e)})f(\Xb)] = 0
\end{equation}
then $f=0$.
\end{proposition}

\begin{proof}
We prove the result by recursion: we show that, if $f$ satisfies the conditions of the lemma for $\mathcal{E}$, then there exists another set of environments $\tilde{\mathcal{E}}$, such that $\mathcal{E} \succ \tilde{\mathcal{E}}$ (strictly), and $f$ satisfies conditions \eqref{f-decomposition} and \eqref{f-orthogonality} with $\tilde{\mathcal{E}}$. We recursively show that $f$ must satisfy a decomposition of the type \eqref{f-decomposition} with increasingly fewer variables and interactions. Given that the set of sets of subsets of $[p]$ is finite, after a finite number of steps we get that $f$ must satisfy the decomposition \eqref{f-decomposition} for $\mathcal{E} = \emptyset$ (i.e. $f$ is constant). We then conclude that $f=0$ using condition \eqref{f-orthogonality} with $h=1$.

Consider the maximal element $\phi_0:= \max\big(\phi(\mathcal{E})\big)$, uniquely attained at $e_0$, i.e. $\phi_0 = \phi(e_0)$. Let $h:\mathbb{R}^{|\phi_0|}\longrightarrow \mathbb{R}$ a square-integrable function under $\mathbb{P}^{e_0}$ such that for all $i \in \phi_0$,
\begin{equation*}
    \mathbb{E}^{e_0}[h(\Xb_{\phi_0})|\Xb_{\phi_0 \setminus \{i\}}] = 0
\end{equation*}
under the distribution $\mathbb{P}^{e_0}$ from environment $e_0$. By assumption \eqref{f-orthogonality}, and the decomposition \eqref{f-decomposition}, we have that 
\begin{align*}
    0 = & \mathbb{E}^{e_0}[h(\Xb_{\phi_0})f(\Xb)]
    \\ = & \mathbb{E}^{e_0}[h(\Xb_{\phi_0})\sum_{e \in \mathcal{E}}f_e(\Xb_{\phi(e)})]
    \\ = & \sum_{e: \phi_0\succ \phi(e)}\mathbb{E}^{e_0}[h(\Xb_{\phi_0})f_e(\Xb_{\phi(e)})] + \mathbb{E}^{e_0}[h(\Xb_{\phi_0})f_{e_0}(\Xb_{\phi_0})]
\end{align*}
Now for all $e \neq e_0$, we have $\phi_0 \succ \phi(e)$. Let $\phi_0 =: \{a_1,\dots, a_{n_A}\}$, $\phi(e) =: \{b_1, \dots, b_{n_B}\}$. We define a subset of indices $\Delta_e$ as follows. If there exists $i$ such that $a_j = b_j$ for all $j < i$, and $a_i>b_i$, then let $\Delta_e := \{a_1,\dots,a_{i-1}\}\cup [a_{i}-1]$. Otherwise, if $n_A > n_B$ and $\{a_1, \dots,a_{n_B}\} = \phi(e)$, then $\Delta_e := \{a_1,\dots,a_{n_{B}}\}\cup [a_{n_B +1}-1]$ (so it can be written as in the first case with $i = n_{B}+1$). In both cases, $\phi(e) \subset \Delta_e$, $\phi_0\setminus \{a_i\} \subset \Delta_e$ and $a_i \notin \Delta_e$. Then
\begin{align*}
    \mathbb{E}^{e_0}[h(\Xb_{\phi_0})f_e(\Xb_{\phi(e)})] = & \mathbb{E}^{e_0}\big[\mathbb{E}^{e_0}[h(\Xb_{\phi_0})f_e(\Xb_{\phi(e)})|\Xb_{\Delta_e}]\big]
    \\ = & \mathbb{E}^{e_0}\big[\mathbb{E}^{e_0}[h(\Xb_{\phi_0})|\Xb_{\Delta_e}]f_e(\Xb_{\phi(e)})\big]
    \\ \stackrel{(*)}{=} & \mathbb{E}^{e_0}\big[\mathbb{E}^{e_0}[h(\Xb_{\phi_0})|\Xb_{\phi_0\setminus \{a_i\}}]f_e(\Xb_{\phi(e)})\big]
    \\ = & 0
\end{align*}
where we used the fact that by assumption on $h$, we have $\mathbb{E}^{e_0}[h(\Xb_{\phi_0})|\Xb_{\phi_0\setminus \{a_i\}}] = 0$, and we later prove $(*)$. We then get that
\begin{equation*}
    0 = \mathbb{E}^{e_0}[h(\Xb_{\phi_0})f_{e_0}(\Xb_{\phi_0})] 
\end{equation*}
We now use Lemma~\ref{lemma:decomposition-function}, given that by assumption $\Xb_{\phi_0}$ are randomized, independently one of another under $\mathbb{P}^{e_0}$. We have that the following holds over the support of the variables indexed by $\phi_0$ in environment $e_0$, where $g_i :\mathbb{R}^{|\phi_0|-1}\rightarrow \mathbb{R}$ are square-integrable functions:
\begin{equation*}
    f_{e_0}(\xb_{\phi_0}) = \sum_{i \in \phi_0} g_{i}(\xb_{\phi_0 \setminus \{i\}})
\end{equation*}
We get that $f$ follows the following decomposition as in equation \eqref{f-decomposition} over the support of $\mathbb{P}^{e_0}$, which by assumption is the same as that of the other environment distributions. We write:
\begin{align*}
    f(\xb) = & \sum_{e \in \mathcal{E}\setminus \{\phi_0\}} f_e(\xb_{\phi(e)}) + \sum_{i \in \phi_0} g_{i}(\xb_{\phi_0 \setminus \{i\}})
    \\ = & \sum_{e \in \tilde{\mathcal{E}}} \tilde{f}_e(\xb_{\phi(e)})
\end{align*}
where we defined a new set of environments $\tilde{\mathcal{E}}$. The maximal element in $\tilde{\mathcal{E}}$ is smaller than $\phi_0$, as for all $i\in \phi_0$, $\phi_0 \succ \phi_0 \setminus \{i\}$, therefore $\mathcal{E} \succ \tilde{\mathcal{E}}$. Now for any $i\in \phi_0$, we can use $e_0$ as the environment where variables in $\phi_0 \setminus\{i\}$ are perturbed, so that condition \eqref{f-orthogonality} still holds. We thus get that $f$ satisfies the same assumptions for the new $\tilde{\mathcal{E}}$. This concludes the recursion.

The remaining statement to prove is the equality $(*)$ above, which is a consequence of the following identity:
\begin{equation*}
    \mathbb{E}^{e_0}[h(\Xb_{\phi_0})|\Xb_{\Delta_e}] = \mathbb{E}^{e_0}[h(\Xb_{\phi_0})|\Xb_{\phi_0\setminus \{a_i\}}]
\end{equation*}
We now prove this identity holds, for which it suffices to show the following conditional independence statement holds under $\mathbb{P}^{e_0}_X$:
\begin{align*}
    & \Xb_{\phi_0} \independent \Xb_{\Delta_e \setminus \big\{\phi_0 \setminus\{a_i\}\big\}}|\Xb_{\phi_0\setminus\{a_i\}}
\end{align*}
which is equivalent to
\begin{align}\label{equation:CIproof}
X_{a_i}\independent \Xb_{\Delta_e \setminus \big\{\phi_0\setminus\{a_i\}\big\}}|\Xb_{\phi_0\setminus\{a_i\}}
\end{align}
To prove the last conditional independence statement, we rely on the equivalence between d-separation and separation in the moral ancestral graph \citep{pearl2013identifying}. We use here our assumptions on how the model $\mathcal{M}^{e_0}$ is generated. We know that $\mathbb{P}^{e_0}_X$ factorizes in the extended graph $\bar{G}^{e_0}_X$. We need to show that $X_{a_i}$ and $\Xb_{\Delta_e \setminus \big\{\phi_0\setminus\{a_i\}\big\}}$ are separated by $\Xb_{\phi_0\setminus \{a_i\}}$ in the moral ancestral graph of variables $\Xb_{\Delta_e \cup \{a_i\}}$ (union of all variables in the conditional independence statement) with respect to $\bar{G}^{e_0}_X$. Variables indexed by $\phi_0$ are randomized, hence do not have any ancestors in $\bar{G}^{e_0}_X$. Therefore the ancestral graph of $\Xb_{\Delta_e \cup \{a_i\}}$ does not contain any additional nodes. Also, variables $\Xb_{\{a_1,\dots,a_i\}}$ do not have descendants in the ancestral graph: by topological ordering, variables $\Xb_{[a_i -1]}$ can not be descendants of $\Xb_{\{a_1,\dots,a_i\}}$. And any variable within $\Xb_{\{a_1,\dots,a_i\}}$ can not be descendant of any other node as they do not have ancestors. Therefore, variables $\Xb_{\{a_1,\dots,a_i\}}$ are isolated in the ancestral graph, and moralizing the ancestral graph does not connect these nodes to any other node. Therefore, in particular we get the separation statement between $X_{a_i}$ and $\Xb_{\Delta_e \setminus \big\{\phi_0\setminus\{a_i\}\big\}}$ by $\Xb_{\phi_0\setminus\{a_i\}}$.

Importantly, the assumptions on the distributions are needed only to show that the conditional independence statement~\eqref{equation:CIproof} holds, as this is only a property of the extended graph $\bar{G}^{e_0}_X$. The assumptions on $\mathcal{M}^{e_0}$ impose constraints on $\bar{G}^{e_0}_X$: the choice of the structural equation functions defining the covariates as well as the joint distribution of the non-randomized disturbance variables $\{\epsilon_j\}_{j\notin \phi_{0}}$ is irrelevant for the validity of the proof.
\end{proof}

\begin{lemma}\label{lemma:decomposition-function}
Consider a square-integrable function $g:\mathbb{R}^p \longrightarrow \mathbb{R}$. Let $\Xb = (X_1, \dots, X_p)$ be independent random variables. Assume that, for any square integrable $h:\mathbb{R}^p \longrightarrow \mathbb{R}$ satisfying
\begin{equation*}
    \mathbb{E}[h(\Xb)|\Xb_{-i}] = 0 \quad \forall i \in [p]
\end{equation*}
where $\Xb_{-i} := \{X_j, j\neq i\}$, the following holds:
\begin{equation}
    \mathbb{E}[h(\Xb)g(\Xb)] = 0
\end{equation}
Then there exist functions $g_i:\mathbb{R}^{p-1}\longrightarrow \mathbb{R}$ for $i\in [p]$ such that
\begin{equation*}
    g(\xb) = \sum_{i=1}^p g_{i}(\xb_{-i})
\end{equation*}
over the support of the random variables $\Xb$.
\end{lemma}

\begin{proof}
Denote $Y:= g(\Xb)$. Define 
\begin{equation*}
    h(\Xb) := \sum_{I \subset [p]} (-1)^{p-|I|}\mathbb{E}[Y | \Xb_I]
\end{equation*}
We have that, for any $i \in [p]$, 
\begin{align*}
    \mathbb{E}[h(\Xb)|\Xb_{-i}] = & \sum_{I \subset [p]} (-1)^{p-|I|}\mathbb{E}[\mathbb{E}[Y | \Xb_I]|\Xb_{-i}]
    \\ = & \sum_{\substack{I \subset [p]\\ i\in I}} (-1)^{p-|I|}\mathbb{E}[\mathbb{E}[Y | \Xb_I]|\Xb_{-i}] + \sum_{\substack{I \subset [p]\\ i \notin I}} (-1)^{p-|I|}\mathbb{E}[\mathbb{E}[Y | \Xb_I]|\Xb_{-i}]
    \\ = & \sum_{\substack{I \subset [p]\\ i\in I}} (-1)^{p-|I|}\mathbb{E}[Y | \Xb_{I\setminus \{i\}}] + \sum_{\substack{I \subset [p]\\ i \notin I}} (-1)^{p-|I|}\mathbb{E}[Y | \Xb_I]
    \\ = & \sum_{J \subset [p]\setminus\{i\}} (-1)^{p-|J|-1}\mathbb{E}[Y | \Xb_{J}] + \sum_{\substack{I \subset [p]\\ i \notin I}} (-1)^{p-|I|}\mathbb{E}[Y | \Xb_I]
    \\ = & \; 0
\end{align*}
By assumption we thus get that $\mathbb{E}[h(\Xb)g(\Xb)] = 0$. We can decompose $h(\Xb)$ as follows:
\begin{align*}
    h(\Xb) = & \; g(\Xb) + \sum_{\substack{I \subset [p]\\ |I|\leq p-1}} (-1)^{p-|I|}\mathbb{E}[Y | \Xb_I ]
\end{align*}
Therefore we get:
\begin{align*}
    \mathbb{E}[h(\Xb)^2] = & \; \mathbb{E}[h(\Xb)(g(\Xb) - \sum_{\substack{I \subset [p]\\ |I| \leq p-1}} (-1)^{p-|I|}\mathbb{E}[Y | \Xb_I])]
    \\ = & \; -\mathbb{E}[h(\Xb)\sum_{\substack{I \subset [p]\\ |I| \leq p-1}} (-1)^{p-|I|}\mathbb{E}[Y | \Xb_I]]
    \\ = & \; -\sum_{\substack{I \subset [p]\\ |I| \leq p-1}}(-1)^{p-|I|}\mathbb{E}[h(\Xb) \mathbb{E}[Y | \Xb_I]]
    \\ = & \; -\sum_{\substack{I \subset [p]\\ |I| \leq p-1}}(-1)^{p-|I|}\mathbb{E}[\mathbb{E}[h(\Xb)| \Xb_I] \mathbb{E}[Y | \Xb_I]]
    \\ = & \; 0
\end{align*}
where we used the fact that $\mathbb{E}[h(\Xb)| \Xb_I] = \mathbb{E}[\mathbb{E}[h(\Xb)| \Xb_{-i}]| \Xb_I]$ for some $i\notin I$ as $|I|\leq p-1$. Therefore $h(\Xb) = 0$, and therefore
\begin{equation*}
    g(\xb) = \sum_{\substack{I \subset [p]\\ |I| \leq p-1}} (-1)^{p-|I|-1}\mathbb{E}[Y | \Xb_I = \xb_I] =  \sum_{i=1}^p g_{i}(\xb_{-i})
\end{equation*}
over the support of $\Xb$ for some choice of functions $g_i:\mathbb{R}^{p-1}\longrightarrow \mathbb{R}$.
\end{proof}

\end{document}